\documentclass{lmcs} %%% last changed 2014-08-20
\pdfoutput=1

\usepackage{lastpage}
\lmcsdoi{18}{4}{11}
\lmcsheading{}{\pageref{LastPage}}{}{}%
{Sep.~08,~2021}{Dec.~12,~2022}{}


\usepackage{hyperref}
\usepackage{todonotes}
\usepackage{tcolorbox}
\usepackage{xspace}
\usepackage{textcomp}

\theoremstyle{plain} %\crefname{satz}{Satz}{S\"atze}
\usepackage{xstring}
\usepackage{mathbbol}
\usepackage{amssymb}             % AMS Math
\usepackage{scalerel,amssymb,graphicx}

\definecolor{mycol1}{RGB}{245, 121, 58}
\definecolor{mycol2}{RGB}{169, 90, 161}
\definecolor{mycol3}{RGB}{15, 32, 128}

\newcommand{\mycolorTrzy}{mycol3}
\newcommand{\mycolorDwa}{mycol2}
\newcommand{\mycolorOne}{mycol1}

\newcommand{\red}[1]{\textcolor{\mycolorTrzy}{#1}}
\newcommand{\green}[1]{\textcolor{\mycolorDwa}{#1}}

\newcommand{\greenWork}[1]{{#1}}

\newcommand{\redWork}[1]{{#1}}
\newcommand{\blue}[1]{\textcolor{mycol1}{#1}}
\newcommand{\blueWork}[1]{{#1}}

\newcommand{\brownWork}[1]{{#1}}

\newcommand{\decproblem}[4]{
\vspace{2mm}
\begin{tcolorbox}
\hspace{-0.5cm}\begin{tabular}{r|c}
{\sc #1} &\\
      {{\sc input}:}
& \begin{minipage}[t]{#4\textwidth} #2.
   \end{minipage}
   \\[5mm]
   {{\sc output}:} & \begin{minipage}[t]{#4\textwidth} #3
   \end{minipage}
\end{tabular}
\end{tcolorbox}
}

\newcommand{\para}[1]{{\bf #1.}}

\newcommand{\binomial}[2]{\genfrac(){0pt}{1}{#1}{#2}}

\newcommand{\PTIME}{{\sc \bf P{\text -}Time}\xspace}
\newcommand{\NPTIME}{{\sc \bf NP}\xspace}
\newcommand{\ETIME}{{\sc \bf Exp{\text -}Time}\xspace}
\newcommand{\NETIME}{{\sc \bf NExp{\text -}Time}\xspace}

\newcommand{\AckermanCOMPLETE}{Ackermann-complete}

\newcommand{\kset}[2]{{[#2]^{#1}}}
\newcommand{\ktuple}[2]{#2^{#1}}
\newcommand{\setof}[2]{\{ #1 \mid #2\}}
\newcommand{\size}[1]{|#1|}
\newcommand{\abs}[1]{\lvert#1\rvert}

\newcommand{\inorm}[1]{\lVert#1\rVert_{\infty}}

\newcommand{\norm}[2]{\lVert#1\rVert_{#2}}

\newcommand{\myeq}[2]{\mathrel{\stackrel{\makebox[0pt]{\mbox{\normalfont\tiny 
{#2}}}}{#1}}}
\newcommand{\eqdef}{\myeq{=}{def}}

\newcommand{\set}[1]{\{#1\}}

\newcommand{\setSizeK}{\brownWork{k}}
\newcommand{\dimension}{{\brownWork{d}}}
\newcommand{\Numbers}{\mathbb{K}}
\newcommand{\N}{{\mathbb N}}
\newcommand{\Z}{{\mathbb Z}}
\newcommand{\Q}{{\mathbb Q}}

\newcommand{\ALG}{ALG}

\newcommand{\Auth}{\Pi}

\newcommand{\reductionMatrix}[3]{\overline{
        \lceil#1,#2,#3\rceil}}
\newcommand{\simple}{simple}

\newcommand{\groupG}{\Gamma}

\newcommand{\Proj}{{\mathfrak{P}}}
\newcommand{\support}[1]{{\text{\sc support}({#1})}}
\newcommand{\smooth}[1]{smooth(#1)}

\newcommand{\perm}[1]{\text{\sc Perm}(#1)}
\newcommand{\eqs}[1]{\text{\sc Eq}(#1)}

\newcommand{\cut}[2]{{#1}_{|-#2}} %% nie za bardzo mi się podoba takie oznaczenie :/
\newcommand{\revcut}[2]{{#1}_{|+#2}}

\newcommand{\setLdwa}{\breve{\setL}}

\newcommand{\da}{\blueWork{\alpha}}
\newcommand{\db}{\blueWork{\beta}}
\newcommand{\dc}{\blueWork{\gamma}}
\newcommand{\dd}{\blueWork{\delta}}
\newcommand{\de}{\blueWork{\varepsilon}}

\renewcommand{\vec}[1]{\redWork{\mathbf #1}}
\newcommand{\target}{\vec v}

\newcommand{\sums}[2]{#1\text{-\sc Sums}(#2)}
\newcommand{\sumsBez}[1]{#1\text{-\sc Sums}}

\newcommand{\EStikz}[5]{
\begin{tikzpicture}
[place/.style={circle,draw=mycol1!50,fill=mycol1!20,thick,inner sep=0pt,minimum size=4mm}]
\node[place]      (A2)            {$#2$};
\node[place]      (B2)       [below =of A2]  {$#3$};
\node[place]      (C2)      [right =of A2] {$#4$};
\node[place]      (D2)      [below =of C2] {$#5$};

\draw (A2) to node [left] {$-#1$} (B2) to node [above] {$#1$} (D2) to node [right] {$-#1$} (C2) 
to node [above] {$#1$} (A2);
\end{tikzpicture}
}
\newcommand{\VStikz}[4]{
\begin{tikzpicture}
[place/.style={circle,draw=mycol1!50,fill=mycol1!20,thick,inner sep=0pt,minimum size=4mm}]
\node[place]      (A1)         {$#2$};
\node[place]      (B1)       [left=of A1]  {$#3$};
\node[place]      (C1)      [right=of A1] {$#4$};
\draw (B1) to node [above] {$#1$} (A1) to node [above] {$-#1$} (C1) ;
\end{tikzpicture}
}

\newcommand{\setInput}{\mathcal{I}}

\newcommand{\setA}{\mathcal A}
\newcommand{\setB}{{\mathcal B}}
\newcommand{\setC}{{\mathcal C}}
\newcommand{\setD}{\mathcal D}

\newcommand{\setS}{\mathcal{S}}

\newcommand{\setX}{\mathcal{X}}
\newcommand{\setY}{\mathcal{Y}}

\newcommand{\systemOfEquationsU}{\mathcal{U}}
\newcommand{\setU}{\mathcal{U}}

\newcommand{\setP}{\mathcal{P}}
\newcommand{\setM}{\mathcal{M}}

\newcommand{\setL}{\mathcal{L}}

\newcommand{\familyL}{\mathcal{F}}

\newcommand{\setG}{\mathcal{G}}
\newcommand{\familyF}{\mathcal{F}}

\newcommand{\setHypergraphsH}{\blueWork{\mathcal{H}}}

\newcommand{\he}[1]{\text{\sc Edges}(#1)}

\newcommand{\setVertices}{{\redWork{V}}}

\newcommand{\setVerticesW}{\redWork{W}}

\newcommand{\placeHolder}{{\Box}}
\newcommand{\setWeight}[2]{%
    \IfEqCase{#1}{%
        {2}{{\blueWork{\mathcal E_{#2}}}}%
        {1}{{\blueWork{\mathcal V_{#2}}}}%
    }[\PackageError{tree}{Undefined option to tree: #1}{}]%
}%

\newcommand{\setSimpleGraph}[2]{%
    \IfEqCase{#1}{%
        {1}{{\redWork{ \mathcal{S}_{#2}^{\bullet} }}}%
        {2}{{\redWork{ \mathcal{S}_{#2}^{\bullet - \bullet} }}}%
    }[\PackageError{tree}{Undefined option to tree: #1}{}]%
}%

\newcommand{\setSupport}{\greenWork{\setVertices}}

\newcommand{\weight}[1]{%
    \IfEqCase{#1}{%
	{non}{{\redWork{\Lambda}}}%
 	{emptyset}{{{\redWork{\Lambda}}_{\emptyset}}}%
	{e-X}{{{\redWork{\Lambda}}_{\edge\setminus \setX}}}%
	{setA}{{{\redWork{\Lambda}}_{\setA}}}%
	{setA'}{{{\redWork{\Lambda}}_{\setA'}}}%
	{setB}{{{\redWork{\Lambda}}_{\setB}}}%
	{setC}{{{\redWork{\Lambda}}_{\setC}}}%
	{setB_1}{{{\redWork{\Lambda}}_{\setB_1}}}%
	{setB_l}{{{\redWork{\Lambda}}_{\setB_l}}}%
	{x}{{\redWork{\Lambda}}_{#1}}%
	 {YcupX}{{\redWork{\Lambda}}_{\setY\cup \setX}}%
	{setX}{{\redWork{\Lambda}}_{\setX}}%
	{setX'}{{\redWork{\Lambda}}_{\setX'}}%
	{setX'_1}{{\redWork{\Lambda}}_{\setX'_1}}%
	{setX'_2}{{\redWork{\Lambda}}_{\setX'_2}}%
	{setX'_j}{{\redWork{\Lambda}}_{\setX'_j}}%
	{setX'_n}{{\redWork{\Lambda}}_{\setX'_n}}%
	{setX-da}{{\redWork{\Lambda}}_{\setX\setminus\{\da\}}}%
	{setX-da'}{{\redWork{\Lambda}}_{\setX\setminus\{\da'\}}}%
	{{setX}-{da,da'}}{{\redWork{\Lambda}}_{\setX\setminus\{\da,\da'\}}}%
	{setY}{{\redWork{\Lambda}}_{\setY}}%
	{setY_j}{{\redWork{\Lambda}}_{\setY_j}}%
	{setY_1}{{\redWork{\Lambda}}_{\setY_1}}%
	{setY_2}{{\redWork{\Lambda}}_{\setY_2}}%
	{setY_n}{{\redWork{\Lambda}}_{\setY_n}}%
	{setY_l}{{\redWork{\Lambda}}_{\setY_l}}%
	{setL}{{\redWork{\Lambda}}_{{\setL}}}%
	{L-Y}{{\redWork{\Lambda}}_{\setL\setminus \setY_{\setL}}}%
	{N+Y}{{\redWork{\Lambda}}_{\setL' \cup \setY_{\setL}}}%
	{N+Y-Y}{{\redWork{\Lambda}}_{\setL' \cup \setY_{\setL}\setminus \setY_{\setL}}}%
	{setLdwa}{{{\redWork{\Lambda}}_{\redWork{\breve{\setL}}}}}
    {setLdwa-Y}{{{\redWork{\Lambda}}_{\redWork{\breve{\setL}\setminus \setY_{\setL}}}}}
	{dc dd}{{\redWork{\Lambda}}_{\{\dc,\dd\}}}%
	{{ x, y}}{{\redWork{\Lambda}}_{\{x,y\}}}%
	{hypergraphH}{{{\redWork{\Lambda}}_{\hypergraphH}}}%
	{hypergraphG}{{{\redWork{\Lambda}}_{\hypergraphG}}}%
	{revcut{F'}{v}}{{{\redWork{\Lambda}}_{\revcut{F'}{v}}}}%
	{revcut{F'}{da}}{{{\redWork{\Lambda}}_{\revcut{F'}{\da}}}}%
	{revcut{F'}{v'}}{{{\redWork{\Lambda}}_{\revcut{F'}{v'}}}}%
	{revcut{F}{v}}{{{\redWork{\Lambda}}_{\revcut{F}{v}}}}%
	{revcut{simpleGraph{m}{a}}{da}}{{{\redWork{\Lambda}}_{\revcut{\simpleGraph{m}{\vec{a}}}{\da}}}}%
	{revcut{simpleGraph{m}{a}}{da'}}{{{\redWork{\Lambda}}_{\revcut{\simpleGraph{m}{\vec{a}}}{\da'}}}}%
	{simpleGraph{m+1}{a}}{{{\redWork{\Lambda}}_{\simpleGraph{m+1}{\vec{a}}}}}%
	{simpleGraph{m}{a}}{{{\redWork{\Lambda}}_{\simpleGraph{m}{\vec{a}}}}}%
	{Zcup{v}}{{{\redWork{\Lambda}}_{Z\cup\{v\}}}}%
	{Zcup{da}}{{{\redWork{\Lambda}}_{Z\cup\{\da\}}}}%
	{Ycup{da}}{{{\redWork{\Lambda}}_{\setY\cup\{\da\}}}}%
	{da}{{\redWork{\Lambda}}_{\da}}%
	{db}{{\redWork{\Lambda}}_{\db}}%
	{db_i}{{\redWork{\Lambda}}_{\db_i}}%{revcut{simpleGraph{m}{a}}{da'}}
	{dc_i}{{\redWork{\Lambda}}_{\dc_i}}%
	{dc}{{\redWork{\Lambda}}_{\dc}}%
	{dd}{{\redWork{\Lambda}}_{\dd}}%
	{de}{{\redWork{\Lambda}}_{\de}}%
	{da'}{{\redWork{\Lambda}}_{\da'}}%
	{db'}{{\redWork{\Lambda}}_{\db'}}%
	{dc'}{{\redWork{\Lambda}}_{\dc'}}%
	{dd'}{{\redWork{\Lambda}}_{\dd'}}%
	{set{da}}{\redWork{\Lambda}_{\{\da\}}}
	{{da,db}}{{\redWork{\Lambda}}_{\{\da,\db\}}}%
	{{da,dc}}{{\redWork{\Lambda}}_{\{\da,\dc\}}}%
	{{da,dd}}{{\redWork{\Lambda}}_{\{\da,\dd\}}}%
	{{db,dc}}{{\redWork{\Lambda}}_{\{\db,\dc\}}}%
	{{db,dd}}{{\redWork{\Lambda}}_{\{\db,\dd\}}}%
	{{dc,dd}}{{\redWork{\Lambda}}_{\{\dc,\dd\}}}%
	{edge}{{{\redWork{\Lambda}}_{\edge}}}%
	{edge'}{{{\redWork{\Lambda}}_{\edge'}}}%
	{edgeL}{{{\redWork{\Lambda}}_{\redWork{\setL}}}}%
	{edgeK}{{{\redWork{\Lambda}}_{\redWork{\widetilde{\setL}}}}}%
    {edgeKe}{{{\redWork{\Lambda}}_{\redWork{\widetilde{\setL}}'}}}%
	{red{bullet}}{{\redWork{\Lambda}}_{\red{\bullet}}}%
	{blue{bullet}}{{\redWork{\Lambda}}_{\blue{\bullet}}}%
	{green{bullet}}{{\redWork{\Lambda}}_{\green{\bullet}}}%
	{{red{bullet},green{bullet}}}{{\redWork{\Lambda}}_{\{\red{\bullet},\green{\bullet}\}}}
	{{red{bullet},blue{bullet}}}{{\redWork{\Lambda}}_{\{\red{\bullet},\blue{\bullet}\}}}
	{{blue{\bullet},green{bullet}}}{{\redWork{\Lambda}}_{\{\blue{\bullet},\green{\bullet}\}}}
       }[\PackageError{tree}{Undefined option to tree: #1}{}]%
}%

\newcommand{\hypergraphH}{\redWork{\mathbb H}}
\newcommand{\hypergraphG}{\redWork{\mathbb G}}
\newcommand{\hypergraphF}{\redWork{\mathbb F}}
\newcommand{\hypergraphK}{\redWork{\mathbb O}}
\newcommand{\hypergraphS}{\redWork{\mathbb S}}

\newcommand{\simpleGraph}[2]{%
    \IfEqCase{#1}{%
	{0}{{\redWork{\mathbb S_{#2}^{0} }}}%
        {1}{{\redWork{\mathbb S_{#2}^{\bullet} }}}%
        {2}{{\redWork{\mathbb S_{#2}^{\bullet - \bullet} }}}%
        {m}{{\redWork{\mathbb S_{#2}^{m} }}}
        {m+1}{{\redWork{\mathbb S_{#2}^{m+1} }}}
        {newl+1}{{\redWork{\mathbb S}}}%{{\redWork{\mathbb S_{\underline{\hspace{0.4cm}}}^{l+1} }}}
        {any}{{\redWork{\mathbb S_{}^{*} }}}
        {setX}{{\redWork{\mathbb S_{#2}^{\size{\setX}} }}}
        {setX'}{{\redWork{\mathbb S_{#2}^{\size{\setX'}} }}}
    }[\PackageError{tree}{Undefined option to tree: #1}{}]%
}%

\newcommand{\Vertices}[1]{Vert(#1)}

\newcommand{\newl}{m}

\newcommand{\edge}{\blueWork{e}}

\newcommand{\edgeK}{\widetilde{\setL}}

\newcommand{\swap}[3]{{{\tau_{\{#1,#2\}}(#3)}}}

\newcommand*\changed[3]{#2}
\newcommand*\changedd[3]{#2}

\begin{document}

\newcommand{\PHtodo}[1]{{\textcolor{red}{#1}}}

\title{Linear equations for unordered data vectors\texorpdfstring{ in $\kset{\setSizeK}{\setD}\to{}\Z^{\dimension}$}{}.}
% $\kset{\setSizeK}{\setD}\to{}\Z^{\dimension}$.
%\titlecomment{{\lsuper*}OPTIONAL comment concerning the title, \eg, if a variant or an extended abstract of the paper has appeared elsewhere.}

\author{Piotr Hofman\lmcsorcid{0000-0001-9866-3723}}	%required
\author{Jakub R\'o\.zycki}
\address{Faculty of Mathematics, Informatics, and Mechanics\\
 University of Warsaw\\
 Warszawa, Polska}	%required
\email{piotr.hofman@uw.edu.pl, j.rozycki@student.uw.edu.pl}  %optional
\thanks{This work is supported by Polish NCN grant nr. UMO-2016/21/D/ST6/01368.}	%optional

% \address{Faculty of Mathematics, Informatics, and Mechanics\\
%  University of Warsaw\\
%  Warszawa, Polska}	%optional
% \email{
% 
% }  %optional
%\thanks{This work is supported by Polish NCN grant nr. UMO-2016/21/D/ST6/01368.}	%optional

% \author[C.~Name3]{Carla Name3}	%optional
% \address{address 3}	%optional
% \urladdr{name3@url3\quad\rm{(optionally, a web-page can be specified)}}  %optional
% \thanks{thanks 3, optional.}	%optional

%% etc.

%% required for running head on odd and even pages, use suitable
%% abbreviations in case of long titles and many authors:

%%%%%%%%%%%%%%%%%%%%%%%%%%%%%%%%%%%%%%%%%%%%%%%%%%%%%%%%%%%%%%%%%%%%%%%%%%%

%% the abstract has to PRECEDE the command \maketitle:
%% be sure not to issue the \maketitle command twice!

\begin{abstract}
  \noindent Following a recently considered generalisation of linear equations to unordered-data vectors
and to ordered-data vectors,
we perform a further generalisation to \changed{k-element-sets-of-unordered-data vectors}{data vectors that are functions from k-element subsets of the 
unordered-data set to vectors of integer numbers}{11}.
These generalised equations naturally appear in the analysis of vector addition systems 
(or Petri nets) extended so that each token carries a set of unordered data.
We show that nonnegative-integer solvability of linear equations is in \changed{nondeterministic-exponential-time}{nondeterministic exponential time}{11}
while integer solvability is in \changed{polynomial-time}{polynomial time}{11}.
\end{abstract}

\maketitle

%% start the paper here:

\section{Introduction.}

\para{Diophantine linear equations}
The solvability problem for systems of linear Diophantine equations is defined as follows:
given a finite input set of $\dimension$-dimensional integer vectors $\setInput = \{\vec v_1\ldots \vec v_m\} \subseteq \ktuple{\dimension}{\Z}$, and
a target vector $\vec v \in \ktuple{\dimension}{\Z}$, we ask if there is a solution $(n_1\ldots n_m)$
such that
\begin{align}\label{eq:Dio}
n_1 \cdot \vec v_1 + \ldots + n_m \cdot \vec v_m = \vec v.
\end{align}
Restricting solutions $n_1\ldots n_m$ \changed{to integers $\Z$ or nonnegative integers $\N$}{to the set $\Z$ of integers or the set $\N$ of nonnegative 
integers}{18}, we may speak of
$\Z$-solvability \changed{or}{and}{19} $\N$-solvability, respectively.
The former problem is in \PTIME,
while the latter is equivalent to \emph{integer linear programming}, a well-known 
\changed{$\NPTIME-complete$}{\NPTIME-complete}{NULL} problem~\cite{Karp21NP-complete-Problems}.

This paper is a continuation of a line of research that investigates generalisations
of the solvability problem to data vectors~\cite{DBLP:conf/lics/HofmanLT17,DBLP:conf/concur/HofmanL18}\changed{, i.e.,}{\ i.e.}{22} 
on a high level of abstraction, to vectors indexed by orbit-finite sets instead of finite ones
(for an introduction to orbit-finite sets, also known as \emph{sets with atoms}, 
see~\cite{DBLP:conf/lics/BojanczykKL11, DBLP:conf/lics/BojanczykKLT13}).
In the simplest setting, given a fixed countable infinite set $\setD$ of data values,
a \emph{data vector} is a function $\vec{a}\colon \setD \to \ktuple{\dimension}{\Z}$. Addition and scalar multiplication are defined pointwise i.e. 
\changed{$\vec{a}{\bf +}\vec{a'}(\da)=\vec{a}(\da)+\vec{a'}(\da)$}
{$(\vec{a}{\bf +}\vec{a'})(\da)=\vec{a}(\da)+\vec{a'}(\da)$}{26} for every $\da\in \setD$.
\changed{In the presence of data, the solvability problems are defined analogously}{The solvability problems over input sets $\setInput$ of data vectors are 
defined analogously}{27}, with the important difference
that the input set $\setInput$ of data vectors \changed{\emph{is orbit-finite}\changed{, i.e.,}{\ i.e.}{22}}{}{27} 
is the closure, under data permutations\changed{/data renaming}{}{35}, of a finite set of data vectors. 
Given such an \changed{orbit-finite}{infinite, but finite up to data permutation}{27} set $\setInput$, and a target data vector $\vec v$, we ask if there are 
data vectors $\vec v_1\ldots \vec v_m\in \setInput$ and numbers $n_1\ldots n_m$, such that
the equality~\eqref{eq:Dio} is satisfied.

\begin{exa}
We consider data vectors $\setD \xrightarrow{} \Z$. For every two different data values $\dd, \de\in \setD$, let
$\vec v_{\dd \de}(x) = 1$ if $x=\dd$ or $x=\de$, and $\vec v_{\dd \de}(x) = 0$ otherwise.
Let $\setInput = \setof{\vec v_{\dd \de}\colon \setD \to \Z}{\dd,\de\in\setD,\dd\neq \de}$. The set $\setInput$ is closed under data permutation. Indeed, for 
any data permutation 
(bijection) 
$\pi\colon\setD\xrightarrow{} \setD$ and any $\vec{v_{\dd \de}}\in \setInput$ we have $\vec{v_{\dd \de}}\circ \pi = \vec{v_{\pi^{-1}(\dd) \pi^{-1}(\de)}}$ 
which is 
also an element of 
$\setInput$.

Finally, let \changed{$\vec v(\db) = 2$}{$\vec{v_{\db}}(\db) = 2$}{70} and \changed{$\vec v(x) = 0$}{$\vec{v_{\db}}(x) = 0$}{70} for $x\neq \db$, where 
$\db\in\setD$ is some fixed data
value.
On the one hand\changed{ side}{}{44}, this instance admits a $\Z$-solution, \changed{since $\vec a$}{since $\vec{v_{\db}}$}{44} is presentable as
\changed{
\begin{align*}
\vec v = \vec v_{\db \dd} + \vec v_{\db \de} - \vec v_{\dd \de},
\end{align*}
}
{
\begin{align*}
\vec{v_{\db}} = \vec v_{\db \dd} + \vec v_{\db \de} - \vec v_{\dd \de},
\end{align*}}{70}

for any two different data values $\dd,\de$ different than $\db$.
 On the other hand\changed{ side}{}{44}, there is no $\N$-solution, as there is no similar presentation
 of \changed{$\vec v$}{$\vec{v_{\db}}$}{70} in terms of data vectors $\vec v_{\dd \de}$ that \changed{use}{uses}{47} nonnegative coefficients. Simply, every 
vector that is a sum of \changedd{of }{}{null}data vectors from the 
family $\setInput$ must be strictly positive for at least two data values. 
Furthermore, if \changed{$\vec v(\db) = 3$}{$\vec{v_{\db}}(\db) = 3$}{70} instead of $2$, then there is no $\Z$-solution, too. Indeed, notice that $\sum_{\da 
\in \setD} \vec{w}(\da)$ is always 
an even number if $\vec{w}$ is a sum of data vectors from the family $\setInput.$
\qed
\end{exa}

The above simple example is covered by the theory developed in~\cite{DBLP:conf/lics/HofmanLT17}. \changed{Here}{Here,}{169,170} we extend the \changed{previous 
result}{results from~\cite{DBLP:conf/lics/HofmanLT17}}{53} to
data vectors in $\kset{\setSizeK}{\setD}\xrightarrow{} \ktuple{\dimension}{\Z}$, where $\kset{\setSizeK}{\setD}$ stands for the \changed{${\setSizeK}$ 
elements}{${\setSizeK}$-element}{54} subsets of 
$\setD$. The complexity of 
analysis of data vectors in $\kset{\setSizeK}{\setD}\xrightarrow{} \ktuple{\dimension}{\Z}$ can already be observed 
for $\setSizeK=2$ and $\dimension=1.$
\begin{exa}\label{ex:two}
Consider data vectors \changed{of a form}{in}{56} $\kset{2}{\setD} \to{} \Z$ where $\kset{2}{\setD} =\{ \{\dd,\de\}\colon \dd,\de\in \setD,\ \dd\neq \de\}$.
You can think of them as weighted graphs with vertices labelled with elements of $\setD$.
Suppose the set $\setInput$ is a set of triangles with weights of all edges \changed{equal $1$}{equal to $1$}{58}\changed{, i.e.}{\ i.e.}{22}

\begin{equation*}
\vec{v_{\dc \dd \de}}(x)=\begin{cases}
          1 \quad &\text{if} \, x \in\{\{\dc,\dd\},\{\dd,\de\},\{\de,\dc\}\} \\
          0 \quad &\text{otherwise}\\
     \end{cases}
  \end{equation*}
and $\setInput=\setof{\vec v_{\dc \dd \de}\colon \kset{2}{\setD} \to \Z}{\dc, \dd, \de\in\setD, \dd\neq \de \neq \dc \neq \dd}$ (see Figure~\ref{fig:77}). 

\begin{figure}[h]
\begin{tikzpicture}
[place/.style={circle,draw=\mycolorOne!50,fill=\mycolorOne!20,thick,inner sep=0pt,minimum size=4mm}]
\node[place] (DA) at (0, 0) {$\da$};
\node[place] (DB) at (2, 2) {$\db$};
\node[place] (DC) at (4, 0) {$\dc$};
\draw[-, color=black] (DA) to node [left] {$1$} (DB);
\draw[-, color=black] (DA) to node [below ] {$1$} (DC);
\draw[-, color=black] (DB) to node [right ] {$1$} (DC);
\end{tikzpicture}
\caption{\changedd{Graph}{The graph}{null} representing the data vector $\vec{v_{\da\db\dc}}$.}\label{fig:77}
\end{figure}

The set $\setInput$ is closed under 
data permutations, indeed for any 
data 
permutation $\pi$ and any $\vec v_{\dc \dd \de}\in \setInput$ the data vector 
$\changed{}{\vec{v_{\dc\dd\de}}\circ \pi =}{NULL}\vec v_{\pi^{-1}(\dc) \pi^{-1}(\dd) \pi^{-1}(\de)}\in \setInput$.
Finally, we want to know if $\setInput$ and the following target data vector $\vec{v_{\dc \dd}}$ \changed{admits}{admits}{70} a $\Z$-solution   
\begin{equation*}
\vec{v_{\dc \dd}}(x)=\begin{cases}
          6 \quad &\text{if} \, x = \{\dc,\dd\} \\
          0 \quad &\text{otherwise}\\
     \end{cases}
  \end{equation*}
(\changed{$v_{\dc \dd}$}{$\vec{v}_{\dc \dd}$}{NULL} is a single edge with weight 6).
The answer is yes, but it is not trivial,
\[
\vec{v_{\dc\dd}}= (\vec{v_{\db\dd\dc}} - \vec{v_{\db\dd\da}} + \vec{v_{\dd\dc\de}} - \vec{v_{\dd\de\da}} + \vec{v_{\db\da\de}} - \vec{v_{\db\dc\de}} )+
(\vec{v_{\dd\db\dc}}  - \vec{v_{\dc\db\da}} + \vec{v_{\dd\de\dc}} - \vec{v_{\dc\de\da}} + \vec{v_{\de\db\da}}  - \vec{v_{\de\db\dd}}) + 2\vec{V_{\da\dd\dc}}
\]
\changed{}{where $\da,\db,$ and $\de$ are any data values different than $\dc,\dd$.}{75}

The idea behind the above sum is presented in Figure~\ref{fig:6}.

\begin{figure}[h]
\newcommand{\trojkat}[8]{
\node      (A#8)    at ( #1 , #2 )     {};
\node      (B#8)    at ( #3, #4 )     {};
\node      (C#8)    at ( #5 , #6 )     {};
\draw[-, color=#7] (A#8) to node [above right] {} (B#8);
\draw[-, color=#7] (A#8) to node [below right] {} (C#8);
\draw[-, color=#7] (B#8) to node [below ] {} (C#8);
}

\begin{center}
\begin{tikzpicture}
[place/.style={circle,draw=\mycolorOne!50,fill=\mycolorOne!20,thick,inner sep=0pt,minimum size=4mm}]

\node[place] (DA) at (9.95, 10) {$\da$};
\node[place] (DB) at (12.3, 10) {$\db$};
\node[place] (DC) at (7.7, 10) {$\dc$};
\node[place] (DD) at (9.95, 12.1) {$\dd$};
\node[place] (DE) at (9.95, 7.8) {$\de$};

\trojkat{10}{10.15}{10}{12}{12}{10.15}{\mycolorTrzy}{1}
\trojkat{10}{9.85}{10}{7.9}{12}{9.85}{\mycolorOne}{2}
\trojkat{9.9}{10.15}{9.9}{12}{7.95}{10.15}{\mycolorTrzy}{3}
\trojkat{9.9}{9.85}{9.9}{7.9}{7.95}{9.85}{\mycolorOne}{4}

\draw[-, color=\mycolorOne] (7.9, 10.23) to node [above right] {} (9.72,11.98);
\draw[-, color=\mycolorOne] (10.20, 11.98) to node [above right] {} (12.10,10.20);
\draw[-, color=\mycolorOne] (8, 10.03) to node [above right] {} (12.02,10.03);

\draw[-, color=\mycolorTrzy] (7.9, 9.78) to node [above right] {} (9.72,7.96);
\draw[-, color=\mycolorTrzy] (10.20, 7.96) to node [above right] {} (12.10,9.78);
\draw[-, color=\mycolorTrzy] (8, 9.97) to node [above right] {} (12.02,9.97);

\node (R) at (13,10) {$=$};

\node[place] (DA1) at (13.95, 10) {$\da$};
\node[place] (DD1) at (13.95, 12.1) {$\dd$};
\node[place] (DE1) at (13.95, 7.8) {$\de$};

\draw[-, color=black] (DA1) to node [right] {$2$} (DD1);
\draw[-, color=black] (DA1) to node [right] {$-2$} (DE1);

\node (TEXT) at (18,10) {
\begin{minipage}{5cm}
First we construct a gadget $\vec{g_{\dd\da\de}}$ as presented on the left. Blue color denotes addition and orange subtraction of an edge.
$\vec{g_{\da\dd\de}}= \vec{v_{\da\db\dd}} - \vec{v_{\da\db\de}} + \vec{v_{\da\dc\dd}} - \vec{v_{\da\dc\de}} + \vec{v_{\dc\db\de}} - \vec{v_{\dc\db\dd}}$

Next, we use such gadgets combined with triangles in the following way \changed{$g_{\dc\dd\da}+g_{\dd\dc\da} + 2\vec{v_{\da\dc\dd}}=\vec{v_{\dc\dd}}.$}
{$\vec{g}_{\dc\dd\da}+\vec{g}_{\dd\dc\da} + 2\vec{v_{\da\dc\dd}}=\vec{v_{\dc\dd}}.$
}{86}
\end{minipage}
};

\end{tikzpicture}
\end{center}
\caption{\changed{Ida}{Idea}{88} behind \changedd{}{the}{null} construction in Example~\ref{ex:two}.}\label{fig:6}
\end{figure}

\changed{}{If we ask about $\N$-solvability then the answer is no. If we add triangles then the number of nonzero edges will be greater than $2$, so there 
is no way of reaching $\vec{v}_{\dc\dd}$.}{89}
\changed{What}{But, what}{89} if we change $6$ to $3$? The answer is postponed to Section~\ref{sec:tohyp}.
 \qed
\end{exa}

\subsection{Related work and our contribution.}
The above-discussed, simplest extension of the $\Z$-solvability problem to data vectors of the form $\setD\xrightarrow{} \ktuple{\dimension}{\Z}$ 
is in \PTIME, and the $\N$-solvability is \changed{$\NPTIME-complete$}{\NPTIME-complete}{20,92}~\cite{DBLP:conf/lics/HofmanLT17}.
Further known results concern the more general case of \emph{ordered} data domain $\setD$~\cite{DBLP:conf/concur/HofmanL18}\changed{: while}{. For ordered data 
we assume that the set of data forms a dens linear order and the set of data permutations is restricted to the set of order preserving bijections 
$\setD\to{}\setD$. In the ordered data case}{93} the $\Z$-solvability problem remains in \PTIME,\changed{}{\ while}{93} the complexity of the 
$\N$-solvability is
equivalent to the reachability problem of vector addition systems with states (VASS),
or Petri nets, and hence \changed{$\AckermanCOMPLETE$}{\AckermanCOMPLETE}{20}~\cite{DBLP:journals/corr/abs-2104-12695, DBLP:journals/corr/abs-2104-13866}.
The increase of complexity caused by the order in data is thus remarkable.
An example of a result that builds on top of \changed{the}{}{97} \cite{DBLP:conf/lics/HofmanLT17} is 
\cite{Continuous-Reachability-for-Unordered-Data-Petri-Nets-is-in-PTime}\changed{;}{,}{97}
where the continuous reachability problem for unordered data nets is shown to be in \PTIME. The question if the continuous reachability problem is 
solvable for the ordered data 
domain remains open.
It is particularly interesting due to \cite{Bonnet2010ComparingPD}, 
where the coverability problem in timed data nets~\cite{DBLP:conf/apn/AbdullaN01} is 
proven 
to be 
interreducible with the coverability problem in ordered data nets.

In this paper, we perform a further generalisation to ${\setSizeK}$-element subsets of unordered data\changed{, i.e.,}{\ i.e. we}{22} 
consider data vectors of the form $\kset{\setSizeK}{\setD} \xrightarrow{} \ktuple{\dimension}{\Z}$, where $\kset{\setSizeK}{\setD}$ stands for
the ${\setSizeK}$-element subsets of $\setD$.
We prove two main results: first, for every fixed $\setSizeK\geq 1$ 
the complexity of the $\Z$-solvability problem again remains polynomial.
\changed{Second, we prove decidability and provide an upper $\NETIME$ complexity bound for 
the $\N$-solvability problem.}
{Second, we present a \NETIME\ algorithm for the $\N$-solvability problem.}{106}
This is done by an improvement of techniques developed in~\cite{DBLP:conf/lics/HofmanLT17}. Namely, we \changed{non-trivially}{nontrivially}{107} extend 
\changed{theorems}{Theorems}{107,119} 11 and 15 
from~\cite{DBLP:conf/lics/HofmanLT17}.
\begin{itemize}
 \item To address $\N$-solvability, we reprove \changed{the theorem}{Theorem}{109} 11 from~\cite{DBLP:conf/lics/HofmanLT17} in a more general setting (the proof 
is 
slightly modified). Next\changed{}{,}{110} 
we 
combined it 
with new idea to obtain a reduction to the (easier) $\Z$-solvability,
witnessing a nondeterministic exponential blowup.
\item Our approach to $\Z$-solvability is an extension \changed{the theorem}{of Theorem}{112} 15 from~\cite{DBLP:conf/lics/HofmanLT17}. 
Precisely, if we reformulate the $\Z$-solvability question in terms of 
(weighted) hypergraphs, then there is a natural way to lift the 
characterisation of $\Z$-solvability proposed in Theorem 15~\cite{DBLP:conf/lics/HofmanLT17}. 
The main contribution of this paper is a new tool-box developed to prove the lifted theorem 15 from~\cite{DBLP:conf/lics/HofmanLT17}. 
The new characterisation is easily checkable in polynomial 
time, \changed{what gives us}{resulting in}{117} the algorithm.
\end{itemize}

\changed{The analogous of \changed{theorems}{Theorems}{107,119} 11 and 15 from~\cite{DBLP:conf/lics/HofmanLT17} are not present in 
\cite{DBLP:conf/concur/HofmanL18},}
{
No analogous of \changed{theorems}{Theorems}{107,119} 11 and 15 from~\cite{DBLP:conf/lics/HofmanLT17} appear in 
\cite{DBLP:conf/concur/HofmanL18} or are otherwise known,
}{119}
thus we are pessimistic about 
applicability of the studied here approach in case of the ordered data domain.

As we comment in \changed{the conclusions}{Conclusions}{122},
we believe that elaboration of the techniques of this paper allows also tackling the case
of tuples of unordered data. 

\para{Motivation}
Our motivation for this research is two-fold. 
\changed{From}{On the one hand, from}{124} a foundational research perspective, our results are a part of a wider research program
aiming at lifting computability results in finite-dimensional linear algebra to
its orbit-finite-dimensional counterpart. 
Up to now, research \changed{is focused}{has been focused}{127} on understanding 
\changed{expressibility}{solvability}{124}\changed{}{~\cite{DBLP:conf/lics/HofmanLT17, DBLP:conf/concur/HofmanL18}}{127}, but there are other natural questions 
about 
definitions of bases, dimension, linear transformations, etc. 
\changed{We are focused on expressibility, as it gives us a tool for the analysis of systems with 
data.}{On the other hand, we are interested in the analysis of systems with data and solvability may be a useful tool.}{124}
 We highlight three areas where understanding of $\N/\Z$-\changed{expressibility}{solvability}{124} may be crucial for further development.

{\bf Unordered Data nets reachability/coverability}~\cite{Nets-with-Tokens-which-Carry-Data}. The model can be seen as a special type of Coloured Petri 
nets~\cite{ColoredPetriNets}.
Data nets are an extension of Petri nets where every token caries a tuple of data values. \changed{In addition}{In addition,}{134} each 
transition is equipped with a Boolean formula that connects data of tokens that are consumed and data of tokens that are produced (for unordered data the 
formula may use $=$ and $\neq$). To fire a transition we take a valuation satisfying the formula and according to it we remove and produce tokens. 

\begin{exa}
Consider the following simple net with $3$ places and one transition. The initial marking has $3$ tokens each with two data values. 

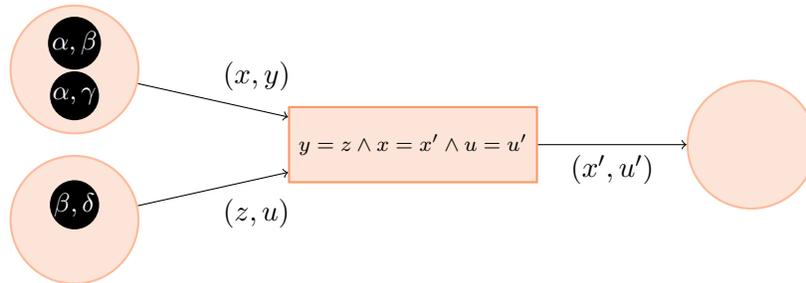
\begin{center}
\begin{tikzpicture}
[place/.style={circle,draw=\mycolorOne!50,fill=\mycolorOne!20,thick,inner sep=0pt,minimum size=17mm},
token/.style={circle,draw=black!100,fill=black!100,thick,inner sep=0pt,minimum size=4mm}]
 \tikzstyle{transition}=[rectangle,thick,draw=\mycolorOne!75, fill=\mycolorOne!20,minimum size=10mm]
\node[place]      (A1)         {$ $};
\node[token]      (TOK1)    at (0,0.35)     {\color{white}{\small \bf $\da,\db$}};
\node[token]      (TOK2)    at (0,-0.35)     {\color{white}{\small \bf $\da,\dc$}};
\node[place]      (B1)      at (0,-2)  {$ $};

\node[token]      (TOK3)    at (0,-1.8)     {\color{white}{\small \bf $\db,\dd$}};

\node[transition]      (T1)     at (4.5,-1) {\scriptsize $y = z \land x=x' \land u=u' $};
\node[place]      (C1)     at (9,-1) {};
\draw[->] (A1) to node [above right] {$(x,y)$} (T1);
\draw[->] (B1) to node [below right] {$(z,u)$} (T1);
\draw[->] (T1) to node [below ] {$(x',u')$} (C1);
\end{tikzpicture}
\end{center}

If we valuate $x=\da, y=\db, z=\db, u=\dd, x'=\da, u'=\dd$ then we may fire the transition and get the new marking.

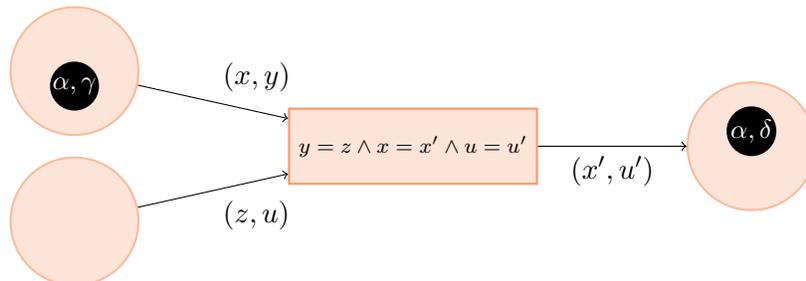
\begin{center}
\begin{tikzpicture}
[place/.style={circle,draw=\mycolorOne!50,fill=\mycolorOne!20,thick,inner sep=0pt,minimum size=17mm},
token/.style={circle,draw=black!100,fill=black!100,thick,inner sep=0pt,minimum size=4mm}]
 \tikzstyle{transition}=[rectangle,thick,draw=\mycolorOne!75, fill=\mycolorOne!20,minimum size=10mm]
\node[place]      (A1)         {$ $};
\node[token]      (TOK2)    at (0,-0.2)     {\color{white}{\small \bf $\da,\dc$}};
\node[place]      (B1)      at (0,-2)  {$ $};

\node[transition]      (T1)     at (4.5,-1) {\scriptsize $y = z \land x=x' \land u=u' $};
\node[place]      (C1)     at (9,-1) {};
\node[token]      (TOK4)    at (9,-0.8)     {\color{white}{\small \bf $\da,\dd$}};
\draw[->] (A1) to node [above right] {$(x,y)$} (T1);
\draw[->] (B1) to node [below right] {$(z,u)$} (T1);
\draw[->] (T1) to node [below ] {$(x',u')$} (C1);
\end{tikzpicture}
\end{center}

\changed{If we chose}{If we had chosen}{152} another valuation $x=\da, y=\db, z=\db, u=\dc, x'=\da, u'=\dc$ then the formula would hold, 
but we would lack tokens that can be consumed so we 
could not fire the transition with this valuation.

\end{exa}

The reachability and coverability questions can be formulated as usual for Petri nets.
It is not hard to imagine that some workflow or a flow of data through a program can be modelled with data nets.
Unfortunately, in this richer model the reachability and coverability problems are undecidable~\cite{10.1007/978-3-319-39086-4_3} already for $k=2$.
For $k=1$\changed{, i.e.}{\ i.e.}{22} data vectors in \changed{$\setD\xrightarrow{} \ktuple{\dimension}{\Z}$}{$\setD\xrightarrow{} 
\ktuple{\dimension}{\Z}$,}{160} the status of reachability is unknown and coverability is decidable but known to 
be Ackermann-hard ~\cite{DBLP:conf/birthday/LazicT17}.
This is not a satisfying answer for engineers and in this case we should look for over and under approximations of the reachability relation, or for some 
techniques that will help in the analysis of industrial cases. One of the classic over-approximations of the Petri nets reachability relation is so-called 
integer reachability\changed{(}{}{165} or Marking Equation\changed{ Lemma}{}{165} in~\cite{desel_esparza_1995} \changed{lemma 2.12)}{Lemma 2.12,}{165} where the 
number of tokens in some places may go negative during the run. It can be 
encoded as integer programming and solved in \NPTIME. Its analogue for data nets can be stated as $\N$-solvability \changed{of an orbit finite linear 
equations.}{over an input set of data vectors, where the set is closed under data permutations.}{167} 

\changed{Here}{Here,}{169,170} we should mention that the integer reachability is a member of a wider family of algebraic techniques for Petri nets. We refer 
to~\cite{SilvaTC96} for 
an exhaustive overview of linear-algebraic and integer-linear-programming techniques 
in the analysis of Petri nets. The usefulness of these techniques is confirmed by multiple applications including, for instance, recently proposed
efficient tools for the coverability problem of Petri 
nets~\cite{DBLP:conf/rp/GeffroyLS16,DBLP:conf/tacas/BlondinFHH16}. 

{\bf $\pi$-calculus.}
Another formalism close to unordered data nets are $\nu$-nets (unordered data nets \changed{equipped with a global freshens test.}{additionaly equipped with 
an operation of creating a new datum that is not present in the current configuration of the net; in other words, a transition may force creation of data that 
are globally fresh/unique}{175}). \changed{In \cite{Velardo-pi}}{In \cite{Velardo-pi},}{175} 
Rosa-Velardo 
observes that they are equivalent to so-called \emph{multiset rewriting with name binding} systems which, as he showed, are a formalism equivalent to 
$\pi$-calculus. 
\changed{Thus}{Thus,}{177} one may try to transfer algebraic techniques for data nets to $\pi$-calculus. This is a long way, 
but there is no possibility to start it without a good understanding of integer solutions of linear equations with data. 

\changed{Here}{Here,}{169,170} it is \changed{worth to mention that}{worth mentioning that}{180} 
in $\pi$\changed{ calculus}{-calculus}{186} we use constructs like $\bar{c}\langle y\rangle.P$ which send a datum $y$ trough the channel $c$, thus the sending 
operation is 
parameterized with 
a pair of data values i.e. the name of the channel and the data value. \changed{Thus}{Thus,}{177} one \changed{can not}{cannot}{357} expect that already 
existing results for data vectors 
$\setD\xrightarrow{} \Z^\dimension$~\cite{DBLP:conf/lics/HofmanLT17} will be sufficient and we will need at least theory for $\setD^2\xrightarrow{} 
\Z^\dimension$, for 
example, to count messages that are sent and received. 
\changed{
But even pairs of data may not be enough. Rosa-Velardo encodes $\pi$-calculus in $\nu$-nets using derivatives, which essentially are terms of bounded depth. 
The bound depends on the process definition. The number of data values that parametrize each derivative is bounded. 
In the encoding, each derivative corresponds to a place while particular instantiations of derivatives we represent with tokens. 
Thus, if use this encoding, then we should expect arbitrary tuples of data and not only pairs. 
}{
In fact even $\setD^2\xrightarrow{}\Z^\dimension$ may be not enough. In~\cite{Velardo-pi}, the fundamental concept for the encoding are derivatives, which 
essentially are terms of bounded depth labelled with data values. In his encoding, Rosa-Velardo represent each derivative with a token, so each token has to 
carry all data that are needed to identify the derivative. For a given process definition, he produces a $\nu$-net with tokens with bounded but 
arbitrary high number of data values. Thus, if one wants to use linear algebra with data to describe some properties of the produced $\nu$-net, he will have to 
work with data vectors in $\ktuple{\setSizeK}{\setD}\to{}\ktuple{\dimension}{\Z}$, for $\setSizeK$ grater than $2$. 
}{186-189}

{\bf \changed{Parikh'}{Parikh\textquotesingle}{190}s theorem.} Finally, we may try to lift \changed{Parikh'}{Parikh\textquotesingle}{190}s theorem from 
context-free grammars 
and finite automata to context-free grammars with data~\cite{DBLP:journals/acta/ChengK98} and 
register automata~\cite{DBLP:journals/tcs/KaminskiF94}. It is not 
clear to what extent it is possible but there are some 
promising results~\cite{RAPARIKH}.
If we want to use this lifted \changed{Parikh'}{Parikh\textquotesingle}{190}s theorem then we have to work with semilinear sets with data and be able to check 
things like membership or 
\changed{non-emptiness}{nonemptiness}{107} of the intersection. \changed{Here}{Here,}{169,170} one more time techniques to solve systems of linear equations 
with data will be inevitable.

\para{Outline}
In Section~\ref{sec:eq} we introduce the setting and define the problems.
Next, in Section~\ref{sec:tohyp} we provide the polynomial-time procedure for the $\Z$-solvability 
problem: the hypergraph reformulation and an effective characterisation
 of hypergraph solvability.
In Section~\ref{sec:prove} \changed{due to}{for}{200} pedagogical reasons we present 
the proof of the characterisation if data vectors are restricted 
to $\kset{2}{\setD}\xrightarrow{} \ktuple{\dimension}{\Z}$. 
 After this, in Sections~
\ref{sec:Proof:thm:core},~\ref{sec:matrix},~\ref{sec:simple_hypergraphs},~\ref{sec:Expressing with simple 
 hypergraphs},~\ref{sec:The construction of simple hypergraphs},~\ref{sec:core_proof} we provide the full proof of the characterisation. 
 The proof follows the same steps\changed{,}{}{202}
as the proof of the case $\kset{2}{\setD}\xrightarrow{} \ktuple{\dimension}{\Z}$, but is much more \changed{involving on}{involved at}{204} the technical level.
Next, in Section~\ref{sec:red} we present a reduction from $\N$- to $\Z$-solvability. 
Finally, Section~\ref{sec:conclusions} concludes this work.

\section{Linear equations with data.}\label{sec:eq}

In this section, we introduce the setting of linear equations with data and formulate our results. 
For a gentle introduction of the setting, we start by recalling classical linear equations.

Let $\Z$ and $\N$ denote integers and nonnegative integers, respectively.
Classical linear equations are of the form
\[
a_1 x_1 + \ldots + a_m x_m = a,
\] 
where $x_1 \ldots x_m$ are variables (unknowns), and $a_1 \ldots a_m\in \Z$ are integer coefficients.
For a finite system $\systemOfEquationsU$ of such equations over the same variables 
$x_1\ldots x_m$, a solution of ${\systemOfEquationsU}$
is a vector $(n_1\ldots n_m) \in \Z^m$ such that the valuation 
$x_1 \mapsto n_1\ldots$, $x_m \mapsto n_m$ satisfies all equations in $\systemOfEquationsU$.
It is well known that \emph{integer solvability problem}
($\Z$-solvability problem)\changed{, i.e.,}{\ i.e.}{22} the question whether $\systemOfEquationsU$ has a solution $(n_1\ldots n_m) \in \Z^m$,
is decidable in \PTIME.
In the sequel we are often interested in nonnegative integer solutions $(n_1\ldots n_m) \in \N^m$,
but one may consider also other solution domains than $\N$.
It is well known that the \emph{nonnegative-integer solvability problem} ($\N$-solvability problem) 
of linear equations\changed{, i.e.}{\ i.e.}{22}~the question whether 
${\systemOfEquationsU}$ has a nonnegative-integer solution, is NP-complete~(for hardness see~\cite{Karp21NP-complete-Problems}; NP-membership is a consequence 
of 
\cite{Pottier}).
The complexity remains the same for other natural variants of this problem, for instance, for inequalities 
instead of equations (a.k.a.~integer linear programming).
The $\Numbers$-solvability problem (where $\Numbers \in\{\Z, \N\}$) 
is equivalently formulated as follows: for a given finite set of coefficient vectors 
$\setInput = \{\vec v_1\ldots \vec v_m\} \subseteq \ktuple{\dimension}{\Z}$ and a target vector $\vec v \in \ktuple{\dimension}{\Z}$ 
 (we use bold font to distinguish vectors from other elements), check whether 
$\vec v$ is an \emph{$\sumsBez{\Numbers}$} of $\setInput$\changed{, 
i.e.,}{\ i.e. }{22} 
\begin{align}\label{eq:lineq}
\vec{v} \in \sums{\Numbers}{\setInput} = \setof{n_1 \cdot \vec v_1 + \ldots + n_m \cdot \vec v_m}{n_1\ldots n_m \in \Numbers}.
\end{align}
The dimension $\dimension$ corresponds to the number of equations in $\systemOfEquationsU$.

\para{Data vectors}
Linear equations can be naturally extended with data.
In this paper, we assume that the data domain $\setD$ is a countable infinite set,
whose elements are called \emph{data values}.
The bijections $\rho\colon \setD\to\setD$ are called \emph{data permutations}.
For a set $\setX$ and $\setSizeK\in\N$, by $\kset{\setSizeK} \setX$ we denote the set of all \changed{${\setSizeK}$-elements}{${\setSizeK}$-element}{54} 
subsets of $\setX$
(called \emph{${\setSizeK}$-sets} in short).
Data permutations \emph{lift naturally} to ${\setSizeK}$-sets of data values:
$\rho(\{\da_1\ldots \da_{\setSizeK}\}) = \{\rho(\da_1)\ldots \rho(\da_{\setSizeK})\}$. 

Fix a positive integer $\setSizeK\geq 1$.
A \emph{data vector} is a function $\vec{v}\colon \kset{\setSizeK}{\setD} \rightarrow \ktuple{\dimension}{\Z}$
such that $\vec v(x) = \vec{0}\in \ktuple{\dimension}{\Z}$ for all but finitely many $x \in \kset{\setSizeK}{\setD}$.
(Again, we use bold font to distinguish data vectors from other elements.)
\changed{The numbers ${\setSizeK}$ and $\dimension$ we call}{We call the numbers ${\setSizeK}$ and $\dimension$}{237} the \emph{arity} and the \emph{dimension} of $\vec v$, respectively.

The vector addition and scalar multiplication are lifted to data vectors pointwise:
$(\vec{v}+\vec{w})(x) \eqdef \vec{v}(x)+\vec{w}(x)$, and
$(c\cdot \vec v)(x)\eqdef c\vec{v}(x)$.
Further\changed{}{,}{241} for a data permutation $\pi\colon\setD\xrightarrow{} \setD$ by $\vec{v}\circ \pi$ we mean the data vector defined as follows 
$(\vec{v}\circ \pi) (\changed{\da_1, \da_2\ldots \da_{\setSizeK}}{\{\da_1, \da_2\ldots \da_{\setSizeK}\}}{242}) \eqdef 
\vec{v}(\changed{\pi(\da_1),\pi(\da_2)\ldots \pi(\da_{\setSizeK})}{\{\pi(\da_1),\pi(\da_2)\ldots \pi(\da_{\setSizeK})\}}{242})$ for all $\da_1, \da_2\ldots 
\da_{\setSizeK}\in \setD$. It is the natural lift 
of normal function composition.
For a set $\setInput$ of data vectors we define
\[
\perm{\setInput} = \setof{\vec v \circ \pi}{\vec v \in \setInput, \pi \changed{\in \text{data permutations}}{\text{ is a data permutation}}{245} }.
\]
A data vector $\target$ 
is said to be \changedd{an}{a}{247} \emph{$\Numbers$-permutation sum} of a finite set of data vectors $\setInput$ if
(we deliberately overload the symbol $\sums{\Numbers}{\placeHolder}$ and use it for data vectors, 
while in~\eqref{eq:lineq} it is used for (plain) vectors)
\begin{equation}
\target \in \sums{\Numbers}{\perm{\setInput}} =
\setof{n_1 \cdot \vec{v_1} + \ldots + n_m \cdot \vec{v_m}}
{\\ \vec{v_1}\ldots\vec{v_m} \in \perm{\setInput}, n_1\ldots n_m \in \Numbers}. 
\end{equation}
We investigate the following decision problems (for $\Numbers \in\{\Z,\N\}$):

\decproblem{$\Numbers$-solvability}
{A finite set $\setInput$ of data vectors and a target data vector $\vec v$, all of the same arity and dimension}
{Is $\vec v$ \changedd{an}{a}{247} $\Numbers$-permutation sum of $\setInput$?}
{.6}

\changed{}{The insightful reader may notice that, in the motivating examples we consider $\Numbers$-sums of an input set of data vectors that is closed under 
data permutations. Here, we switch to the $\Numbers$-solvability problem, which is defined by $\Numbers$-permutation sum. But observe that, to formalise 
expressibility by $\Numbers$-sums as a problem we have to provide a finite representation of the input set. That is why as an input to the problem we take a 
finite set of data vectors $\setInput$, and we consider $\Numbers$-sums of its closure under data permutations i.e. of the set $\perm{\setInput}$. This is why 
we consider $\Numbers$-permutation sums instead of pure $\Numbers$-sums.}{253}  
  
For complexity estimations we assume binary encoding of numbers appearing in the input to all decision problems discussed in this paper.
Our main results are the following complexity bounds:
\begin{thm}\label{thm:problem2}
For every fixed arity ${\setSizeK}\in \N$, the $\Z$-solvability problem is in \PTIME.
(The dependency \changed{form}{on}{261} ${\setSizeK}$ is exponential).
\end{thm}
\begin{thm}\label{thm:problem1}
For every fixed arity ${\setSizeK}\in \N$, the $\N$-solvability problem is in \NETIME.
\end{thm}
For \changed{a}{the}{263} special case of the solvability problems when the arity ${\setSizeK} = 1$, 
the \PTIME\ and \NPTIME\ complexity bounds, respectively, have been shown in~\cite{DBLP:conf/lics/HofmanLT17}.
Thus, according to Theorem~\ref{thm:problem2}, 
in the case of $\Z$-solvability the complexity remains polynomial for every ${\setSizeK} > 1$.
In the case of $\N$-solvability the \changed{$\NPTIME-hardness$}{\NPTIME-hardness}{20} carries over to every ${\setSizeK}>1$, and
hence a complexity gap remains open between \NPTIME\ and \NETIME.

\section{Proof of Theorem~\ref{thm:problem2}.}\label{sec:tohyp}

We start by reformulating the problem in terms of (undirected) weighted uniform hypergraphs
(Lemma~\ref{lem:red} below).

Fix a positive integer ${\setSizeK} \geq 1$.
By a \emph{${\setSizeK}$-hypergraph} we mean a pair $\hypergraphH = (\setVertices, \mu)$ where $\setVertices\subset \setD$ 
is a finite set called vertices
and $\mu\colon \kset{\setSizeK}{\setVertices} \to \ktuple{\dimension}{\Z}$ is a weight function (when ${\setSizeK}$ 
is not relevant we skip it and write a hypergraph).
As before, \changed{the numbers ${\setSizeK}$ and $\dimension$ we call}{we call the numbers ${\setSizeK}$ and $\dimension$}{237}  the arity and the dimension of $\hypergraphH$, respectively. 
\changed{In the case of arity $\setSizeK=2$}{When $\setSizeK=2$,}{274} we speak of \emph{graphs} instead of hypergraphs.
(Note however that the (hyper)graphs we consider are always \emph{weighted}, with weights from $\ktuple{\dimension}{\Z}$.)

Because vertices are\changed{ essentially}{}{277} data then instead of usual $u,v$ for vertices we will use Greek letters $\da,\db,\dc,\dd,\de\ldots$. 
\changedd{}{Also, for a hypergraph we denote the set of its vertices $\setVertices$ by $$\Vertices{\hypergraphH}= \setVertices.$$}{405} 

The set of hyperedges is then defined as
\[
\he \hypergraphH = \setof{{\edge} \in \kset{\setSizeK} \setVertices}{\mu({\edge}) \neq {\vec{0}}}.
\]

When $\da\in {\edge}$ for $\da\in \setVertices$ and ${\edge} \in \he \hypergraphH$, we say that the vertex $\da$ 
is \emph{\changed{incident}{adjacent}{NULL}} with the 
hyperedge ${\edge}$.
The degree of $\da$ is the  number of hyperedges \changed{incident}{adjacent}{NULL} with $\da$. 
\changed{Vertices of degree 0 we call}{We call vertices of degree $0$}{237} \emph{isolated}.
Two hypergraphs are \emph{isomorphic} if there is a bijection between their sets of 
vertices that preserves the value of the weight function. 
Two hypergraphs are \emph{equivalent} if they are isomorphic after
removing their isolated vertices.
For a \changedd{set}{family}{null} ${\setHypergraphsH}$ of hypergraphs, by $\eqs {\setHypergraphsH}$ we denote the set of all hypergraphs equivalent
to ones from ${\setHypergraphsH}$. If ${\setHypergraphsH}=\{\hypergraphH\}$ then instead of $\eqs{\{\hypergraphH\}}$ we write $\eqs{\hypergraphH}$.

Scalar multiplication and addition are defined naturally for hypergraphs.
First, for $c \in \Z$ and a hypergraph $\hypergraphH = (\setVertices, \mu)$, 
let $c \cdot \hypergraphH \eqdef (\setVertices, c\cdot \mu)$.
Second, given two hypergraphs $\hypergraphG = (\setVerticesW, \widetilde{\mu})$ and $\hypergraphH = (\setVertices, \mu)$
of the same arity ${\setSizeK}$ and dimension $\dimension$,
we first add isolated vertices to both hypergraphs to make their 
vertex sets equal to the union \changed{$\setVertices \cup \setVerticesW$}{$\setVerticesW \cup \setVertices$}{292},
thus obtaining $\hypergraphG' = (\setVerticesW\cup \setVertices, \widetilde{\mu}')$ and $\hypergraphH' = 
(\setVerticesW\cup \setVertices, \mu')$ with the 
accordingly extended
weight functions $\widetilde{\mu}', \mu'\colon \kset{\setSizeK}{{(\setVerticesW \cup \setVertices)}} \to \ktuple{\dimension}{\Z}$, 
and then define $\hypergraphG + \hypergraphH \eqdef (\setVerticesW\cup \setVertices, \widetilde{\mu}' + \mu')$.
Using these operations we define $\Numbers$-sums 
of a family ${\setHypergraphsH}$ of hypergraphs of the same arity and dimension
(again, we overload the symbol $\sums \Numbers {\placeHolder}$ further and use it for hypergraphs):
\[
\sums \Numbers {{\setHypergraphsH}} = 
\setof{c_1 \cdot \hypergraphH_1 + \ldots + c_l \cdot \hypergraphH_l}
{c_1\ldots c_l \in \Numbers, \hypergraphH_1\ldots \hypergraphH_l \in {\setHypergraphsH}}.
\]

We say that a hypergraph $\hypergraphH$ is a $\Numbers$-sum of ${\setHypergraphsH}$ \emph{up to equivalence} if 
$\hypergraphH$ is a $\Numbers$-sum of hypergraphs equivalent to elements of ${\setHypergraphsH}$: 
$
\hypergraphH \in \sums \Numbers {\eqs {\setHypergraphsH}}.
$
\begin{exa}
We illustrate the $\Z$-sums in arity $\setSizeK=2$\changed{, i.e.,}{\ i.e.}{22} using graphs.
Consider the following graph $\hypergraphG$ consisting of 3 vertices and 2 edges:

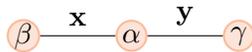
\begin{tikzpicture}
[place/.style={circle,draw=\mycolorOne!50,fill=\mycolorOne!20,thick,inner sep=0pt,minimum size=4mm}]
\node[place]      (A1)         {$\da$};
\node[place]      (B1)       [left=of A1]  {$\db$};
\node[place]      (C1)      [right=of A1] {$\dc$};
\draw (B1) to node [above] {$\vec{x}$} (A1) to node [above] {$\vec{y}$} (C1) ;
\end{tikzpicture}

\noindent
Let $\vec x, \vec y \in \ktuple{\dimension}{\Z}$ be arbitrary vectors.
\changed{Here}{Here, there}{169,170} are two examples of graphs which can be presented
as a $\Z$-sum of $\{\hypergraphG\}$ up to equivalence, using a sum
of two graphs equivalent to $\hypergraphG$: 

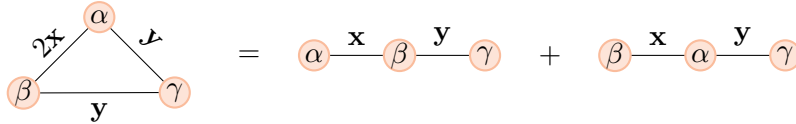
\begin{tikzpicture}
[place/.style={circle,draw=\mycolorOne!50,fill=\mycolorOne!20,thick,inner sep=0pt,minimum size=4mm}]
\begin{scope}[node distance=7mm and 7mm]%Here we change it for everything inside this scope
\node[place]      (B1i)      at (0,-1) {$\db$};
\node[place]      (A1i)      [above right =of B1i]   {$\da$};
\node[place]      (C1i)      [below right=of A1i] {$\dc$};
\draw (B1i) to node [above, sloped] {$2\vec x$} (A1i) to node [above, sloped] {$\vec y$} (C1i) 
to node [below] {$\vec y$} (B1i);

\node(R) at (3,-0.5) {=};

\node[place]      (Ai)    at (5,-0.5)                     {$\db$};
\node[place]      (Bi)       [left=of Ai]  {$\da$};
\node[place]      (Ci)      [right=of Ai] {$\dc$};
\draw (Bi) to node [above] {$\vec x$} (Ai) to node [above] {$\vec y$} (Ci);

\node(R1) at (7,-0.5) {+};

\node[place]      (Aii)    at (9,-0.5)                     {$\da$};
\node[place]      (Bii)       [left=of Aii]  {$\db$};
\node[place]      (Cii)      [right=of Aii] {$\dc$};
\draw (Bii) to node [above] {$\vec x$} (Aii) to node [above] {$\vec y$} (Cii);
\end{scope}
\end{tikzpicture}

\noindent
and a difference of two such graphs:

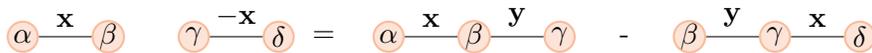
\begin{tikzpicture}
[place/.style={circle,draw=\mycolorOne!50,fill=\mycolorOne!20,thick,inner sep=0pt,minimum size=4mm}]
\begin{scope}[node distance=7mm and 7mm]%Here we change it for everything inside this scope
\node[place]      (A2ii)  at (-1,-2.5)     {$\da$};
\node[place]      (B2ii)       [right=of A2ii]  {$\db$};
\node[place]      (C2ii)      [right =of B2ii] {$\dc$};
\node[place]      (D2ii)      [right=of C2ii] {$\dd$};

\draw (A2ii) to node [above] {$\vec x$} (B2ii);
\draw (C2ii) to node [above] {$-\vec x$} (D2ii);

\node(Rii) at (3,-2.5) {=};

\node[place]      (Aii)    at (5,-2.5)                     {$\db$};
\node[place]      (Bii)       [ left=of Aii]  {$\da$};
\node[place]      (Cii)      [right=of Aii] {$\dc$};
\draw (Bii) to node [above] {$\vec x$} (Aii) to node [above] {$\vec y$} (Cii);

\node(R1ii) at (7,-2.5) {-};

\node[place]      (Cii)    at (9,-2.5)   {$\dc$};
\node[place]      (Aii)   [left =of Cii]  {$\db$};
\node[place]      (Bii)       [right=of Cii]  {$\dd$};
\draw (Aii) to node [above] {$\vec y$} (Cii) to node [above] {$\vec x$} (Bii);
\end{scope}
\end{tikzpicture}
\end{exa}

We are now ready to formulate the \emph{hypergraph $\Z$-sum problem},
to which \changed{the}{}{312} $\Z$-solvability is going to be reduced:

\decproblem{hypergraph $\Z$-sum}
{A finite set ${\setHypergraphsH}$ of hypergraphs and a target hypergraph $\hypergraphH$, all of the same arity and dimension}
{Is $\hypergraphH$ a $\Z$-sum of ${\setHypergraphsH}$ up to equivalence?}
{.60}

\begin{lem}\label{lem:red}
The $\Z$-solvability problem reduces in logarithmic space to the hypergraph $\Z$-sum problem.
The reduction preserves the arity and dimension.
\end{lem}
\begin{proof}
The reduction encodes each data vector $\vec v$ by a hypergraph $\hypergraphH = (\setVertices, \mu)$, where
\[
\setVertices \quad = \quad \bigcup \{x\in \kset{\setSizeK}{\setD} \mid \vec{v}(x)\not={\vec{0}}\} 
\]
and $\mu$ is the restriction of $\vec v$ to $\kset{\setSizeK}{\setVertices}$.
In this way\changed{}{,}{323} a set $\setInput$ of data vectors and a target data vector $\vec v$ are transformed into
a set of ${\setHypergraphsH}$ hypergraphs and a target hypergraph $\hypergraphH$ such that
$\vec v$ is a $\Z$-permutation sum of $\setInput$ if, and only if\changed{}{,}{325} $\hypergraphH$ is a $\Z$-sum of ${\setHypergraphsH}$
up to equivalence.
\end{proof}
Thus, from now on, we concentrate on solving the hypergraph $\Z$-sum problem. One may ask why we perform such a reduction. 
The reasons are of pedagogical nature, namely graphs are more convenient for examples and proofs by pictures.  

As the next step, we formulate our core technical result (Theorem~\ref{thm:core}). \changed{The theorem
provides a local characterisation of the hypergraph $\Z$-sum problem  
and in consequence it enables solving the problem in polynomial time.}{In the theorem we state that the hypergraph $\Z$-sum problem
is equivalent to a local $\Z$-sum problem, defined in the following paragraph. It is easy to design a polynomial time algorithm for the local $\Z$-sum problem. 
This gives us the proof of Theorem~\ref{thm:problem2}.}{331}
Let $\hypergraphH = (\setVertices, \mu)$ be a hypergraph. For a set $\setX \subset \setD$ we consider
$\setof{\edge\in \he \hypergraphH}{\setX \subseteq \edge}$ ---
the set of all hyperedges that include $\setX$ ---
and define the weight of $\setX$ as the sum of weights of all these edges:
\[
\weight{setX}(\hypergraphH) \eqdef \sum_{\edge\in \he \hypergraphH, \setX\subseteq \edge} \mu(\edge).
\]

In particular when $\setX\not\subseteq \setVertices$ then $\weight{setX}(\hypergraphH) =\vec{0}$. 
Further, $\weight{emptyset}(\hypergraphH)$ is the sum of weights of all 
hyperedges of $\hypergraphH$. 
When $\setX = \{\da\}$, then $\weight{setX}(\hypergraphH)$ is the sum of weights of all hyperedges \changed{incident}{adjacent}{NULL} with $\da$
(we write $\weight{da}$ 
instead of 
$\weight{set{da}}$
). Finally,
when the cardinality of $\setX$ equals ${\setSizeK}$ i.e. $\size{\setX} = {\setSizeK}$, $\weight{setX}(\hypergraphH)= \mu(\setX)$
is the weight of \changedd{a}{the}{null} hyperedge $\setX$\changed{, given as a 
defining component of $\hypergraphH$}{}{341}.

\begin{exa}
Let $\vec x, \vec y, \vec z \in \ktuple{\dimension}{\Z}$ be arbitrary vectors.
As an illustration, consider the graph $\hypergraphG$ on the left below, with the weights of chosen subsets of its vertices
listed on the right:

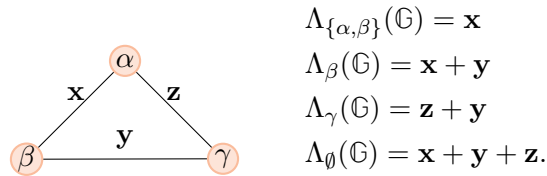
\begin{figure*}[h]
\begin{tikzpicture}
[place/.style={circle,draw=\mycolorOne!50,fill=\mycolorOne!20,thick,inner sep=0pt,minimum size=4mm}]
\node[place]      (A)                       {$\da$};
\node[place]      (B)       [below left=of A]  {$\db$};
\node[place]      (C)      [below right=of A] {$\dc$};
\draw (A) to node [above] {$\vec x$} (B) to node [above] {$\vec y$} (C) to node [above] {$\vec z$} (A);

\node (D) at (4,-0.5) {
\begin{minipage}[t][2cm][b]{0,5\textwidth}
\ \\
 \begin{align*}
     &\weight{{da,db}}(\hypergraphG)=\vec x\\
     &\weight{db}(\hypergraphG)=\vec x+\vec y\\
     &\weight{dc}(\hypergraphG)=\vec z+ \vec y\\
     &\weight{emptyset}(\hypergraphG)=\vec x+\vec y+\vec z.
 \end{align*}
\end{minipage}
};
\end{tikzpicture}
\end{figure*}
\end{exa}

Weights \changed{}{of sets }{351}are important as they form a family
of homomorphisms from hypergraphs to $\ktuple{\dimension}{\Z}$. Namely,
for any subset of vertices $\setX$ we have that 
\[\weight{setX}(\hypergraphH + \hypergraphH') = \weight{setX}(\hypergraphH) + \weight{setX}(\hypergraphH')\]
and \
\[c \weight{setX}(\hypergraphH)= \weight{setX}(c\hypergraphH)\text{, for any }c\in \Z.\]
This \changed{allow}{allows}{356} us to design a partial test for the hypergraph 
$\Z$-sum problem. If there is a set of vertices $\setX$ such that
$\weight{setX}(\hypergraphH)$ \changed{can not}{cannot}{357} be expressed as a $\Z$-sum of vectors in $\weight{setX}(\eqs {\setHypergraphsH})$ then
$\hypergraphH\not\in \sums \Z {\eqs {\setHypergraphsH}}$. This motivates the next definition.

We say that a ${\setSizeK}$-hypergraph $\hypergraphH\changedd{ = (\setVertices, \mu)}{}{405}$ is 
\emph{locally} a $\Z$-sum of a family $\setHypergraphsH$ of hypergraphs if
for every $\setX\subseteq \changedd{\setVertices}{\Vertices{\hypergraphH}}{405}$ of \changed{the }{}{360}cardinality $\size{\setX}\leq \setSizeK$, its weight 
$\weight{setX}(\hypergraphH)$ 
is a $\Z$-sum
of weights of $\size{\setX}$-element subsets of vertex sets of hypergraphs from $\setHypergraphsH$

\changedd{\begin{equation}
\label{eq:loc}
\weight{setX}(\hypergraphH) \in \sums \Z {\setof{\weight{setX'}(\hypergraphH ')}{\hypergraphH'= (\setVertices', \mu') \in {\setHypergraphsH}, 
\setX'\subseteq \setVertices', \size{\setX'} = \size{\setX}}}.
\end{equation}
}
{\begin{equation}
\label{eq:loc}
\weight{setX}(\hypergraphH) \in \sums \Z {\setof{\weight{setX'}(\hypergraphH ')}{\hypergraphH' \in {\setHypergraphsH}, 
\setX'\subseteq \Vertices{\hypergraphH'}, \size{\setX'} = \size{\setX}}}.
\end{equation}
}{405}
Note that we only consider hypergraphs $\hypergraphH'\in {\setHypergraphsH}$, and we do not need to consider hypergraphs 
$\hypergraphH' \in \eqs {\setHypergraphsH}$ equivalent to ones from ${\setHypergraphsH}$ as 
\changedd{
\[\setof{\weight{setX'}(\hypergraphH')}{\hypergraphH' = (\setVertices', \mu') \in {\setHypergraphsH},
\setX'\subseteq \setVertices', \size{\setX'} = \size{\setX}} = \setof{\weight{setX}(\hypergraphH')}{\hypergraphH' \in \eqs {\setHypergraphsH}}.\]
}{
\[\setof{\weight{setX'}(\hypergraphH')}{\hypergraphH' \in {\setHypergraphsH},
\setX'\subseteq \Vertices{\hypergraphH'}, \size{\setX'} = \size{\setX}} = \setof{\weight{setX}(\hypergraphH')}{\hypergraphH' \in \eqs {\setHypergraphsH}}.\]
}{405}
\begin{thm}\label{thm:core}
The following conditions are equivalent,
 for a finite set ${\setHypergraphsH}$ of hypergraphs and a hypergraph $\hypergraphH$, all of the same arity and dimension:
 \begin{enumerate}
 \item $\hypergraphH$ is a $\Z$-sum of $\eqs{{\setHypergraphsH}}$;
 \item $\hypergraphH$ is locally a $\Z$-sum of ${\setHypergraphsH}$.
 \end{enumerate}
\end{thm}

\noindent
Before we embark on proving the result (in the next section) we first discuss how it
implies Theorem~\ref{thm:problem2}.
Recall that the arity ${\setSizeK}$ is fixed. 
Instead of checking if $\hypergraphH\changedd{ = (\setVertices, \mu)}{}{405}$ is a $\Z$-sum of $\eqs{{\setHypergraphsH}}$, the 
algorithm checks if
$\hypergraphH$ is \emph{locally} a $\Z$-sum of ${\setHypergraphsH}$. 
Observe that the condition~\eqref{eq:loc} amounts to solvability of a 
(classical) system of linear equations\changed{, as in~\eqref{eq:lineq}}{}{374}. \changedd{}{Let $\setVertices$ be the set of vertices of the 
hypergraph $\hypergraphH$.}{405}
Therefore, the algorithm tests $\Z$-solvability of a 
system of the corresponding 
$\dimension\cdot (1+\size{\setVertices}+\size{\kset{2}{\setVertices}}+\ldots \size{\kset{\setSizeK}{\setVertices}})$
linear equations,
$\dimension$ for every subset $\setX \subseteq \setVertices$ of the cardinality at most ${\setSizeK}$.
The number of equations is exponential in ${\setSizeK}$, but due to fixing ${\setSizeK}$
it is polynomial in the input hypergraphs $\hypergraphH$ and ${\setHypergraphsH}$.
Thus, Theorem~\ref{thm:problem2} is proved once we prove Theorem~\ref{thm:core}.

\begin{exa}
Let us \changed{recall}{continue}{379} Example~\ref{ex:two}. For the target we have that
$\weight{emptyset}(\changed{a}{\vec{v}}{379-381}_{\dc \dd})= \weight{dc} (\changed{a}{\vec{v}}{379-381}_{\dc \dd})= $
${\weight{dd}}(\changed{a}{\vec{v}}{379-381}_{\dc \dd})=$
${\weight{dc dd}}(\changed{a}{\vec{v}}{379-381}_{\dc \dd})= 6$, 
and for the triangle we have  
$\weight{emptyset}(\changed{a}{\vec{v}}{379-381}_{\dd \dc \de})=3$, $\weight{x}(\changed{a}{\vec{v}}{379-381}_{\dd \dc \de}) =2 $ for any $x\in \{\dd, \dc, \de\}$,
and $\weight{{ x, y}}(\changed{a}{\vec{v}}{379-381}_{\dd \dc \de})=1$ for any $\{x, y\}\subset \{\dd,\dc,\de\}$.

As $6 = 3 + 3$ and $6= 2+ 2+ 2$ and $6 = 1+ 1+ 1+1 +1 +1$ we see that
the target is locally a $\Z$-sum of the triangle. Hence \changed{}{the }{383}target is a 
\changed{$\Z-$sum}{$\Z$-sum}{383} of the triangle up to equivalence. Moreover, if we change 
$6$ to a smaller positive number then the target graph will not be 
a $\Z$-sum of the triangle up to equivalence. For example\changed{}{\ if, we change $6$ to}{NULL} $3$\changed{\ can not be 
expressed as a sum of twos.}{, then $\vec{v}_{\dc \dd}$ is not a $\Z$-sum of the \changed{triangle}{triangles}{NULL} as ${\weight{dd}}(\vec{v}_{\dc \dd})=3$  
is not divisible by $2$
i.e. the weight of any single vertex of the triangle. 
}{386}
\end{exa}

\section{Proof of Theorem~\ref{thm:core} (the case for \texorpdfstring{$\setSizeK=2$}{k=2}).}\label{sec:prove}

The implication $1 \implies 2$ is immediate. 
\changedd{Indeed, suppose $\hypergraphH = (\setVertices, \mu) = c_1 \cdot \hypergraphH_1 + \ldots + c_l \cdot \hypergraphH_l$, 
where $c_i \in \Z$ and
$\hypergraphH_i \changed{= (\setVertices_i, \mu_i)}{}{NULL} \in \eqs {\setHypergraphsH}$.
Let $\setX\subseteq \setVertices$ be a subset of cardinality $\size{\setX} \leq \setSizeK$.
We  have
\[
\weight{setX}(\hypergraphH) = c_1 \cdot \weight{setX}(\hypergraphH_1) + \ldots + c_l \cdot \weight{setX}(\hypergraphH_l),
\]
which implies,
for some hypergraphs $\hypergraphH'_i = \in {\setHypergraphsH}$ equivalent to 
$\hypergraphH_i$, and subsets
$\setY_i \subseteq \changed{\setVertices'_i}{\setVertices_i}{NULL}$ of cardinality $\size{\setY_i} = \size{\setX}$, 
the following equality holds
\[
\weight{setX}(\hypergraphH) = c_1 \cdot \weight{setY_1}(\hypergraphH_1') + \ldots + c_l \cdot \weight{setY_l}(\hypergraphH_l').
\]
As $\setX$ was chosen arbitrarily, this shows that $\hypergraphH$ is locally a $\Z$-sum of ${\setHypergraphsH}$.
}
{
Indeed, suppose $\hypergraphH = c_1 \cdot \hypergraphH_1 + \ldots + c_l \cdot \hypergraphH_l$, 
where $c_i \in \Z$ and
$\hypergraphH_i \changedd{\changed{= (\setVertices_i, \mu_i)}{}{NULL}}{}{405} \in \eqs {\setHypergraphsH}$.
Let $\setX\subseteq \Vertices{\hypergraphH}$ be a subset of cardinality $\size{\setX} \leq \setSizeK$.
We  have
\[
\weight{setX}(\hypergraphH) = c_1 \cdot \weight{setX}(\hypergraphH_1) + \ldots + c_l \cdot \weight{setX}(\hypergraphH_l),
\]
which implies,
for some hypergraphs $\hypergraphH'_i  \in {\setHypergraphsH}$ equivalent to 
$\hypergraphH_i$, and subsets
$\setY_i \subseteq \Vertices{\hypergraphH'_i}$ of cardinality $\size{\setY_i} = \size{\setX}$, 
the following equality holds
\[
\weight{setX}(\hypergraphH) = c_1 \cdot \weight{setY_1}(\hypergraphH_1') + \ldots + c_l \cdot \weight{setY_l}(\hypergraphH_l').
\]
As $\setX$ was chosen arbitrarily, this shows that $\hypergraphH$ is locally a $\Z$-sum of ${\setHypergraphsH}$.
}{405}

The proof of the converse implication $2 \implies 1$ is \changedd{much }{}{null}more involved.\changed{}{ The case for arity $\setSizeK =1$ is considered in 
\cite{DBLP:conf/lics/HofmanLT17}.}{398}
Here, for pedagogical reasons, we provide a simplified version of the proof for arity $\setSizeK =2$\changed{, i.e.,}{\ i.e.}{22} for (undirected) graphs. 
\changed{The case for arity $\setSizeK =1$ is considered in 
\cite{DBLP:conf/lics/HofmanLT17}.}{}{398}
Consequently, we speak of edges instead of hyperedges.
We recall that the graphs we consider are actually $\ktuple{\dimension}{\Z}$-\emph{weighted} graphs.

Let $\hypergraphH\changedd{ = (\setVertices , \mu)}{}{405}$ be a graph
and assume that $\hypergraphH$ is locally a $\Z$-sum of ${\setHypergraphsH}$.
We are going to demonstrate that $\hypergraphH$ is equivalent to a $\Z$-sum of $\eqs{{\setHypergraphsH}}$.
Since we work with arity $\setSizeK=2$, 
the assumption amounts to the following conditions, 
for every vertex $\da \in \changedd{\setVertices}{\Vertices{\hypergraphH}}{405}$ and every edge $\edge \in \he \hypergraphH$:
\begin{equation}
\label{eq:ass0}
\weight{emptyset}(\hypergraphH) \in \sums \Z {\setof{\weight{emptyset}(\hypergraphH')\;}{\hypergraphH' \changedd{= (\setVertices', \mu')}{}{405} \in 
{\setHypergraphsH}}} 
\end{equation}

\begin{equation}\label{eq:ass1}
\weight{da}(\hypergraphH) \in \sums \Z {\setof{\weight{da'}(\hypergraphH')}{ \hypergraphH' \changedd{= (\setVertices', \mu')}{}{405} \in {\setHypergraphsH}, 
\da' \in \changedd{\setVertices'}{\Vertices{\hypergraphH'}}{405}}} 
\end{equation}

\begin{equation}\label{eq:ass2}
\weight{edge}(\hypergraphH) \in \sums \Z {\setof{\weight{edge'}(\hypergraphH')\,}{ \hypergraphH' \changedd{= (\setVertices', \mu')}{}{405} \in 
{\setHypergraphsH}, 
\edge' \in \he {\hypergraphH'}}}.
\end{equation}

\begin{clm}\label{cl:zero_total_weight}
W.l.o.g.~we may assume that $\hypergraphH \changedd{= (\setVertices, \mu)}{}{405}$ satisfies $\weight{emptyset}(\hypergraphH) = {\vec{0}}$.
\end{clm}
\begin{proof}
Indeed, due to the assumption~\eqref{eq:ass0} and due the following equality
(note that the symbol $\sums \Z {\placeHolder}$ applies to vectors on the left,
and to graphs on the right)
\begin{equation}
\sums \Z {\setof{\weight{emptyset}(\hypergraphH')\;}{\hypergraphH' \in {\setHypergraphsH}}} = 
{\setof{\weight{emptyset}(\hypergraphH')\;}{\hypergraphH' \in \sums \Z {\setHypergraphsH}}},
\end{equation}
there is a graph $\hypergraphH' \changedd{= ({\setVertices}', \mu')}{}{405} \in \sums \Z {{\setHypergraphsH}}$
with $\weight{emptyset}(\hypergraphH') = \weight{emptyset}(\hypergraphH)$.
Therefore 
$\hypergraphH \in \sums \Z {\eqs{{\setHypergraphsH}}}$ if, and only if\changed{}{,}{325} $\hypergraphH-\hypergraphH' \in \sums \Z {\eqs{{\setHypergraphsH}}}$, and hence
we can replace $\hypergraphH$ by $\hypergraphH-\hypergraphH'$.
\changed{}{Note that, if $\hypergraphH$ is locally a ${\Z}$-sum of ${\setHypergraphsH}$ then $\hypergraphH-\hypergraphH'$ is also locally a ${\Z}$-sum of 
$\setHypergraphsH$.}{NULL}
\end{proof}

We proceed in two steps.
We start by defining a class of particularly simple graphs, called \emph{$\hypergraphH$-\simple} graphs,
and argue that $\hypergraphH$ is a $\Z$-sum of these graphs (Lemma~\ref{lem:SG} below).
Then we prove that every $\hypergraphH$-\simple\ graph is representable as a $\Z$-sum of ${\setHypergraphsH}$ up to equivalence
(Lemma~\ref{lem:SGH} below). 
\changed{Composing the two claims we get that $\hypergraphH$ is a $\Z$-sum of ${\setHypergraphsH}$ up to equivalence.}
{We may compose these two lemmas because of Lemma~\ref{lem:sumofupto}. By this we get that 
$\hypergraphH$ is a $\Z$-sum of ${\setHypergraphsH}$ up to equivalence.
}{419}

\changed{This structure of the proof is correct due to the following simple lemma:}{}{421}

\begin{lem}\label{lem:sumofupto}
For a family of ${\setSizeK}$-hypergraphs \changed{${\setHypergraphsH}$}{$\setG$}{442} and ${\setSizeK}$-hypergraphs \changed{$\hypergraphH, \hypergraphH_1\ldots \hypergraphH_l$}{$\hypergraphG, \hypergraphG_1\ldots \hypergraphG_l$}{442}, all of the 
same dimension,
If \changed{$\hypergraphH_1\ldots \hypergraphH_l \in \sums \Z { \eqs {\setHypergraphsH}}$}{$\hypergraphG_1\ldots \hypergraphG_l \in \sums \Z { \eqs {\setG}}$}{422}
and \changed{$\hypergraphH \in \sums \Z {\eqs{\{\hypergraphH_1\ldots 
\hypergraphH_l\}}}$}{$\hypergraphG \in \sums \Z {\eqs{\{\hypergraphG_1\ldots 
\hypergraphG_l\}}}$}{422} then
\changed{$\hypergraphH \in \sums \Z { \eqs {\setHypergraphsH}}$.}{$\hypergraphG \in \sums \Z { \eqs {\setG}}$.}{422}
\end{lem}
(Lemma~\ref{lem:sumofupto} is expressible, more succinctly, as 
\[
\sums \Z {\eqs{\sums \Z { \eqs {\changed{\setHypergraphsH}{\setG}{422}}}}} = \sums \Z { \eqs {\changed{\setHypergraphsH}{\setG}{422}}}.)
\]
\begin{proof}
 This is simply because of three trivial facts:
 \begin{enumerate}
 \item $\eqs{{\hypergraphG}_i+{\hypergraphG}_j}\subseteq \eqs{{\hypergraphG}_i} + \eqs{{\hypergraphG}_j}$ where ${\hypergraphG}_i, {\hypergraphG}_j$ are any 
two 
 ${\setSizeK}$-hypergraphs of the same dimension, and \changed{}{the }{429}second plus 
is \changed{}{the }{429}Minkowski sum.
\item $\sums \Z \familyF= \sums \Z {\sums \Z \familyF}$, for any $\familyF$ a family of ${\setSizeK}$-hypergraphs of the same 
dimension.
\item $\eqs \familyF= \eqs {\eqs \familyF}$, for any $\familyF$ a family of ${\setSizeK}$-hypergraphs of the same dimension.
 \end{enumerate}
 Now, 
 \begin{align*}
    \sums \Z {\eqs{\sums \Z { \eqs {\setHypergraphsH}}}} \subseteq  &&\text{ \changed{due to}{because of}{200} 1 }\\
    \sums \Z {\sums \Z {\eqs {\eqs {\setHypergraphsH}}}}= &&\text{ \changed{due to}{because of}{200} 2 and 3}\\
    \sums \Z {\eqs {\setHypergraphsH}}.&&
 \end{align*}
 
 The inclusion in the opposite direction is trivial.
\end{proof}

\begin{defi}\label{def:simple2graphs}
For every vector $\vec a \in \ktuple{\dimension}{\Z}$ we define an \emph{$\vec a$-edge simple graph} $\simpleGraph{2}{\vec a}$
(shown on the left) and an \emph{$\vec a$-vertex simple graph} $\simpleGraph{1}{\vec a}$ (shown on the right):

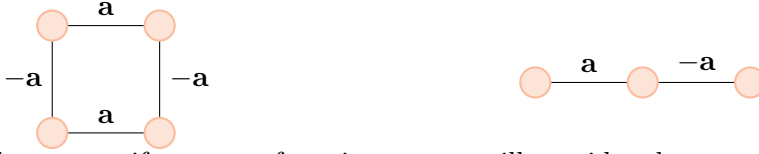
\begin{minipage}{0.45\linewidth}\label{edge_simple}
\EStikz{\vec a}{}{}{}{}
\end{minipage}
\begin{minipage}{0.45\linewidth}
\VStikz{\vec a}{}{}{}
\end{minipage}

\noindent
We do not specify names of vertices, as \changed{anyway }{}{443}we will consider these graphs up to equivalence.
\changed{Both types of graphs we call}{We call both types of graphs}{237} \emph{\simple\ graphs}.
\end{defi}
Let 
\begin{align}\label{eq:defsetWeight2}
\setWeight{2}{\hypergraphH} = \sums \Z {\setof{\weight{edge}(\hypergraphH)}{\edge \in \he{\hypergraphH}}} \subseteq \ktuple{\dimension}{\Z}
\end{align}
and 
\begin{align}\label{eq:defsetWeight1}
\setWeight{1}{\hypergraphH} = \sums \Z {\setof{\weight{da}(\hypergraphH)}{\da \in \changedd{{\setVertices}}{\Vertices{\hypergraphH}}{405}}} \subseteq 
\ktuple{\dimension}{\Z}.
\end{align}
Now, let $\setSimpleGraph{1} \hypergraphH \eqdef \setof{\simpleGraph{1}{\vec a}}{\vec a \in \setWeight{1}{\hypergraphH}}$ and 
$\setSimpleGraph{2}{\hypergraphH} \eqdef \setof{\simpleGraph{2}{\vec a}}{\vec a \in \setWeight{2}{\hypergraphH}}.$

\begin{lem}\label{lem:SG}
$\hypergraphH \in \sums \Z {\eqs{\setSimpleGraph{1}{\hypergraphH} \cup \setSimpleGraph{2}{\hypergraphH}}}$.
\end{lem}
\begin{proof} 
\changedd{ Suppose $\hypergraphH= ({\setVertices}, \mu)$.}{Let $\setVertices$ be the set of vertices of the hypergraph $\hypergraphH$.}{405}
The proof of the lemma is done in steps.

\begin{clm}\label{cl:grapha_one}
There is a ${\hypergraphG}\in \sums{\Z}{\eqs{\setSimpleGraph{1}{\hypergraphH}}}$ such that for any vertex $\da\in \setD$ \changed{holds}{it holds that}{451} 
$\weight{da}({\hypergraphG} + 
\hypergraphH) 
= {\vec{0}}$.
\end{clm}
We will use it to further simplify our problem, as $\hypergraphH \in \sums \Z {\eqs{\setSimpleGraph{1}{\hypergraphH} \cup \setSimpleGraph{2}{\hypergraphH}}}$ 
if\changed{}{,}{325} and only if\changed{}{,}{325} ${\hypergraphG} + \hypergraphH \in \sums \Z {\eqs{\setSimpleGraph{1}{\hypergraphH} \cup 
\setSimpleGraph{2}{\hypergraphH}}}$.

\begin{proof}[Proof of the claim.]
Let $\da, \da'$ \changed{are}{be}{456} two vertices not in ${\setVertices}$ and let $\familyF\subseteq 
\eqs{\setSimpleGraph{1}{\hypergraphH}}$ be the family of all 
$\simpleGraph{1}{\vec{a}}$ \simple\ graphs (depicted below) 
\begin{figure}[h!]
\begin{minipage}{0.45\linewidth}
\VStikz{\vec a}{\da}{\da'}{\db}
\end{minipage}
\end{figure}
where $\vec{a}= \weight{db}(\hypergraphH)$ for $\db \in {\setVertices}$. 

We define ${\hypergraphG}= \sum_{\hypergraphF\in \familyF} \hypergraphF$. As $\familyF\subset 
\eqs{\setSimpleGraph{1}{\hypergraphH}}$\changed{}{,}{461} the 
graph \changed{${\hypergraphG}\in \sums \Z {\eqs{\setSimpleGraph{1}{\hypergraphH}}}$}{${\hypergraphG}$ is in $\sums \Z 
{\eqs{\setSimpleGraph{1}{\hypergraphH}}}$}{461}. Moreover, 
\begin{enumerate}
 \item for every $\db \in {{\setVertices}}$ \changed{holds}{it holds that}{451} 
$-\weight{db}(\hypergraphH)=\weight{db}({\hypergraphG})$;
 \item $\weight{da}({\hypergraphG})={\vec{0}}$ as for each $\hypergraphF\in \familyF$\changed{}{,}{463} $\weight{da}(\hypergraphF) ={\vec{0}}$;
 \item $\weight{da'}({\hypergraphG})={\vec{0}}$ as $\weight{da'}({\hypergraphG})=\sum_{\db \in {\setVertices}} 
\weight{db}(\hypergraphH) = 
2\weight{emptyset}(\hypergraphH)={\vec{0}}$,
where the second equality reflects the fact that every edge has two ends and the last equality is due to Claim~\ref{cl:zero_total_weight}.
\end{enumerate}
As $\weight{x}$ is a homomorphism we get that $\weight{x}(\hypergraphH+{\hypergraphG}) ={\vec{0}}$ for every $x\in {\setVertices}\cup \{\da,\da'\}$. \changedd{The same 
holds for other vertices as they are isolated.}{}{null}
\end{proof}

\changed{Due to}{Because of}{200} the previous claim \changed{w.l.g}{w.l.o.g}{468}\changed{}{\ we may restrict our self to the following case.}{468-478} 
\begin{clm} We assume that $\hypergraphH$ has \changed{following 
properties}{the following property}{468} \changed{$\weight{db}(\hypergraphH)=\vec{0}$ 
for every $\db\in \setVertices$}{for every $\db\in \setVertices$, it holds that $\weight{db}(\hypergraphH)=\vec{0}$}{468-469}.
\end{clm}
\changed{Here}{Here,}{169,170} there is one issue that should be discussed. Suppose $\hypergraphH'= \hypergraphH + \hypergraphG$ as in \changed{the}{}{471} 
Claim~\ref{cl:grapha_one}. The issue is \changed{}{that}{471}
$\setSimpleGraph{1}{\hypergraphH} \cup \setSimpleGraph{2}{\hypergraphH} = \setSimpleGraph{1}{\hypergraphH'} \cup \setSimpleGraph{2}{\hypergraphH'}$ may not 
hold.
So if we prove the lemma for $\hypergraphH'$ and $\setSimpleGraph{1}{\hypergraphH'} \cup \setSimpleGraph{2}{\hypergraphH'}$ it does not necessarily carry 
\changed{on }{}{472}
to $\hypergraphH$ and $\setSimpleGraph{1}{\hypergraphH} \cup \setSimpleGraph{2}{\hypergraphH}$.
Fortunately, it is sufficient for us if $\sums{\Z}{\setSimpleGraph{1}{\hypergraphH'} \cup \setSimpleGraph{2}{\hypergraphH'}} \subseteq 
\sums{\Z}{\setSimpleGraph{1}{\hypergraphH} \cup 
\setSimpleGraph{2}{\hypergraphH}}$. The last inclusion holds. 
Indeed, for every $\db\in \changedd{\setVertices'}{\Vertices{\hypergraphH'}}{405}$ \changed{holds}{it holds that}{451} $\weight{db}(\hypergraphH')$ is a sum of 
weights of vertices in $\hypergraphH$ thus 
 $\sums{\Z}{\setSimpleGraph{1}{\hypergraphH'}} \subseteq 
 \sums{\Z}{\setSimpleGraph{1}{\hypergraphH}}$ and similarly $\sums{\Z}{\setSimpleGraph{2}{{\hypergraphH'}}} \subseteq 
\sums{\Z}{\setSimpleGraph{2}{\hypergraphH}}$.
Thus, we do not loose generality \changed{due to}{because of}{200} the proposed restriction.

\begin{clm}
There is a \changed{hypergraph}{graph}{480} ${\hypergraphG\changed{}{'}{NULL}}\in \sums \Z {\eqs{\setSimpleGraph{2}{\hypergraphH}}}$ 
 such that\changedd{ ${\hypergraphG\changed{}{'}{NULL}}=(\setVertices' ,\mu')$}{}{405}, $\hypergraphH + \hypergraphG\changed{}{'}{NULL}$ has at most $3$ 
\changed{non-isolated}{nonisolated}{107} vertices, and 
$\weight{db}({\hypergraphH} + {\hypergraphG\changed{}{'}{NULL}})={\vec{0}}$ \changed{for ever}{for every}{481} $\db\in\changedd{ {\setVertices}\cup 
\setVertices'}{\Vertices{\hypergraphH+\hypergraphG'}}{405}$.
\end{clm}

We will use it to further simplify our problem, as $\hypergraphH \in \sums {\Z}  {\eqs{\setSimpleGraph{1}{\hypergraphH} \cup 
\setSimpleGraph{2}{\hypergraphH}}}$ if\changed{}{,}{325} and only if\changed{}{,}{325} 
${\hypergraphH} + {\hypergraphG} \in \sums {\Z}  {\eqs{\setSimpleGraph{1}{\hypergraphH} \cup 
\setSimpleGraph{2}{\hypergraphH}}}$. 

\begin{proof}[Proof of the claim.]

\changed{We construct ${\hypergraphG}'$ gradually as a sum ${\hypergraphG}'=\sum_{i=1}^{last} {\hypergraphG}_i'$ in parallel with a sequence ${\hypergraphH_i}$ 
where $ {\hypergraphH_0}=  {\hypergraphH}$ and
$ {\hypergraphH_{i+1}}= {\hypergraphH_i} + {\hypergraphG}_{i+1}'$.}{
We construct ${\hypergraphG}'$ gradually as a sum of a sequence of graphs ${\hypergraphG}_1',{\hypergraphG}_2'\ldots{\hypergraphG}_{last}'$.
The sequence ${\hypergraphG}_i'$ is constructed in parallel with a sequence ${\hypergraphH_i}$ 
where $ {\hypergraphH_0}=  {\hypergraphH}$ and
$ {\hypergraphH_{i+1}}= {\hypergraphH_i} + {\hypergraphG}_{i+1}'$.}{486}

The main property of the sequence $ {\hypergraphH_i}$ is that 
$ {\hypergraphH_i}>  {\hypergraphH_{i+1}}$ in some \changed{well founded}{well-founded}{488} 
\changed{quasi order}{quasi-order}{488} on graphs. \changed{Let $\ALG$ be an algorithm}{Suppose that, there is $\ALG$ an algorithm}{486-494}, that takes as an 
input $ {\hypergraphH_{i-1}}$ and produces the next ${\hypergraphG}_i'\in 
\eqs{\setSimpleGraph{2}{\hypergraphH}}$. The 
precondition of the algorithm $\ALG$ is that 
$ {\hypergraphH_{i-1}}$ has more than $3$ \changed{non-isolated}{nonisolated}{107} vertices. 
\changed{Thus}{Now}{486-494}, due to the well-foundedness of the quasi order, the sequence ${\hypergraphG}_i'$ is finite. 
Further, the graph $ {\hypergraphH_{last}}$ has at most $3$ 
\changed{non-isolated}{nonisolated}{107} vertices \changed{due to}{because of}{200} the precondition of $\ALG$. 
To this end, we need to define the order on graphs and provide the algorithm $\ALG$.

\para{Order on graphs} We assume an arbitrary total order $<$ on vertices \changed{{\setVertices}}{${\setVertices}$}{495}.
We lift the order to an order on edges $\{\da, \db\}$. We define it as the lexicographic order on 
pairs $(\da, \db)$ satisfying $\da > \db$.	
Finally, the order is extended to a quasi-order on graphs: ${\hypergraphG} < {\hypergraphG}'$ iff
$\edge < \edge'$, where $\edge$ and $\edge'$ are the largest edges in ${\hypergraphG}$ and ${\hypergraphG}'$, respectively.

\para{The algorithm $\ALG$}
	\changed{Suppose,}{Suppose}{499} \changed{$(\da,\db)$}{$\{\da,\db\}$}{499} is the \changed{biggest}{largest}{499} edge in the graph $ {\hypergraphH_i}$ 
and that $ {\hypergraphH_i}$ has at least $4$ \changed{non-isolated}{nonisolated}{107} 
vertices. 
	Observe that 
	$\weight{da}( {\hypergraphH_i})={\vec{0}}$. Indeed, $\weight{da}( {\hypergraphH_i})=
	\weight{da}(\hypergraphH+\sum_{j=1}^{i}{\hypergraphG}_j')=\weight{da}(\hypergraphH)+\sum_{j=1}^{i}\weight{da}({\hypergraphG}_j')= 
	{\vec{0}} + \sum_{j=1}^{i}\weight{da}({\hypergraphG}_j')$ but each \changed{${\hypergraphG}_j'\in\eqs{\setSimpleGraph{2}{\hypergraphH}}$}{${\hypergraphG}_j'$ belongs to $\eqs{\setSimpleGraph{2}{\hypergraphH}}$}{503} and
	by the definition of edge simple graphs $\weight{de}(\setSimpleGraph{2}{\hypergraphH})={\vec{0}}$ for every vertex $\de$. 
	As a consequence there must be at least one vertex $\dc\not\in \{\da,\db\}$ such 
	that $(\da,\dc)$ is an edge in $ {\hypergraphH_i}$. As $ {\hypergraphH_i}$ has at least $4$ \changed{non-isolated}{nonisolated}{107} vertices, 
	there is $\dd\not\in \{\dc,\da,\db\}$ a \changed{non-isolated}{nonisolated}{107} vertex in 
	$ {\hypergraphH_i}$. We define \changed{${\hypergraphG}_{i+1}$}{${\hypergraphG}_{i+1}'$}{509} as follows 
	
	\begin{minipage}{0.45\linewidth}
\EStikz{\vec{a}}{\da}{\db}{\dc}{\dd} 
\end{minipage}
	for $\vec{a} = \weight{{da,db}}( {\hypergraphH_i}).$

	Observe that edges $(\da,\dc), (\db,\dd), (\dc,\dd)$ are smaller than the edge $(\da,\db)$ and \changed{$({\hypergraphH_i} + 
{\hypergraphG}_{i+1})(\{\da,\db\})=\vec{0}$}{$({\hypergraphH_i} + 
{\hypergraphG'}_{i+1})(\{\da,\db\})=\vec{0}$}{512} so 
	$ {\hypergraphH_i} + \changed{{\hypergraphG}_{i+1}}{{\hypergraphG'}_{i+1}}{512} <  {\hypergraphH_i}$.
\end{proof}

After usage of the claim we get a graph $\hypergraphH'\eqdef \hypergraphH + \changed{\sum_{i=0}^{last} {\hypergraphG}_i'}{{\hypergraphG}'}{514}$. We know that 
$\hypergraphH\in \sums \Z {\eqs{\setSimpleGraph{1}{\hypergraphH}\cup \setSimpleGraph{2}{\hypergraphH}}}$ if\changed{}{,}{325} 
and 
only if\changed{}{,}{325}
$ \hypergraphH'\in \sums \Z {\eqs{\setSimpleGraph{1}{\hypergraphH}\cup \setSimpleGraph{2}{\hypergraphH}}}$.
The graph $\hypergraphH'$ has at most $3$ \changed{non-isolated}{nonisolated}{107} \changed{nodes}{vertices}{NULL}, $\weight{emptyset}(\hypergraphH') 
={\vec{0}},$ and
$\weight{da}(\hypergraphH')={\vec{0}}$ for any vertex $\da$.

Suppose that the set of \changed{non-isolated}{nonisolated}{107} vertices of $\hypergraphH'$ is a subset of $\{\da,\db,\dc\}$. We may write the following system 
of equations:
 \begin{align}
 \begin{split}
 \weight{da}(\hypergraphH') = \weight{{da,db}}(\hypergraphH') + \weight{{da,dc}}(\hypergraphH') = {\vec{0}}\\
 \weight{db}(\hypergraphH') = \weight{{da,db}}(\hypergraphH') + \weight{{db,dc}}(\hypergraphH') = {\vec{0}}\\
 \weight{dc}(\hypergraphH') = \weight{{da,dc}}(\hypergraphH') + \weight{{db,dc}}(\hypergraphH') = {\vec{0}}
 \end{split}
 \label{eq:unique_solutions}
 \end{align}
The only solution is $\weight{{da,db}}(\hypergraphH') = \weight{{db,dc}}(\hypergraphH') = \weight{{da,dc}}(\hypergraphH') = {\vec{0}}$. So $\hypergraphH'$ 
is the empty graph.
In consequence, $\hypergraphH = - \sum_{i=1}^{last}{\hypergraphG}_i \in \sums \Z {\eqs{\setSimpleGraph{1}{\hypergraphH}\cup 
\setSimpleGraph{2}{\hypergraphH}}}$, which completes
the proof of Lemma~\ref{lem:SG}.
\end{proof}

\section*{Representation of \texorpdfstring{$\hypergraphH$}{H}-\simple\ graphs}\label{subsec:simple_graps} 

As the second step\changed{}{,}{527} we prove that every $\hypergraphH$-\simple\ \changed{graphs}{graph}{528} is
representable as a $\Z$-sum of ${\setHypergraphsH}$ up to equivalence (Lemma~\ref{lem:SGH}). We start \changed{from}{with}{529} two preparatory lemmas.

\begin{lem}\label{lem:exists_edge}
Let $\edge\in \he{{\hypergraphG}}$ for a graph ${\hypergraphG}\changedd{ = ({\setVertices\changed{}{'}{NULL}}, \mu\changed{}{'}{NULL})}{}{405}$, and let 
$\weight{edge}({\hypergraphG}) = \vec a$.
Then $\simpleGraph{2}{\vec a} \in \sums \Z {\eqs{\{\hypergraphG\}}}$.
\end{lem}
\begin{proof}
\changedd{}{Let $\setVertices'$ be the set of vertices of the graph $\hypergraphG'$.}{405}
\changedd{Let}{Suppose}{405} $\edge = \{\da, \db\} \subseteq {\setVertices\changed{}{'}{NULL}}$, $\weight{edge}({\hypergraphG}) = \vec a$, and let $\dc, \dd 
\notin 
{\setVertices\changed{}{'}{NULL}}$ \changed{are}{be}{532} two additional 
fresh
vertices outside of ${\setVertices\changed{}{'}{NULL}}$. We consider three graphs ${\hypergraphG}_1, {\hypergraphG}_2, {\hypergraphG}_{12}$ which differ from 
${\hypergraphG}$
only by:
\begin{itemize}
\item replacing $\da$ with $\dc$ (in case of ${\hypergraphG}_1$),
\item replacing $\db$ with $\dd$ (in case of ${\hypergraphG}_2$),
\item replacing both $\da$ and $\db$ with $\dc$ and $\dd$, respectively (in case of ${\hypergraphG}_{12}$).
\end{itemize}

\noindent
Clearly all the three graphs are equivalent to ${\hypergraphG}$.
We claim that ${\hypergraphG} - {\hypergraphG}_1 - {\hypergraphG}_2 + {\hypergraphG}_{12}$ yields:

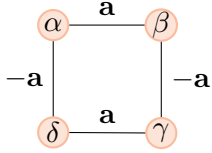
\EStikz{\vec a}{\da}{\dd}{\db}{\dc}

\noindent
that is a graph equivalent to $\simpleGraph{2}{\vec a}$.
Indeed, the above graph operations cancel out 
all edges \changed{non-incident}{nonadjacent}{107} to the four vertices $\da, \db, \dc, \dd$,
as well as all edges \changed{incident}{adjacent}{NULL} to only one of them.
The two horizontal $\vec a$-weighted edges originate from ${\hypergraphG}$ and ${\hypergraphG}_{12}$,
while the two remaining vertical ones originate from $-{\hypergraphG}_1$ and $-{\hypergraphG}_2$.
\end{proof}

\begin{lem}\label{lem:exists_vertex}
Let $\da\changedd{\in {\setVertices\changed{}{'}{NULL}}}{}{405}$ be a vertex of a graph ${\hypergraphG}\changedd{ = ({\setVertices\changed{}{'}{NULL}}, 
\mu\changed{}{'}{NULL})}{}{405}$, and let 
$\weight{da}({\hypergraphG}) = \vec a$.
Then $\simpleGraph{1}{\vec a} \in \sums \Z {\eqs{\{\hypergraphG\}}}$.
\end{lem}

\begin{proof}
\changedd{}{Let $\setVertices'$ be a set of vertices of the graph $\hypergraphG$.}{405} \changedd{Let}{Suppose}{405} a graph ${\hypergraphG}'$ differ from 
${\hypergraphG}$ only by
replacing $\da$ with a fresh vertex $\db \notin {\setVertices\changed{}{'}{NULL}}$.
The graph ${\hypergraphG} - {\hypergraphG}'$ has the following shape (with blue square vertices representing
the set ${\setVertices\changed{}{'}{NULL}} \setminus \{\da\} = \{\dc_1, \dc_2\ldots \dc_n\}$):

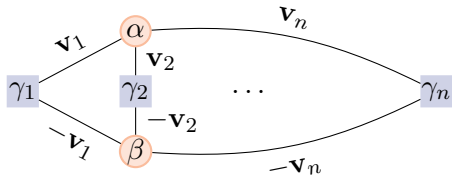
\begin{tikzpicture}
[place/.style={circle,draw=\mycolorOne!50,fill=\mycolorOne!20,thick,inner sep=0pt,minimum size=4mm},
stplace/.style={rectangle,draw=\mycolorTrzy!20,fill=\mycolorTrzy!20,thick,inner sep=0pt,minimum size=4mm}
]
\node[place]    (V)     at ( 0,0.3)    {$\da$};
\node[place]      (V1)       at (0,-1.3)  {$\db$}; 
\node[stplace]      (A1)      at ( -1.5,-0.5)   {$\dc_1$};
\node[stplace]      (A11)     at ( 0,-0.5)  {$\dc_2$};
\node[]      (A111d2)    at ( 1.5,-0.5)  {$\ldots$};
\node[stplace]      (A1111)    at ( 4,-0.5)  {$\dc_n$};

\draw (V) to node [above, sloped] {$\vec v_1$} (A1) to node [below, sloped] {$-\vec v_1$} (V1) ;
\draw (V) to node [right] {$\vec v_2$} (A11) to node [right] {$-\vec v_2$} (V1) ;
\draw [bend left=15] (V) to node [ above, sloped] {$\vec v_n$} (A1111);
\draw [bend left=15] (A1111) to node [below, sloped] {$-\vec v_n$} (V1) ;
\end{tikzpicture}

\noindent
Relying on Lemma~\ref{lem:exists_edge}, we add to the graph ${\hypergraphG} - {\hypergraphG}'$ the following graphs equivalent to
$\simpleGraph{2}{\vec v_i}$, for $i = 1, 2\ldots n-1$:

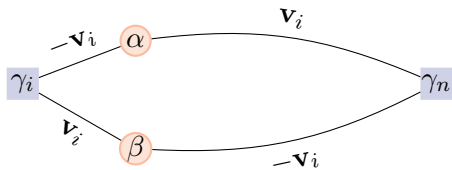
\begin{tikzpicture}
[place/.style={circle,draw=\mycolorOne!50,fill=\mycolorOne!20,thick,inner sep=0pt,minimum size=4mm},
stplace/.style={rectangle,draw=\mycolorTrzy!20,fill=\mycolorTrzy!20,thick,inner sep=0pt,minimum size=4mm}
]
\node[place]    (Vd)         [below= of A11] {$\da$};
\node[place]      (V1d)       [below=of Vd]  {$\db$};
\node[stplace]      (A1d)      at (-1.5,-2.5)   {$\dc_i$};
\node[stplace]      (A1111d)    at (4,-2.5)  {$\dc_n$};

\draw (Vd) to node [above, sloped] {$-\vec v_i$} (A1d) to node [below, sloped] {$\vec v_i$} (V1d) ;
\draw [bend left=15] (Vd) to node [above, sloped] {$\vec v_i$} (A1111d);
\draw [bend left=15] (A1111d) to node [below, sloped] {$-\vec v_i$} (V1d) ;
\end{tikzpicture}

\noindent
This results in collapsing all blue vertices into one:

\begin{tikzpicture}
[place/.style={circle,draw=\mycolorOne!50,fill=\mycolorOne!20,thick,inner sep=0pt,minimum size=4mm},
stplace/.style={rectangle,draw=\mycolorTrzy!20,fill=\mycolorTrzy!20,thick,inner sep=0pt,minimum size=4mm}
]
\node[place]    (V)         {$\da$};
\node[place]    (V1)   at (7,0)    {$\db$};
\node[stplace]      (A1)      at ( 3.5,0)   {$\dc_n$};
\draw (V) to node [above] {$\vec v_1+\ldots+\vec v_n$} (A1) to node [above] {$-(\vec v_1+\ldots+\vec v_n)$} (V1) ;
\end{tikzpicture}

\noindent
Since $\vec a = \vec v_1 + \ldots + \vec v_n$, the resulting graph is equivalent to $\simpleGraph{1}{\vec a}$, as required.
\end{proof}

\begin{lem}\label{lem:SGH}
\changedd{Every}{Suppose $\hypergraphH$ is locally a $\Z$-sum of ${\setHypergraphsH}$, then every}{null}
 $\hypergraphH$-\simple\ graph is a $\Z$-sum of ${\setHypergraphsH}$ up to equivalence:
$\setSimpleGraph{1} \hypergraphH \subseteq \sums \Z {\eqs{{\setHypergraphsH}\}}}$ and $\setSimpleGraph{2} \hypergraphH \subseteq \sums \Z 
{\eqs{{\setHypergraphsH}\}}}$.
\end{lem}
\begin{proof}
\changedd{Let $\hypergraphH = ({\setVertices}, \mu)$.}{}{405} We have to prove that for any $\vec{a}\in 
\changed{\setWeight{2}{\hypergraphH}}{\setWeight{1}{\hypergraphH}}{NULL}$ 
\changed{}{($\setWeight{1}{\hypergraphH}$ is defined in Equation~\ref{eq:defsetWeight1})}{571} \changed{holds}{it holds that}{451} $\simpleGraph{1}{\vec{a}}\in 
\sums{\Z}{\eqs{\setHypergraphsH}}$ and that for any $\vec{b}\in \changed{\setWeight{1}{\hypergraphH}}{\setWeight{2}{\hypergraphH}}{NULL}$ 
\changed{}{($\setWeight{2}{\hypergraphH}$ is 
defined in Equation~\ref{eq:defsetWeight2})}{571} \changed{holds}{it holds that}{451} $\simpleGraph{2}{\vec{b}}\in 
\sums{\Z}{\eqs{\setHypergraphsH}}$. We prove only $\simpleGraph{1}{\vec{a}}\in 
\sums{\Z}{\eqs{\setHypergraphsH}}$ as the second proof is the almost same. 
\changed{Due to}{Because of}{200} Lemma~\ref{lem:sumofupto} we know that 
\[
\sums{\Z}{\eqs{\sums{\Z}{\eqs{\setHypergraphsH}}}}=\sums{\Z}{\eqs{\setHypergraphsH}}.
\]
Thus, \changed{due to}{because of}{200} Lemma~\ref{lem:exists_vertex} it is sufficient to prove that there is a graph $\hypergraphG_{\vec{a}}\in \sums{\Z}{\eqs{\setHypergraphsH}}$ such 
that $\weight{da}({\hypergraphG}_{\vec{a}})=\vec{a}$ \changedd{}{for some $\da\in \Vertices{\hypergraphG_{\vec{a}}}$.}{null} 

As $\vec a\in \changed{\setWeight{2}{\hypergraphH}}{\setWeight{1}{\hypergraphH}}{NULL}$ we know \changed{}{that}{581} 
\begin{align}\label{eq:a}
\vec a = z_1 \cdot \vec a_1 + \ldots + z_l \cdot \vec a_l,
\end{align}
where $z_i\in \Z$ and $\vec{a_i}$ is $\weight{db_i}(\hypergraphH)$ for some $\db_i\in \changedd{{\setVertices}}{\Vertices{\hypergraphH}}{405}$.
By the assumption~\eqref{eq:ass1}, for every $i$ we know 
\[
\vec a_i \in \sums \Z {\setof{\weight{dc_i}(\hypergraphH')}{\hypergraphH'\changedd{ = ({\setVertices}', \mu')}{}{405} \in {\setHypergraphsH}, \dc_i \in 
\changedd{{\setVertices}'}{\Vertices{\hypergraphH'}}{405}}}.
\]
we may use the equivalence relation and rewrite it as follows
\[
\vec a_i \in \sums \Z {\setof{\weight{da}(\hypergraphH')}{\hypergraphH'\changedd{ = ({\setVertices}', \mu') }{}{405}\in \eqs{{\setHypergraphsH}}}}.\\
\]
\text{ \changed{We concretize it}{Thus}{585}}
\[
\vec a_i = z_{i,1} \weight{da}(\hypergraphH'_{i,1}) + z_{i,2} \weight{da}(\hypergraphH'_{i,2}) +\ldots + z_{i,h(i)} \weight{da}(\hypergraphH'_{i,\changed{h}{h(i)}{586}}) \text{ and 
further }
\]
\[
\vec{a} = \sum_{i=1}^{l} z_i \left( z_{i,1} \weight{da}(\hypergraphH'_{i,1}) + z_{i,2} \weight{da}(\hypergraphH'_{i,2}) +\ldots + z_{i, h(i)} 
\weight{da}(\hypergraphH'_{i,h(i)}) \right) 
\]
\text{ $\weight{da}$ is a homomorphism, so}\\
\[
\vec{a} = \weight{da}\left(\sum_{i=1}^{l} z_i \left( z_{i,1} \hypergraphH'_{i,1} + z_{i,2} \hypergraphH'_{i,2} +\ldots + z_{i, h(i)} \hypergraphH'_{i,h(i)} 
\right)\right) 
\]
\text{ so we  define }\\
\[
{\hypergraphG}_{\vec{a}} \eqdef \sum_{i=1}^{l} z_i \left( z_{i,1} \hypergraphH'_{i,1} + z_{i,2} \hypergraphH'_{i,2} +\ldots + z_{i, h(i)} 
\hypergraphH'_{i,h(i)}\right). 
\]
The construction of elements of $\setSimpleGraph{2}{\hypergraphH}$ relays on Lemma~\ref{lem:exists_edge} instead of Lemma~\ref{lem:exists_vertex}.
\end{proof}

Combining Lemmas~\ref{lem:SG}~and~\ref{lem:SGH} we know that $\hypergraphH$ is
a $\Z$-sum of $\hypergraphH$-\simple\ graphs, each of which in turn is a $\Z$-sum of ${\setHypergraphsH}$ up to equivalence. 
By Lemma~\ref{lem:sumofupto}, $\hypergraphH$ is a $\Z$-sum of ${\setHypergraphsH}$ up to equivalence, as required.
The proof of Theorem~\ref{thm:core}, for graphs, is thus completed.

\section{Proof of Theorem~\ref{thm:core} (the outline).}\label{sec:Proof:thm:core}
Now\changed{}{,}{NULL} we are ready to prove the implication $2 \implies 1$ from Theorem~\ref{thm:core} in full generality.
The proof is split into five parts. We want to mimic the approach presented in Section~\ref{sec:prove}, thus we need to generalise \changed{}{a }{608}few things:
\begin{enumerate}
 \item In Section~\ref{sec:matrix} we build an algebraic background to be able to solve systems of equations that in the general case correspond to 
Equations~\ref{eq:unique_solutions}.
\item In Section~\ref{sec:simple_hypergraphs} we introduce simple hypergraphs that generalise graphs defined in Definition~\ref{def:simple2graphs}.
\item In~Section~\ref{sec:Expressing with 
simple hypergraphs} we show that our target hypergraph is a $\Z$-sum of simple 
hypergraphs. \changed{}{This generalises Lemma~\ref{lem:SG}.}{NULL}
\item In Section~\ref{sec:The construction of simple hypergraphs} we prove that simple hypergraphs are $\Z$-sums of 
$\eqs{{\setHypergraphsH}}$. \changed{}{This generalises Subsection~Representation of $\hypergraphH$-\simple\ graphs i.e. 
Lemmas~\ref{lem:exists_edge},~\ref{lem:exists_vertex},~and~\ref{lem:SGH}.}{NULL}
\item Finally, in Section~\ref{sec:core_proof} we complete the proof of~Theorem~\ref{thm:core}.
\end{enumerate}

\section{Reduction matrices.}\label{sec:matrix}

In this section we introduce a notion of reduction matrices and 
prove a key 
lemma about their rank, Lemma~\ref{lem:main_matrix}. They are 
$0,1$ matrices related to adjacency matrices of Kneser graphs. We recall that by \changed{$x-$set}{$x$-set}{619} we mean a set with $x$ elements.

\begin{defi}%[Reduction Matrix]
\label{def:reduction matrix}
Let $a\in \N$ and $\setA$ be an $a$-set.
\changed{Matrix is a \emph{reduction matrix} for $a \geq b \geq c$, $a,b,c \in \mathbb{N}$, denoted by $\reductionMatrix{a}{b}{c}$, if}{
For $a \geq b \geq c$, $a,b,c \in \mathbb{N}$, we define a matrix $\reductionMatrix{a}{b}{c}$ as follows:
}{620}
\begin{enumerate}
    \item \changed{}{it has $\binom{a}{b}$ columns and $\binom{a}{c}$ rows, }{622}
    \item columns and rows are indexed with $b$-element subsets of $\setA$ and $c$-element subsets of $\setA$, respectively,
    \item $\reductionMatrix{a}{b}{c}[\setC,\setB]=1$ if $\setC\subseteq \setB$ and $0$ otherwise, 
    where $\setC$ is an index of a row and $\setB$ is an index of a column.
\end{enumerate}
\changed{}{The definition is up to reordering of rows and columns.}{622}
\end{defi}

\begin{exa}
\changed{}{Suppose $\setA=\{\da,\db,\dc,\dd\}$.}{626}
\begin{enumerate}
    \item
    $\begin{bmatrix}
        \reductionMatrix{4}{3}{1}
    \end{bmatrix}
    =
    \begin{bmatrix}
	&\{\dc,\db,\da\} &\{\dd,\db,\da\} &\{\dd,\dc,\da\} & \{\dd,\dc,\db\} \\
        \{\da\} & 1 & 1 & 1 & 0 \\
        \{\db\} & 1 & 1 & 0 & 1 \\
        \{\dc\} & 1 & 0 & 1 & 1 \\
        \{\dd\} & 0 & 1 & 1 & 1
    \end{bmatrix}$ \changed{}{\begin{minipage}{4.5cm}
 up to a permutation of rows and columns.
                              \end{minipage}}{628}
    
    \item
    $\begin{bmatrix}
        \reductionMatrix{4}{2}{1}
    \end{bmatrix}
    =
    \begin{bmatrix}
        & \{\db,\da\} &\{\dc,\da\}&\{\dc,\db\}&\{\dd,\da\}&\{\dd,\db\}&\{\dd,\dc\}\\
       \{\da\}& 1 & 1 & 1 & 0 & 0 & 0 \\
       \{\db\}& 1 & 0 & 0 & 1 & 1 & 0 \\
       \{\dc\}& 0 & 1 & 0 & 1 & 0 & 1 \\
       \{\dd\}& 0 & 0 & 1 & 0 & 1 & 1 
    \end{bmatrix}$\changed{}{\begin{minipage}{3.5cm}
                               up to a permutation of rows and columns.
                             \end{minipage}}{628}
\end{enumerate}
\end{exa}

The below lemma expresses the important property of reduction matrices.
\begin{lem}%{lemMainMatrix}
\label{lem:main_matrix}
Any $\reductionMatrix{2{\setSizeK}+1}{{\setSizeK}+1}{{\setSizeK}}$ reduction matrix has maximal rank.
\end{lem}

The proof relies on \changed{the}{a}{640} result from the spectral theory of Kneser graphs.

\begin{defi}\changed{{\emph Kneser graph.}}{}{NULL}
Let $\setA$ be an $a$-set.
The \changed{Kneser graph}{\emph{ Kneser graph}}{641} $K_{a,c}$ is the graph whose vertices are $c$-element subsets of $\setA$, 
and where two vertices $x,y$ are adjacent if\changed{}{,}{325} and only if\changed{}{,}{325} $x\cap y=\emptyset$. 
\end{defi}

The \emph{\changed{incidence}{adjacency}{NULL} matrix} \changed{}{of }{644}a graph $({\setX},E)$ is a ${\setX}\times {\setX}$-matrix $M$ such that
$M[x,y]=1$ if \changed{$x,y$ are adjacent}{there is an edge between $x$ and $y$}{645} and $0$ otherwise.

\begin{thmC}[{\cite[\changed{(page. 200)}{}{646}Theorem 9.4.3]{Godsil}}]
All eigenvalues of the \changed{incidence}{adjacency}{NULL} matrix for a Kneser graph are \changed{non-zero}{nonzero}{107}.
\end{thmC}
\begin{cor}\label{cor:Kneser}
The rank of the \changed{incidence}{adjacency}{NULL} matrix for a Kneser graph is maximal.\qed
\end{cor}

\begin{proof}[Proof of Lemma~\ref{lem:main_matrix}]
We relabel columns of $\reductionMatrix{2{\setSizeK}+1}{{\setSizeK}+1}{{\setSizeK}}$ in the following way. If a column is labelled with
the set $\setB\subset \setA$ we relabel it to $\setA\setminus \setB$. Now\changed{}{,}{NULL} both rows and columns are labelled with
${\setSizeK}$-subsets of $\setA$.  \changed{}{We denote}{NULL} the relabelled matrix \changed{}{by}{NULL} $M'$.
On the one hand, observe that if a ${\setSizeK}$-set $\setC$ is a subset of a \changed{$\setSizeK+1$-}{$(\setSizeK+1)$-}{779}set $\setB$ then
$\setC\cap (\setA\setminus \setB)=\emptyset$, in which case $M'[\setC, \setA\setminus \setB]=1$. 
On the other hand if a ${\setSizeK}$-set $\setC$ is not a subset of a \changed{$\setSizeK+1$-}{$(\setSizeK+1)$-}{779}set $\setB$ then
$\setC\cap (\setA\setminus \setB)\neq \emptyset$, and $M'[ \setC, \setA\setminus \setB]=0$.
So we see that $M'$ is the \changed{incidence}{adjacency}{NULL} matrix of 
the Kneser graph $K_{2{\setSizeK}+1,{\setSizeK}}$. But \changed{due to}{, because of}{200} Corollary~\ref{cor:Kneser}\changed{}{,}{655} it has maximal rank, and the same holds for $M$
as relabelling of columns does not change the rank of the matrix. 
\end{proof}

\section{Simple hypergraphs.}\label{sec:simple_hypergraphs}

\changed{}{
\begin{defi}\label{def:m isiolation}
Let ${\hypergraphH}\changedd{=(\setVertices,\mu)}{}{405}$ be a ${\setSizeK}$-hypergraph\changed{}{ and $\setVertices$ be the set of its vertices}{405}. We say 
that it is \emph{$m$-isolated} 
if for any subset $\setX \subseteq {\setVertices}$ such that $\size{\setX}\leq m$ \changed{holds}{it holds that}{451} 
$\weight{setX}({\hypergraphH})=\vec{0}$. 
\end{defi}}{666}

\begin{defi}\label{def:simple_hypergraphs}
Let \changedd{}{$\hypergraphG$ be a ${\setSizeK}$-hypergraph and $\setVertices'$ be its set of vertices. Suppose }{405}$0\leq m \leq \setSizeK$ \changed{}{and 
$\vec{a}\in \ktuple{\dimension}{\Z}$}{658}. We call \changedd{a ${\setSizeK}$-hypergraph 
${\hypergraphG}=({\setVertices\changed{}{'}{NULL}},\mu\changed{}{'}{NULL})$}{$\hypergraphG$}{405} \emph{$(m, \vec{a})$-\simple} 
if there exist pairwise disjoint sets $\setA, \setB, \setC$ such that
\begin{enumerate}
 \item ${\setVertices\changed{}{'}{NULL}} = \setA \cup \setB \cup \setC$.
 \item $\setA = \{\da_1, \da_2 \ldots \da_m\}$ and $\setB = \{\db_1, \db_2 \ldots \db_m\}$ are $m$-sets.
 \item $\size{\setC} = 2({\setSizeK}-m)-1$ or $\setC=\emptyset$ if $m=\setSizeK$.
 \item $\forall \setX\subseteq \setA\cup \setB$, such that $\setX = \{x_1, x_2 \ldots x_m\}$, 
 where $x_i \in \{\da_i,\db_i\}$, the equality $\weight{setX}({\hypergraphG})= (-1)^{\size{\setB \cap \setX}}\cdot \vec{a}\ $ holds.
 \item For any other $\setX \subseteq {\setVertices\changed{}{'}{NULL}}$ such that $\size{\setX} = m$ the equality $\weight{setX}({\hypergraphG}) = \vec{0}$ 
holds.
 \changed{\item $\forall{\setX\subseteq {\setVertices}}$, such that $\size{\setX}<m$ the equality $\weight{setX}({\hypergraphG})=\vec{0}$ holds.}
 {\item $\hypergraphG$ is $(m-1)$-isolated.}{666}
\end{enumerate}
\end{defi}

This definition \changed{looks horrible}{is hard}{667} so we analyse it using examples, and explain the required properties.
\changed{}{Also, observe that contrary to $\simpleGraph{1}{\vec a}$ and $\simpleGraph{2}{\vec a}$ $(m, \vec{a})$-\simple\ hypergraphs 
are not fully specified up to isomorphism. Further, it is not clear that they exist for all $(m, \vec{a})$. In the example below we comment on this as 
well.}{667}

\begin{exa}
From the left to the right: a \changed{$(0,\vec{a})-\simple$}{$(0,\vec{a})$-\simple}{20} $2$-hypergraph ${\hypergraphG}_0$, where 
$\vec{a}=\vec{x}+\vec{y}+\vec{z}$;
 a \changed{$(1,\vec{a})-\simple$}{$(1,\vec{a})$-\simple}{20} $2$-hypergraph ${\hypergraphG}_1$;
 a \changed{$(2,\vec{a})-\simple$}{$(2,\vec{a})$-\simple}{20} $2$-hypergraph ${\hypergraphG}_2.$ The sets $\red{\setA},\blue{\setB},\green{\setC}$ are marked 
with colours.    

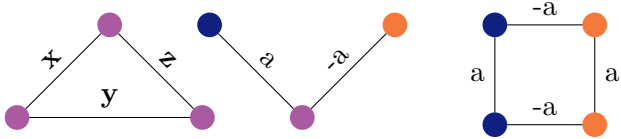
\begin{tikzpicture}
[place/.style={circle,draw=\mycolorOne!50,fill=\mycolorOne!20,thick,inner sep=0pt,minimum size=3mm}]
\node[place, color=\mycolorDwa] (A) {};
\node[place, color=\mycolorDwa] (B) [below left=of A] {};
\node[place, color=\mycolorDwa] (C) [below right=of A] {};
\draw (A) to node [above, sloped] {$\vec{x}$} (B) to node [above] {$\vec{y}$} (C) to node [above, sloped] {$\vec{z}$} (A);

\node[place, color=\mycolorDwa] (A1) [right =of C] {};
\node[place, color=\mycolorTrzy] (B1) [above left=of A1] {};
\node[place, color=\mycolorOne] (C1) [above right=of A1] {};
\draw (B1) to node [above, sloped] {a} (A1) to node [above, sloped] {-a} (C1) ;

\node[place, color=\mycolorTrzy] (A2) [right =of C1] {};
\node[place, color=\mycolorTrzy] (B2) [below =of A2] {};
\node[place, color=\mycolorOne] (C2) [right =of A2] {};
\node[place, color=\mycolorOne] (D2) [below =of C2] {};

\draw (A2) to node [left] {a} (B2) to node [above] {-a} (D2) to node [right] {a} (C2) to node [above] {-a} (A2);
\end{tikzpicture} 

\begin{enumerate}
 \item If $m=0$ then $\setA, \setB$ are empty, \changed{condition}{Property}{685-688} 2 states that $\setC$ has $3$ vertices, 
 \changed{condition}{Property}{685-688} $4$ that 
$\weight{emptyset}({\hypergraphG}_0)=\vec{a}$, 
\changed{conditions}{Properties}{685-688} $5$ and 
$6$ 
are 
\changed{empty}{trivial}{670}.
 \item If $m=1$ then $\size{\setA}=\size{\setB}=\size{\setC}=1$, \changed{conditions}{Properties}{685-688} $4$ and $5$ provide 3 equations 
 \begin{align}
\vec{a}=\weight{red{bullet}}({\hypergraphG}_1)=\weight{{red{bullet},green{bullet}}}({\hypergraphG}_1)+\weight{{red{bullet},blue{bullet}}}({
\hypergraphG}_1) \\ 
-\vec{a}=\weight{blue{bullet}}({\hypergraphG}_1)=\weight{{blue{\bullet},green{bullet}}}({\hypergraphG}_1)+\weight{{red{bullet},blue{bullet}}}
({\hypergraphG}_1) \\
\vec{0}=\weight{green{bullet}}({\hypergraphG}_1)=\weight{{blue{\bullet},green{bullet}}}({\hypergraphG}_1)+\weight{{red{bullet},green{bullet}}}
({\hypergraphG}_1). 
  \end{align}
  It is not hard to derive from them that $\vec{a}=\weight{{red{bullet},green{bullet}}}({\hypergraphG}_1)$, $-\vec{a} = 
\weight{{blue{\bullet},green{bullet}}}({\hypergraphG}_1)$, 
and $\vec{0}= 
\weight{{red{bullet},blue{bullet}}}({\hypergraphG}_1)$.
\changed{The condition}{Property}{685-688} $6$ states that $\weight{emptyset}({\hypergraphG}_1)=\vec{0}$.

\item If $m=2$ then $\size{\setA}=\size{\setB}=2$ and $\setC=\emptyset$. Conditions $4$ and $5$ define weights of all edges, the \changed{condition}{Property}{685-688} $4$ is responsible 
for 
edges with \changed{non-zero}{nonzero}{107} weights and $5$ for the edges with the weight $\vec{0}$. \changed{The condition}{Property}{685-688} $6$ says that 
the weight of $\emptyset$ and weights of single 
vertices 
are $\vec{0}$. 
\end{enumerate}
\qed

\changed{}{\begin{rem}
 Note that $(m,\vec{a})$-\simple\ hypergraphs are not defined uniquely. For example, \changed{}{the}{711} $(0,\vec{a})$-\simple\ hypergraph from the example above is not fully 
defined\changed{}{, as for any $\vec{x}$ and $\vec{y}$ we may choose appropriate $\vec{z}=\vec{a}-\vec{x}-\vec{y}$}{712}. 
 The existence of all \simple\ hypergraphs is proven later (we show how to construct them). 
\end{rem}
}{667}

\end{exa}

\changedd{Now, why we characterise \simple\  hypergraphs using such complicated \changed{conditions}{Properties}{685-688}?}{}{689} 
\changed{}{In the following we justify our design.}{703}
The most important property of \simple\  hypergraph is the last one. It implies the following lemma:
\begin{lem}\label{lem:adding simple graphs doesn't change some weights}
Let $\hypergraphG\changedd{=(\setVertices', \mu')}{}{405}$ be a ${\setSizeK}$-hypergraph and $\simpleGraph{m}{\vec{a}}$ be an 
\changed{$(m,\vec{a})-\simple$}{$(m,\vec{a})$-\simple}{20} ${\setSizeK}$-hypergraph. 
Then for any $\setX \subseteq \changedd{{\setVertices'}}{\Vertices{\hypergraphG}}{405}, \size{\setX} < m$ we have $\weight{setX}(\hypergraphG) = 
\weight{setX}(\hypergraphG+ 
\simpleGraph{m}{\vec{a}})$.
\end{lem}
\begin{proof}
 Indeed, $\weight{setX}$ is a homomorphism and
 $\weight{setX}(\simpleGraph{m}{\vec{a}})=\vec{0}$ \changed{due to}{because of}{200} \changed{property}{Property}{NULL} 
6~\changed{}{of}{694}~Definition~\ref{def:simple_hypergraphs}.  
\end{proof}

\changed{It justifies}{Property $6$, by this lemma, allows for}{703} the following 
structure of the proof of \changed{the theorem}{Theorem~\ref{thm:core}}{NULL}. 
We gradually simplify \changed{your}{the}{696} target hypergraph by adding $(i,\vec{a})$-\simple\ hypergraphs for 
growing $i$. In the step $i$ we start from \changed{a hypergraph 
with $\weight{setX}(\hypergraphH)=\vec{0}$ for every $\setX\subset \setVertices$ with at most $i-1$ vertices.}{an $(i-1)$-isolated hypergraph.}{666} Using 
$(i,\vec{a})$-\simple\ hypergraphs we 
simplify it to a hypergraph with $\weight{setX}(\hypergraphH)=\vec{0}$ for all $\setX\subset \setVertices$ with exactly $i$ vertices.
The property $6$ guaranties that 
while we perform the step $i$, we do not ruin our work from previous steps i.e. we reach \changed{a hypergraph \changed{such}{such that}{700} 
$\weight{setX}(\hypergraphH)=\vec{0}$ for every 
$\setX\subset \setVertices$ with at most $i$ vertices.}{an $i$-isolated hypergraph.}{666}
Eventually, we reach a hypergraph with all the weights equal \changed{}{to}{702} $\vec{0}$ i.e. the empty hypergraph. 

Now\changed{}{, let us justify the design of}{703} \changed{conditions}{Properties}{685-688} $4$ and $5$. We already mentioned that using $(i,\vec{a})$-\simple\ hypergraphs we want to reduce to $\vec{0}$ all weights of sets in 
$\kset{i}{{\setVertices}}$; thus it is good to keep weights on the level $i$ as 
simple as possible\changed{, i.e.}{\ i.e.}{22} $\vec{a},\vec{0},-\vec{a}$. 

\changed{Other}{The designs of}{703} \changed{conditions}{Properties $1$-$3$}{685-688} \changed{are}{is}{703} a trade-off between two things: 
\begin{itemize}
\item \changed{we wanted}{we want}{707} to have \simple\  hypergraphs as small as possible in terms of number of vertices and number of 
\changed{non-zero}{nonzero}{107} 
hyperedges, 
\item \changed{required \simple\  hypergraphs must be in $\sums{\Z}{\eqs{{\setHypergraphsH}}}$.}{
for the proof of Theorem~\ref{thm:core}, we need that \simple\ hypergraphs, from which we construct the target hypergraph $\hypergraphH$, may be constructed from
hypergraphs in the family $\setHypergraphsH$.
}{709}
\end{itemize}

\changed{\begin{rem}
 Note that $(m,\vec{a})$-\simple\ hypergraph are not defined uniquely. For example $(0,\vec{a})$-\simple\ hypergraph from the example above is not fully 
defined. 
 It is also not clear if \simple\  hypergraphs exist, this is proven later (we show how to construct them). 
\end{rem}}{}{667}

\changed{}{Now, that Properties $1$-$6$ are commented, we may go back to the proof of the Theorem~\ref{thm:core}.}{703}

We \changed{also}{}{NULL} introduce \changed{$2$}{two}{714} other notions to work with families of \simple\  hypergraphs.
\begin{defi}\label{def:simplification}
Let $\setG$ be a family of ${\setSizeK}$-hypergraphs\changed{. $\hypergraphG=({\setVertices}_{\hypergraphG}, \mu_{\hypergraphG})$ for every $\hypergraphG\in \setG$}{, for every $\hypergraphG\in \setG$ we denote the set of its vertices by $\setVertices_{\hypergraphG}$.}{715} 
\changed{For every $0\leq m\leq \setSizeK$\changed{}{,}{716} let ${\groupG}_m$ be \changed{a}{the}{716} group generated by 
$\setof{\weight{setX}(\hypergraphG)}{ \changed{\text{ for every }}{}{716} \hypergraphG\in \setG, \setX\in 
\kset{\changed{i}{m}{716}}{{\setVertices}_{\hypergraphG}}}$.
A family of ${\setSizeK}$-hypergraphs $\setS$ is a \emph{simplification of the family} 
$\setG$ if for every $0\leq m\leq\setSizeK$ and 
every $\vec{g}\in 
{\groupG}_m$ it contains an \changed{$(m,\vec{g})-\simple$}{$(m,\vec{g})$-\simple}{20} ${\setSizeK}$-hypergraph.}
{
A family of ${\setSizeK}$-hypergraphs $\setS$ is a \emph{simplification of the family} $\setG$ if  
for every $0\leq m\leq\setSizeK$ and 
every $\vec{g}$ in 
$\Z$-sums of $\setof{\weight{setX}(\hypergraphG)}{ \changed{\text{ for every }}{}{716} \hypergraphG\in \setG, \setX\in 
\kset{\changed{i}{m}{716}}{{\setVertices}_{\hypergraphG}}}$ it contains an 
\changed{$(m,\vec{g})-\simple$}{$(m,\vec{g})$-\simple}{20} ${\setSizeK}$-hypergraph.
}{716 and (last) 716}
\end{defi}
\begin{defi}\label{def:simple family for}
The family $\setS$ is \emph{\changed{simple for itself}{self-simplified}{720}} if $\setS$ is a simplification of $\setS$.
\changed{A family $\setS$ is \emph{simple for} a given ${\setSizeK}$-hypergraph 
$\hypergraphH=({\setVertices},\weight{hypergraphH})$}{For a family of
${\setSizeK}$-hypergraphs $\setHypergraphsH$, 
\changedd{\emph{a family $\setS$ is simple for} $\setHypergraphsH$}{\emph{a family $\setS$ is $\setHypergraphsH$-simplified}}{1169} 
}{NULL} if 
$\setS$ 
is \changed{a family simple for itself}{self-simplified}{720} and is
 a simplification of $\setHypergraphsH$.
\end{defi}

\begin{rem}\label{rem:simple family}
To produce a \changedd{simple}{$\setHypergraphsH$-simplified}{1169} family,
it suffices to design an algorithm $alg(\setS)$ that produces a simplification of its input. 
\changed{As a \changed{simple}{self-simplified}{721} family we 
take}{The following family is \changedd{self}{$\setHypergraphsH$}{1169}-simplified}{724}  
\[
\bigcup_{i\in \N} alg^{i+1}(\setHypergraphsH) \text{ where }
alg^i \text{ is i-th iteration of the}\ alg.
\]
Note that the produced family is not necessarily finite. 
\end{rem}

\begin{exa}
\changed{Figure}{The picture below}{729} presents an example of \changedd{a simple family $\setG$ for the graph $\hypergraphH$}{an 
$\{\hypergraphH\}$-simplified family $\setG$}{1169}.
$z \in \Z$ and the sets $\red{\setA},\blue{\setB},\green{\setC}$ are marked with colours. \changedd{}{Like in Definition~\ref{def:simplification} 
$\setVertices_{\hypergraphG}=\Vertices{\hypergraphG}$ for any hypergraph $\hypergraphG\in \setG$.}{null}

\begin{tikzpicture}
[place/.style={circle,draw=\mycolorOne!50,fill=\mycolorOne!20,thick,inner sep=0pt,minimum size=3mm}]
\node (G) at (-2,-0.5) {${\hypergraphH} =$}; 

\node[place] (A) {};
\node[place] (B) [below left=of A] {};
\node[place] (C) [below right=of A] {};
\draw (A) to node [above, sloped] {$\vec{x_1}$} (B) to node [above] {$\vec{x_2}$} (C) to node [above, sloped] {$\vec{x_3}$} (A);

\node[place, color=\mycolorDwa] (A1) at (-2, -3) {};
\node[place, color=\mycolorDwa] (B1) [below left=of A1] {};
\node[place, color=\mycolorDwa] (C1) [below right=of A1] {};
\draw (A1) to node [above, sloped] {$z\cdot \vec{x_1}$} (B1) to node [above] {$z\cdot \vec{x_2}$} (C1) to node [above, sloped] {$z\cdot \vec{x_3}$} (A1);

\node[place, color=\mycolorDwa] (A2) at (-2, -5) {};
\node[place, color=mycol3] (B2) [below left=of A2] {};
\node[place, color=\mycolorOne] (C2) [below right=of A2] {};
\draw (B2) to node [above, sloped] {$\vec{a}$} (A2) to node [above, sloped] {$-\vec{a}$} (C2) ;

\node (X0) at (2.2, -4) {for \changed{any}{all}{NULL} $z\in \Z$.};

\node (X1) at (5.3, -6) {\changed{where $\vec{a}$ is}{for all $\vec{a}$}{NULL} in $\sums{\Z}{\{ \vec{x_1}+ \vec{x_2},\ \vec{x_2}+\vec{x_3},\ 
\vec{x_3}+\vec{x_1}\}}$.};

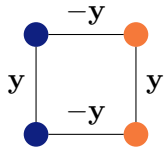
\node[place, color=mycol3] (A21) at (-2.7, -8.5) {};
\node[place, color=mycol3] (B21) [above =of A21] {};
\node[place, color=\mycolorOne] (C21) [right =of A21] {};
\node[place, color=\mycolorOne] (D21) [above =of C21] {};

\draw (A21) to node [left] {$\vec{y}$} (B21) to node [above] {$-\vec{y}$} (D21) to node [right] {$\vec{y}$} (C21) to node [above] {$-\vec{y}$} (A21);

\node (X1) at (4, -8) {\changed{where $\vec{y}$ is}{for all $\vec{y}$}{NULL} in $\sums{\Z}{\{ \vec{x_1},\ \vec{x_2},\ \vec{x_3}\}}$.};

\end{tikzpicture}

\changed{}{
The family is infinite. Depicted \changedd{}{three types of}{null} graphs represent $(0,z(\vec{x_1}+\vec{x_2}+\vec{x_3}))$-simple hypergraphs, 
$(1,\vec{a})$-simple hypergraphs,
$(2,\vec{y})$-simple hypergraphs. To see that the family is self-simplified we observe the few following facts:
\begin{itemize}
 \item $\setof{\weight{emptyset}(\hypergraphG)}{ \hypergraphG\in \setG} = \setof{z(\vec{x_1}+\vec{x_2}+\vec{x_3})}{z\in \Z}=
 {\sums{\Z}{\setof{z(\vec{x_1}+\vec{x_2}+\vec{x_3})}{z\in \Z}}} $. Indeed, 
 $\weight{emptyset}(\hypergraphG) = \vec{0}$ for $\hypergraphG$ of the second or the third type of graphs. So the family contains all required 
$(0,\placeHolder)$-\simple\ graphs (triangles). 
 \item $\setof{\weight{da}(\hypergraphG)}{ \hypergraphG\changedd{=(\setVertices_{\hypergraphG}, \mu_{\hypergraphG})}{}{405}\in \setG, \da\in 
\setVertices_{\hypergraphG}} = 
\setof{\vec{a}}{\vec{a}\in\sums{\Z}{\{ \vec{x_1}+ \vec{x_2},\ \vec{x_2}+\vec{x_3},\ \vec{x_3}+\vec{x_1}\}}}$. 
This is because of, 
 \begin{itemize}
  \item  $\weight{da}(\hypergraphG) = \vec{0}$ for any $\hypergraphG\changedd{=(\setVertices_{\hypergraphG}, \mu_{\hypergraphG})}{}{405}$ of the third type of 
graphs and any 
$\da\in \setVertices_{\hypergraphG}$. 
  \item  $\weight{da}(\hypergraphG) = \vec{0}$ for $\hypergraphG$ of the second type of graphs and $\da$ being the middle vertex.
  \item  it is easy to check that $\weight{da}(\hypergraphG) \in  \setof{\vec{a}}{\sums{\Z}{\{ \vec{x_1}+ \vec{x_2},\ \vec{x_2}+\vec{x_3},\ 
\vec{x_3}+\vec{x_1}\}}}$ \changedd{also}{}{null} for other cases.
 \end{itemize}
 Thus, the family contains all required 
$(1,\vec{a})$-\simple\ graphs (the second type). 
 \item $\setof{\weight{{da,db}}(\hypergraphG)}{ \hypergraphG\changedd{=(\setVertices_{\hypergraphG}, \mu_{\hypergraphG})}{}{405}\in \setG,
 \{\da,\db\}\in \kset{2}{\setVertices_{\hypergraphG}}} = 
 \setof{\vec{a}}{\sums{\Z}{\{ \vec{x_1},\ \vec{x_2},\ \vec{x_3}\}}}$. This is also easy to check.\\ 
\changedd{ Because of $\sums{\Z}{\sums{\Z}{\setX}}=\sums{\Z}{\setX}$ for any set $\setX$, we see that all needed \simple\ graphs are elements of the family $\setG$. 
 Thus, the family contains all required 
 $(2,\vec{y})$-\simple\ graphs (the square type).}{}{null} 
 \end{itemize} 
}{731}
\end{exa}

\begin{lem}\label{lem:adding_simple_graphs}
Let \changedd{}{$\hypergraphH$ be a $\setSizeK$-hypergraph and}{1169} $\setS$ be \changedd{a family simple for a \changed{}{$\setSizeK$-}{749}hypergraph 
${\hypergraphH}$}{an $\hypergraphH$-simplified family}{1169}. Suppose $\hypergraphG\in \sums{\Z}{\eqs{\setS}}$ then the family 
$\setS$ is \changedd{simple for 
$\{{\hypergraphH}+\hypergraphG\}$}{$\{{\hypergraphH}+\hypergraphG\}$-simplified}{1169}.  
\end{lem}
\begin{proof}
As $\setS$ is self-simplified, we only need prove that $\setS$ is a simplification of $\{\hypergraphH+\hypergraphG\}$.
Let $\setVertices_{\hypergraphH+\hypergraphG}=\Vertices{\hypergraphH}\cup \Vertices{\hypergraphG}$, and $m\leq 
\setSizeK$. 
Let $\vec{g}\in \sums{\Z}{{\setof{\weight{setX}(\hypergraphH + \hypergraphG)}{  \setX\in 
\kset{{m}}{\setVertices_{\hypergraphH+\hypergraphG}}}}}$.
We have to prove that $\setS$ contains an $(m,\vec{g})$-\simple\ hypergraph.
Suppose that $\vec{g}=c_1\vec{a_1}+c_2\vec{a_2}\ldots +c_l\vec{a_l}$ where 
$\vec{a_i}\in {\setof{\weight{setX}(\hypergraphH + \hypergraphG)}{  \setX\in 
\kset{{m}}{\setVertices_{\hypergraphH+\hypergraphG}}}}
$ and $c_i\in \Z$.
Observe that it is sufficient to prove that $\setS$ contains an $(m,\vec{a_i})$-\simple\ hypergraph for each $\vec{a_i}$. This is because $\setS$ is 
self-simplified, precisely.
If in $\setS$ there are $(m,\vec{a_1})$-\simple\ and $(m,\vec{a_2})$-\simple\ hypergraphs then $c_1\vec{a_1}+c_2\vec{a_2}\in 
\sums{\Z}{{\setof{\weight{setX}(\hypergraphG')}{\hypergraphG'\in \setS\  \setX\in 
\kset{{m}}{\Vertices{\hypergraphG'}}}}}$. But, as $\setS$ is self-simplified we have that an $(m,c_1\vec{a_1}+c_2\vec{a_2})$-\simple\ hypergraph is an element 
 of $\setS$. 

The fact that $(m, \vec{a_i})$-\simple\ hypergraphs are elements of $\setS$ is easy. Suppose that $\vec{a_i}= \weight{setX}({\hypergraphH}+\hypergraphG)$. 
Observe that 
$\weight{setX}({\hypergraphH}+\hypergraphG)=\weight{setX}({\hypergraphH}) + \weight{setX}(\hypergraphG)$. One more time, as $\setS$ is self-simplified it is 
sufficient to show that in $\setS$ there are $(m,\weight{setX}({\hypergraphH}))$-\simple\ and $(m,\weight{setX}(\hypergraphG))$-\simple\ hypergraphs.
The first one is in $\setS$ as $\setS$ is $\hypergraphH$-simplified. To show that the second is an element of $\setS$ we have to use the fact 
that $\hypergraphG\in \sums{\Z}{\eqs{\setS}}$, then $\weight{setX}(\hypergraphG)\in  \sums{\Z}{{\setof{\weight{setX}(\hypergraphG')}{\hypergraphG'\in \setS\ 
 \setX\in 
\kset{{m}}{\Vertices{\hypergraphG'}}}}}$. But as 
$\setS$ is 
self-simplified we have that an $(m, \weight{setX}(\hypergraphG))$-\simple\ hypergraph is an element of $\setS$.
\end{proof}

\section{Expressing \texorpdfstring{${\hypergraphH}$}{a hypergraph} with simple hypergraphs}\label{sec:Expressing with simple hypergraphs}

Our goal in this section is to prove Theorem~\ref{th:expresibility with simple}\changed{ which is a~slightly stronger version of the the following claim.\\
\noindent{\bf Claim:}
If $\setS$ is a simple family for a ${\setSizeK}$-hypergraph ${\hypergraphH}$ then ${\hypergraphH}\in \sums{\Z}{\eqs{\setS}}$. \\
We need a stronger version to use in the proof of existence of \simple\  hypergraphs. 
The stronger version requires the notion of a support of a sum.}{.}{753-757}

\begin{defi}
\changedd{For a hypergraph ${\hypergraphH}=({\setVertices},\mu)$ its support is 
${\setVertices}$. We denote it $\support{{\hypergraphH}}$.}{}{405} 
Let ${\hypergraphH}\in \sums{\Z}{\eqs{\setHypergraphsH}}$. We say that 
\emph{${\setVertices\changed{}{'}{759}}$ supports $\hypergraphH$ for the family 
$\setHypergraphsH$} 
if there is a solution \changed{of}{to}{760} the following equation
\[
 {\hypergraphH}=\sum_{i>0} a_i {\hypergraphH}_i \text{ where } a_i\in \Z, {\hypergraphH}_i\in \eqs\setHypergraphsH, \text{ and } 
\changedd{\support{{\hypergraphH}_i}}{\Vertices{{\hypergraphH}_i}}{405}\subseteq 
{\setVertices\changed{}{'}{759}} \text{ for all } i.
\]
\end{defi}

\changed{}{\changedd{For the second part of the definition, one}{One}{405} should think that $\setVertices'$ \changedd{may include}{includes}{null} all 
\changedd{supports}{vertices}{405} of the elements of the sum.
\changedd{This naming is 
motivated by the literature where support is used for objects but also for sets of objects, expressions build from objects \changedd{e.t.c.}{etc.}{758} 
For the 
convenience of 
the reader in the following text whenever we speak about support of the hypergraph we use $\support{}$ symbol and all occurrences of word support in the text 
are used as ``\ldots supports \ldots for a family \ldots ''.}{}{null}}{759}

\begin{thm}\label{th:expresibility with simple}
Let $\hypergraphH\changedd{=(\setVertices, \mu)}{}{405}$ be a ${\setSizeK}$-hypergraph\changedd{}{ and $\setVertices$ be its set of vertices}{405}. Further, 
\changedd{let $\setS$ be a simple family for ${\hypergraphH}$}{let $\setS$ be an $\{\hypergraphH\}$-simplified family of hypergraphs}{1169}. 
Then\changed{,}{}{765} 
${\hypergraphH}\in 
\sums{\Z}{\eqs{\setS}}$. 
Moreover, if $\size{{\setVertices}}>2k-1$ then ${\setVertices}$ supports ${\hypergraphH}$ for the family $\setS$. 
\end{thm}

\changed{Proof}{The proof}{766} of this theorem requires a few definitions and lemmas stated below. We start with them
and then we prove the theorem while proofs of lemmas are postponed.

\changed{}{Let us recall (Definitiont~\ref{def:m isiolation}) that, a ${\setSizeK}$-hypergraph ${\hypergraphH}\changedd{=(\setVertices,\mu)}{}{405}$ is 
\emph{m-isolated} 
if for any subset $\setX \subseteq \changedd{{\setVertices}}{\Vertices{\hypergraphH}}{405}$ such that $\size{\setX}\leq m$ it holds that
$\weight{setX}({\hypergraphH})=\vec{0}$.}{666}

\begin{rem}\label{rmk:fully isolated graph vanishes}
If ${\hypergraphH}$ is a ${\setSizeK}$-isolated ${\setSizeK}$-hypergraph, then ${\hypergraphH}$ is equivalent to the empty hypergraph \changed{}{i.e. it is a 
union of isolated vertices}{770}.
\end{rem}

\begin{defi}\label{def:almost m isolation}
Let ${\hypergraphH}\changedd{=(\setVertices,\mu)}{}{405}$ be a ${\setSizeK}$-hypergraph\changed{}{ and $\setVertices$ be its set of vertices}{405}. We say that 
\changedd{it}{$\hypergraphH$}{405} is \emph{\changed{almost}{pre}{777} m-isolated} if the following two 
conditions are satisfied:
\begin{itemize}
    \item ${\hypergraphH}$ is $(m-1)$-isolated,
    \item there is $\setX\subseteq {\setVertices}$ a set of vertices such that $\size{\setX}\leq 2m-1$ and 
for any $\setY \in \kset{m}{{\setVertices}}$ such that $\setY\not\subseteq \setX$ \changed{holds}{it holds that}{451} $\weight{setY}({\hypergraphH})=\vec{0}$. 
\end{itemize}
\end{defi}

\begin{lem}\label{lem:small k-1-isolated graphs vanish}
If ${\hypergraphH}$ is a ${\setSizeK}$-hypergraph that is \changed{almost}{pre}{777} $m$-isolated then it is $m$-isolated.
\end{lem}

\begin{lem}\label{lem:isolating a vertex}
Let ${\hypergraphH}\changedd{=(\setVertices, \mu)}{}{405}$ be an $m$-isolated ${\setSizeK}$-hypergraph\changedd{}{ and $\setVertices$ be its 
set of vertices.}{405}\changed{,}{}{778} \changedd{and $\setS$ be a simple 
family for ${\hypergraphH}$.}{Suppose $\setS$ is an $\{\hypergraphH\}$-simplified family of hypergraphs.}{1169}
Then there is a hypergraph $\hypergraphG\in \sums{\Z}{\eqs{\setS}}$, such that 
${\hypergraphH}+\hypergraphG$ is \changed{almost}{pre}{777} \changed{$m+1$-}{$(m+1)$-}{779}isolated. Moreover, if $\size{{\setVertices}}>2k-1$ then 
${\setVertices}$ supports $\hypergraphG$ for the 
family $\setS$.
\end{lem}

\begin{proof}[Proof of Theorem~\ref{th:expresibility with simple}]
\changed{Without loss of generality we may}{We}{781} assume that 
$\size{{\setVertices}}>2k-1$, as otherwise we extend \changed{it}{$\hypergraphH$}{781} with a 
few isolated vertices.

\changed{We construct}{Suppose we have}{NULL} a sequence of hypergraphs $\hypergraphG_i\in \sums{\Z}{\eqs{\setS}}$
\changed{ 
such that \changed{}{for every $j\leq\setSizeK$ the hypergraph}{783} 
${\hypergraphH}-\sum_{i=0}^j \hypergraphG_i$ is $j$-isolated.
Also, for every \changed{$i\leq j$}{$i \leq\setSizeK$}{NULL} the set 
${\setVertices}$ supports $\hypergraphG_i$ for the family $\setS$.  
We define ${\hypergraphG}$ as $\sum_{i=0}^{\setSizeK} \hypergraphG_i$. Observe that $\hypergraphH = \hypergraphG$ as $\hypergraphH-\sum_{i=0}^{\setSizeK} 
\hypergraphG_i$
is a $\setSizeK$-isolated $\setSizeK$-hypergraph, \changed{so}{ and every $\setSizeK$-isolated $\setSizeK$-hypergraph}{NULL} is empty \changed{due to}{because 
of}{200} Remark~\ref{rmk:fully isolated graph vanishes}. Moreover, $\setVertices$ supports $\hypergraphG$ for the 
family $\setS$ as for every $i\leq\setSizeK$ the set $\setVertices$ supports $\hypergraphG_i$ for the family $\setS$.
}
{
which satisfies following properties.
\begin{itemize}
  \item For every $j\leq\setSizeK$ the hypergraph ${\hypergraphH}-\sum_{i=0}^j \hypergraphG_i$ is $j$-isolated.
 \item For every \changed{$i\leq j$}{$i \leq\setSizeK$}{NULL} the set ${\setVertices}$ supports $\hypergraphG_i$ for the family $\setS$. 
\end{itemize}
We define ${\hypergraphG}$ as $\sum_{i=0}^{\setSizeK} \hypergraphG_i$.
Now\changed{}{,}{NULL} observe that 
\begin{itemize}
 \item $\hypergraphG\in \sums{\Z}{\eqs{\setS}}$ as each $\hypergraphG_i$ is in $\sums{\Z}{\eqs{\setS}}$.
 \item $\changedd{\support{\hypergraphG}\subseteq \setVertices}{\Vertices{\hypergraphG}\subseteq \setVertices}{1169}$ as 
$\changedd{\support{\hypergraphG_i}\subseteq \setVertices}{\Vertices{\hypergraphG_i}\subseteq \setVertices}{1169}$ for each $i\leq \setSizeK$.
 \item $\setVertices$ supports $\hypergraphG$ for the family $\setS$ as $\setVertices$ supports $\hypergraphG_i$ for the family $\setS$ for each $i\leq 
\setSizeK$.
 \item $\hypergraphH = \hypergraphG$ up to equivalence. Because of the first property we know that $\hypergraphH - \hypergraphG$ is a $\setSizeK$-isolated 
$\setSizeK$-hypergraph, \changed{so}{ 
and every $\setSizeK$-isolated 
$\setSizeK$-hypergraph}{NULL} is the empty hypergraph, \changed{due to}{because of}{200} Remark~\ref{rmk:fully isolated graph vanishes}.
\end{itemize}
Thus, if we have the sequence $\hypergraphG_i$ then we are done. Now, we show how the sequence $\hypergraphG_i$ may be constructed.
}{783,790}
The construction is via induction on $\changed{j}{i}{790}$.
As $\setS$ is \changedd{a simple family for ${\hypergraphH}$}{an $\{\hypergraphH\}$-simplified family of hypergraphs}{1169} then there is 
$\hypergraphG_0 \in{\eqs{\setS}}$ and $\changedd{\support{\changed{\hypergraphG}{\hypergraphG_0}{NULL}}}{\Vertices{\hypergraphG_0}}{1169}\subset 
{\setVertices}$ such that ${\hypergraphH}-\hypergraphG_0$ 
is 
$0$-isolated. 
This creates the induction base.

For the inductive step we reason as follows. First we observe that \changed{due to}{because of}{200} Lemma~\ref{lem:adding_simple_graphs} the family $\setS$ is 
\changedd{simple for 
${\hypergraphH}-\sum_{j=0}^{i}\hypergraphG_j$}{$\left({\hypergraphH}-\sum_{j=0}^{i}\hypergraphG_j\right)$-simplified}{1169}.
Thus, we may use Lemma~\ref{lem:isolating a vertex} for \changed{a}{the}{795} hypergraph
${\hypergraphH}-\sum_{j=0}^i \hypergraphG_j$ and the family $\setS$.
As a consequence we get $\hypergraphG_{i+1}\in \sums{\Z}{\eqs{\setS}}$ such that
$({\hypergraphH}-\sum_{j=0}^i \hypergraphG_j)-\hypergraphG_{i+1}$ is \changed{almost}{pre}{777} \changed{$i+1$-}{$(i+1)$-}{779}isolated.
Moreover, $\setVertices$ supports $\hypergraphG_{i+1}$ for the family $\setS$ (the property 2).
Further, Lemma~\ref{lem:small k-1-isolated graphs vanish} implies 
that $({\hypergraphH}-\sum_{j=0}^i \hypergraphG_j)-\hypergraphG_{i+1}$ is \changed{$i+1$-}{$(i+1)$-}{779}isolated (the property 1). 

\changedd{Note that $\setVertices$ supports $\hypergraphG_{i+1}$ for the family $\setS$ thus $\setVertices$ supports $({\hypergraphH}-\sum_{j=0}^{i+1} \hypergraphG_j)$ 
for the family $\setS$, too.}{}{null}
This ends the inductive step.
\end{proof}

\section*{ Proof of Lemma~\ref{lem:small k-1-isolated graphs vanish}.}

Before we prove Lemma~\ref{lem:small k-1-isolated graphs vanish} we prove \changedd{a simple}{an easy}{null} 
lemma about the $\weight{setX}$ functions. \changed{The main tool in the proof of Lemma~\ref{lem:small k-1-isolated graphs vanish} 
are reduction matrices defined in Section~\ref{sec:matrix}.}{}{804}

\begin{lem}\label{lem:weights of sums are proportional}
Suppose ${\hypergraphH}=(\setVertices, \mu)$ is a ${\setSizeK}$-hypergraph, $\setX \in \kset{m}{{\setVertices}}$\changed{}{,}{806} and $m\leq l \leq\setSizeK$. 
Let $\familyF = \setof{\setY \in \kset{l}{{\setVertices}}}{ \setX \subseteq \setY}$, be \changed{a}{the}{807} family of \changed{$l-$}{$l$-}{807}element 
supersets of $\setX$. 
Then:
\[
 \sum\limits_{\setY \in \familyF}\weight{setY}({\hypergraphH}) = \binomial{k-m}{l-m} \weight{setX}({\hypergraphH}).
\]
 \end{lem}

\begin{proof}
\changed{}{Let us recall definition of $\weight{setX}(\hypergraphH)$. It is the sum of all edges $\edge$ 
in $\hypergraphH$ such that $\setX\subseteq \edge$.}{NULL}
Let $\edge$ be a hyperedge in ${\hypergraphH}$ such that $\setX \subseteq \edge$.
It suffices to prove that $\mu(\edge)$ 
appears the same number of times on \changed{the}{}{813} both sides of the equation. 
\changed{\[
\sum\limits_{\setY \in \familyF} \left( \sum_{\edge\in \kset{\setSizeK}{{\setVertices}}, \setY\subseteq \edge} \mu(\edge) \right)=
{k-m \choose l-m} \sum_{\edge\in \kset{\setSizeK}{{\setVertices}}, \setX\subseteq \edge} \mu(\edge).
\]}{}{812,816}
On the right\changed{-hand}{}{44} side $\edge$ is added $\binomial{k-m}{l-m}$ times\changed{}{, as $\mu(e)$ appears once in the expression 
$\weight{setX}({\hypergraphH})$}{NULL}.
On the left\changed{-hand}{}{44} side the number of times when \changed{$\edge$}{$\mu(\edge)$}{NULL} is added is equal to the number of 
$l$-element supersets of $\setX$ that are included in $\edge$. \changed{}{This is because $\mu(e)$ appears once in the expression 
$\weight{setY}({\hypergraphH})$ if $\setY\subseteq e$.}{NULL} But this is equal to
$\binomial{\size{\edge}-\size{\setX}}{l-\size{\setX}}=\binomial{k-m}{l-m}$, as 
required.
\end{proof}

\changed{}{The main tool in the proof of Lemma~\ref{lem:small k-1-isolated graphs vanish} 
\changed{are}{is}{804} reduction matrices defined in Section~\ref{sec:matrix}.}{804}

\begin{proof}[Proof of Lemma~\ref{lem:small k-1-isolated graphs vanish}] 
Let $\setX$ be a set of vertices as in Definition~\ref{def:almost m isolation}. 
We \changed{only}{}{823} need to prove that for any $\setX'\subseteq \setX$ such that $\size{\setX'}=m$ \changed{holds}{it holds that}{451} 
$\weight{setX'}({\hypergraphH})=\vec{0}$.

Let $\setY_1, \setY_2 \ldots \setY_n$ be all \changed{}{the}{825} $(m-1)$-subsets of $\setX$ and let $\setX'_1, \setX'_2 \ldots \setX'_{n'}$ be all
\changed{}{ the}{825} $m$-subsets of $\setX$.

\changed{Due to}{Because of}{200} Lemma~\ref{lem:weights of sums are proportional}\changed{}{,}{827} 
we know that\changed{}{,}{827} for any $\setY_i$\changed{}{,}{827} the equation
\[
\sum\limits_{\setX'_j \supset \setY_i} \weight{setX'_j}({\hypergraphH}) =\changed{((2m-1)-(m-1))}{\binom{k-m}{m - (m-1)}}{829}\cdot 
\weight{setY_j}({\hypergraphH})\changed{\text{ holds.}}{}{828}
\]
\changed{}{holds. }{828}But ${\hypergraphH}$ is $m-1$ isolated so $\weight{setY_j}({\hypergraphH})=\vec{0}$.

This system of equation may be rewritten in matrix form
\[
 Cu = \vec{0} \text{ where } 
\]
\begin{displaymath}
u \eqdef \begin{bmatrix} 
\weight{setX'_1}({\hypergraphH}) \\ \weight{setX'_2}({\hypergraphH}) \\ \vdots \\ \weight{setX'_n}({\hypergraphH})
\end{bmatrix} 
\end{displaymath} and 
\changed{$C\eqdef \reductionMatrix{2m-1}{m}{m-1}$ ($\reductionMatrix{\bullet}{\bullet}{\bullet}$ are 
defined in Definition~\ref{def:reduction matrix}).}{
$C$ is the matrix with $\binom{\size{\setX}}{m}=\binom{2m-1}{m}$ columns indexed with \changed{$m$-elements}{$m$-element}{54} 
subsets of the set $\setX$, $\binom{{\size{\setX}}}{m-1}=\binom{2m-1}{m-1}$ 
rows indexed with \changed{$(m-1)$-elements}{$(m-1)$-element}{54} subsets of the set $\setX$, such that each 
individual entry represents inclusion between the index of the row and the index of the column. So up to permutation of rows and columns $C$ 
is the $\reductionMatrix{2m-1}{m}{m-1}$ matrix ($\reductionMatrix{\bullet}{\bullet}{\bullet}$ \changed{are}{is}{840} 
defined in Definition~\ref{def:reduction matrix}).
}{840}

But according to Lemma~\ref{lem:main_matrix} the rank of the matrix $C$ is maximal, 
which implies $u = \vec{0}$ is the only solution of the system of equations. \changed{Thus }{Thus, }{177}$\weight{setX'_j}({\hypergraphH}) = \vec{0}$ 
for any $j\leq n'$ and consequently ${\hypergraphH}$ is $m$-isolated.
\end{proof}

\section*{Proof of Lemma~\ref{lem:isolating a vertex}.}

The proof of Lemma~\ref{lem:isolating a vertex} requires some preparation.

\begin{defi}\label{def:cut vertex}
\changed{Suppose $\setX \subset {\setVertices}, \size{\setX} <\setSizeK$. 
We define $\cut{{\hypergraphH}}{\setX} \eqdef ({\setVertices} \setminus \setX, {\cut{\mu}{\setX}})$, where 
${\cut{\mu}{\setX}}(\edge)$ is a function from $\kset{k-\size{\setX}}{({\setVertices}-\setX)}$ 
 to $\ktuple{\dimension}{\Z}$ and is defined as follows
\[
{\cut{\mu}{\setX}}(\edge)\eqdef \mu(\edge\cup{\setX})\text{ where } 
\edge\in \kset{k-\size{\setX}}{({\setVertices}-\setX)}. 
\]
The above operation is called \emph{$\setX$-cut} of ${\hypergraphH}$. 
The reverse operation called \emph{enrich} ${\hypergraphH}$ with $\setX$ is denoted by $\revcut{{\hypergraphH}}{\setX}$ 
and its effect is a minimal in the sense of inclusion\footnote{Inclusion of the sets of hyperedges.} 
$(k+\size{\setX})$-hypergraph ${\hypergraphG}$ such that $\cut{{\hypergraphG}}{\setX}={\hypergraphH}$.
}
{
Let $\hypergraphF=(\setVerticesW,\mu')$ is a $\setSizeK$-hypergraph.
Suppose $\setX \subset {\setVerticesW}, \size{\setX} < \setSizeK$. 
We define $\cut{{\hypergraphF}}{\setX} \eqdef ({\setVerticesW} \setminus \setX, {\cut{\mu'}{\setX}})$, where 
${\cut{\mu'}{\setX}}(\edge)$ is a function from $\kset{\setSizeK-\size{\setX}}{({\setVerticesW}-\setX)}$ 
 to $\ktuple{\dimension}{\Z}$ and is defined as follows
\[
{\cut{\mu'}{\setX}}(\edge)\eqdef \mu'(\edge\cup{\setX})\text{ where } 
\edge\in \kset{\setSizeK-\size{\setX}}{({\setVerticesW}-\setX)}. 
\]
The above operation is called \changed{$\setX$ \emph{cut}}{\emph{$\setX$-cut}}{851} of ${\hypergraphF}$. \\
For $\setX\cap \setVerticesW=\emptyset$ we define the reverse operation called \changed{\emph{enrich}}{\emph{enriching}}{851} 
${\hypergraphF}$ with $\setX$. It is denoted by $\revcut{{\hypergraphF}}{\setX}$ 
and its effect is the minimal in the sense of inclusion\footnote{\changed{Inclusion of the sets of hyperedges.}{we say that a hypergraph includes a 
second 
hypergraph if the set of vertices of the 
first hypergraph includes the set of vertices of the second one and every hyperedge of the second hypergraph is also a hyperedge of the first 
hypergraph.}{852}} 
$(\setSizeK+\size{\setX})$-hypergraph ${\hypergraphF'}$ such that $\cut{{\hypergraphF'}}{\setX}={\hypergraphF}$.
}{847}

\changed{If $\setX = \{\da\}$ is a singleton, then we simplify the notation 
$\cut{{\hypergraphG}}{\da} \eqdef \cut{{\hypergraphG}}{\{\da\}}$ and $\revcut{{\hypergraphG}}{\da} \eqdef \revcut{{\hypergraphG}}{\{\da\}}$.
}
{
If $\setX = \{\da\}$ is a singleton, then we simplify the notation 
$\cut{{\hypergraphF}}{\da} \eqdef \cut{{\hypergraphF}}{\{\da\}}$ and $\revcut{{\hypergraphF}}{\da} \eqdef \revcut{{\hypergraphF}}{\{\da\}}$.
}{847}

\end{defi}

\begin{exa}
\changed{removed figure}{}{857}

\begin{tikzpicture}[%tdplot_main_coords,
line cap=butt,line join=bevel]
\pgfmathsetmacro{\B}{2}
\pgfmathsetmacro{\H}{2}
 \draw[\mycolorDwa,thick] (-\B,-\B,0) -- (\B,-\B,0) --  (-\B,\B,0) -- cycle;
 \node [font=\large\bfseries] at (-\B,\B+0.3,0) {$\da$};
 \node (aux1) at (0,0,2/3*\H) {\color{\mycolorTrzy}{$\bullet$}};
 \node (aux2) at (0,-2/3*\B,2/3*\H) {\color{\mycolorOne}{$\bullet$}};
 \node (aux4) at (-2/3*\B,0,2/3*\H) {\color{\mycolorDwa}{$\bullet$}};
 
 \coordinate (aux3) at (-2/3*\B,2/5*\B,0);
 \draw[\mycolorOne,thick,fill=\mycolorTrzy,fill opacity=0.3] (-\B,\B,0) -- (0,0,2*\H) -- (\B,-\B,0) -- cycle;
 \draw[\mycolorOne,thick,fill=\mycolorOne,fill opacity=0.7] (-\B,-\B,0) -- (\B,-\B,0) --  (0,0,2*\H) -- cycle -- cycle;
 \draw[\mycolorOne,thick,fill=\mycolorDwa,fill opacity=0.4] (-\B,-\B,0) -- (0,0,2*\H) -- (-\B,\B,0) -- cycle;
 \begin{scope}[%tdplot_screen_coords
 ]
  \draw (aux1) -- ++ (2,0.2) node[right,font=\itshape] {\color{\mycolorTrzy}{a}};
  \draw (aux4) -- ++ (-1.5,0.3) node[left,font=\itshape] {\color{\mycolorDwa}{b}};
  \draw (aux3) -- ++ (1,1) node[above right,font=\itshape] {c};
  \draw (aux2) -- ++ (1,-1) node[below right,font=\itshape] {\color{\mycolorOne}{d}};
 \end{scope}
 
 \coordinate (A) at (5+-\B,-\B,0) ;
 \coordinate (B) at (5+\B,-\B,0) ;
 \coordinate (C) at (5,0,2*\H) {};
 
 \draw [\mycolorDwa] (A) -- (C) coordinate[midway](A1) ;
 \node[ left = 0.1 cm and 0 cm of A1,font=\itshape] {\color{\mycolorDwa}{b}};
 \draw [\mycolorTrzy] (B) -- (C)coordinate[midway](B1);
 \node[ above = 0 cm and 0 cm of B1,font=\itshape] {\color{\mycolorTrzy}{a}};
 \draw [black] (B) -- (A) coordinate[midway](C1);
 \node[below= 0.1 cm and 0 cm of C1,font=\itshape] {c};

\end{tikzpicture}

On the left\changed{-hand}{}{44} side there is a $3$-hypergraph ${\hypergraphF}$ and
on the right\changed{ hand}{}{44} side 
there is a $2$-hypergraph
$\cut{{\hypergraphF}}{\da}$. \changed{}{We use letters $a,b,c,d$ to name corresponding hyperedges.}{863} The orange bottom hyperedge \changed{}{$d$}{863} 
disappears as it does not contain $\da$.
$\revcut{\cut{{\hypergraphF}}{\da}}{\da}$ would look like ${\hypergraphF}$ but without the orange hyperedge, as the \changed{enrich}{enriching}{851} operation
takes the minimal hypergraph among all such that cutting $\da$ returns $\cut{{\hypergraphF}}{\da}$.

\end{exa}

\begin{lem}\label{lem:isolation_down}
Let ${\changed{\hypergraphH}{\hypergraphF}{847}}=(\changed{\setVertices}{\setVerticesW}{847}, \changed{\mu}{\mu'}{847})$ \changed{}{be a}{869} 
${\setSizeK}$-hypergraph\changed{,}{}{869} and $\setX\in \kset{m}{{\changed{\setVertices}{\setVerticesW}{847}}}$ for $m<\setSizeK$.
Then for any nonempty \changed{}{set }{451}$\setY\subseteq {\changed{\setVertices}{\setVerticesW}{847}}$ such that $\setY\cap \setX=\emptyset$ 
\changed{holds}{it holds that}{451} 
$\weight{YcupX}({\changed{\hypergraphH}{\hypergraphF}{847}})=\weight{setY}(\cut{{\changed{\hypergraphH}{\hypergraphF}{847}}}{\setX})$.
In particular, if $\changed{\hypergraphH}{\hypergraphF}{847}$ is $l$-isolated and $m<l$ then $\cut{{\changed{\hypergraphH}{\hypergraphF}{847}}}{\setX}$ is 
$(l-m)$-isolated.
\end{lem}
 \begin{proof}
\changed{Note that $\edge$ and $\edge'$ are parameters of the sums below.}{In the equation below, $e$ and $e'$ are indexes, while $\setVerticesW, \setY$ and 
$\setX$ are fixed.}{874, 880} 
\begin{align*} 
 \weight{YcupX}({\changed{\hypergraphH}{\hypergraphF}{847}})=\\
 \sum_{\edge \in \kset{\setSizeK}{{\changed{\setVertices}{\setVerticesW}{847}}}, \setY\cup 
\setX\subseteq \edge} \changed{\mu}{\mu'}{847}(\edge)=\\
  \sum_{\edge\in \kset{\setSizeK}{{\changed{\setVertices}{\setVerticesW}{847}}}, \setY\cup \setX\subseteq \edge} 
\weight{e-X}(\cut{{\changed{\hypergraphH}{\hypergraphF}{847}}}{\setX})=\\
  \sum_{\changed{\edge\setminus \setX) \kset{\setSizeK-\size{\setX}}{\setVerticesW}}
  {(\edge\setminus \setX) \kset{\setSizeK-\size{\setX}}{\setVerticesW\setminus \setX}}{875-883}
  , \setY\subseteq (\edge\setminus \setX)} 
\weight{e-X}(\cut{{\changed{\hypergraphH}{\hypergraphF}{847}}}{\setX})= \\
  \changed{ \text{By $\edge'$ we denote hyperedges of the hypergraph $\cut{{\changed{\hypergraphH}{\hypergraphF}{847}}}{\setX}$.}}{}{875-883}\\
  \sum_{\changed{}{\edge'\in \kset{\setSizeK-\size{\setX}}{\setVerticesW\setminus \setX},}{875-883} \setY\subseteq \edge'} 
\weight{edge'}(\cut{{\changed{\hypergraphH}{\hypergraphF}{847}}}{\setX})=\\
% align qed box with the last line
\weight{setY}(\cut{{\changed{\hypergraphH}{\hypergraphF}{847}}}{\setX}) \tag*{\qedhere}
 \end{align*}
 \end{proof}

\begin{proof}[Proof of Lemma~\ref{lem:isolating a vertex}]
Without loss of generality, we assume that $\size{{\setVertices}}>2\setSizeK-1$. Indeed, if 
$\size{{\setVertices}}\leq 2\setSizeK-1$ then we extend ${\hypergraphH}$ with a few isolated vertices.

In the proof, instead of constructing $\hypergraphG$ directly, we build a finite sequence of hypergraphs 
\changed{${\hypergraphH}_i$}{${\hypergraphH}_0\ldots{\hypergraphH}_{last}$}{888}, such that:
\begin{itemize}
 \item ${\hypergraphH}_0={\hypergraphH}$, 
 \item ${\hypergraphH}_{last}$ is 
\changed{almost}{pre}{777} \changed{${\newl}+1$-}{$({\newl}+1)$-}{779}isolated,
\changed{and}{
\item there is a well-quasi order $>$ on hypergraphs such that $\hypergraphH_i >\hypergraphH_{i+1}$ for each $0\leq i < last$, and
}{885-890}
\item each 
${\hypergraphH}_{i+1}={\hypergraphH}_i-\simpleGraph{newl+1}{}_i$ 
where $\simpleGraph{newl+1}{}_i$ is a $({\newl}+1, {\placeHolder})$-\simple\ hypergraph such that $\simpleGraph{newl+1}{}_i\in \eqs{\setS}$ and 
$\changedd{\support{\simpleGraph{newl+1}{}_i}}{\Vertices{\simpleGraph{newl+1}{}_i}}{1169}\subseteq 
{\setVertices}$.  
\end{itemize}
Intuitively, we build a sequence of improvements \changed{}{$\simpleGraph{newl+1}{}_i$}{885-890} and consequently 
\changed{$\hypergraphG=-\sum_{i}\simpleGraph{newl+1}{}_i.$}{ the hypergraph, that we need to build,  $\hypergraphG$ will be set to
$-\sum_{i}\simpleGraph{newl+1}{}_i.$}{891}

First, we define \changed{a}{the}{885-890} well-founded quasi-order on hypergraphs, 
\changed{next}{then}{892} we present a single improvement i.e. how to obtain ${\hypergraphH}_{i+1}$ from ${\hypergraphH}_i$,
finally we show that ${\hypergraphH}_{i+1}<{\hypergraphH}_i$ in the defined quasi-order and 
that if the hypergraph \changed{can not}{cannot}{894} be improved anymore then it is
\changed{almost}{pre}{777} \changed{${\newl}+1$-}{$({\newl}+1)$-}{779}isolated.

{\bf Order}. First\changed{}{,}{896} we impose \changed{a}{an arbitrary}{896} \changed{well founded}{well-founded}{896} linear order on the set of all vertices.
\changed{A set of vertices $\setX_1$ 
is smaller than a set $\setX_2$}{For two sets of vertices $\setX_1$ and $\setX_2$, we write $\setX_1<\setX_2$}{897} if there is a bijection $f\colon 
\setX_1\rightarrow 
\setX_2$ such that $\da \leq f(\da)$ for every $\da\in \setX_1$. 
Finally, \changed{a hypergraph ${\hypergraphH}'=({\setVertices}',\mu') \leq_{{\newl}+1}{\hypergraphH}$}{for two hypergraphs $\hypergraphH=(\setVertices,\mu)$ 
and $\hypergraphH'=(\setVertices',\mu')$ we write $\hypergraphH'\leq_{{\newl}+1} \hypergraphH$}{898} if for every 
$\setX'\in\kset{{\newl}+1}{{\setVertices}'}$ such that $\weight{setX'}({\hypergraphH}')\neq \vec{0}$ there is a set $\setX\in \kset{{\newl}+1}{{\setVertices}}$ 
such that 
$\weight{setX}({\hypergraphH})\neq \vec{0}$ and \changed{$\setX$ is bigger than $\setX'$}{$\setX'\leq \setX$}{899}.

The strict inequality $<_{{\newl}+1}$ holds if ${\hypergraphH}'\leq_{{\newl}+1} {\hypergraphH}$ but ${\hypergraphH}'\not\geq_{{\newl}+1} {\hypergraphH}$.

\begin{asm}
From now on until the end of this proof\changed{}{,}{901} whenever we enumerate vertices of some set \changed{$\da_1,\da_2\ldots \da_n$}{$\{\da_1,\da_2\ldots 
\da_n\}$,}{902} 
we assume that $\da_i<\da_j$ for $i<j$.
\end{asm}

{\bf Improvement}. We define the improvement for ${\hypergraphH}_i$.
Let $\familyL\subset \kset{{\newl}+1}{\changedd{\support{{\hypergraphH}_i}}{\Vertices{{\hypergraphH}_i}}{1169}}$ \changed{$\setL\in \familyL \implies 
\weight{setL}({\hypergraphH}_i)\neq \vec{0}$}{be the 
family of all sets such that for every set 
$\setL\in \familyL$ it holds that $\weight{setL}({\hypergraphH}_i)\neq \vec{0}$}{903}. Note that, if ${\newl}+1=\setSizeK$ then $\familyL$ is the set of 
hyperedges of  
${\hypergraphH}_i$.

Suppose $\setL\in \familyL$ is maximal in $\familyL$ and there is a set $\setL'\in\kset{{\newl}+1}{{\setVertices}}$ disjoint from $\setL$ such that 
$\setL'<\setL$.
Let $\setL=\{\da_1, \da_2 \ldots \da_{{\newl}+1}\}$. Let $\simpleGraph{newl+1}{}_i\in\eqs{\setS}$ be an 
$({\newl}+1,\weight{setL}({\hypergraphH}_i))$-\simple\ hypergraph \changed{supported by ${\setVertices}$}{such that 
$\changedd{\support{\simpleGraph{newl+1}{}_i}}{\Vertices{\simpleGraph{newl+1}{}_i}}{1169}\subseteq 
{\setVertices}$, and}{759, 907}  
such that \changed{(according to Definition~\ref{def:simple_hypergraphs})}{(using the notation of Definition~\ref{def:simple_hypergraphs})}{907} 
$\setA=\setL$ and $\setB=\setL'=\{f^{-1}(\da_1),f^{-1}(\da_2)\ldots $ $f^{-1}(\da_{{\newl}+1})\}$, where $f$ witnesses $\setL'<\setL$. \changed{It}{The 
hypergraph $\simpleGraph{newl+1}{}_i$}{908} exists \changed{due 
to}{because of}{200} 
Lemma~\ref{lem:adding_simple_graphs}.

Let ${\hypergraphH}_{i+1}\eqdef {\hypergraphH}_i-\simpleGraph{newl+1}{}_i$.
We claim ${\hypergraphH}_{i+1}<_{{\newl}+1} {\hypergraphH}_i$. 
\changed{
  Indeed, $\weight{setL}({\hypergraphH}_{i+1})=\vec{0}$ and $\setL$ \changed{was}{is}{911} maximal in $\familyL$ so it \changed{}{is}{911} sufficient to 
  prove that for any 
  $\edgeK\in\kset{{\newl}+1}{\support{\simpleGraph{newl+1}{}_i}}$ such that $\edgeK\neq \setL$ and $\weight{edgeK}(\simpleGraph{newl+1}{}_i)\neq \vec{0}$ 
\changed{holds}{it holds that}{451} 
  $\edgeK< \setL$. 
  But $\weight{edgeK}(\simpleGraph{newl+1}{}_i)\neq \vec{0}$ if\changed{}{,}{325} and only if\changed{}{,}{325} $\edgeK$ is a subset 
  of $\setL'\cup \setL$ that contains exactly one element from each pair of vertices 
  $(\da_i,f^{-1}(\da_i))$ (Definition~\ref{def:simple_hypergraphs}). As $\da_i>f^{-1}(\da_i)$ for every $i$ then trivially $\setL=\{\da_1,\da_2\ldots 
  \da_{{\newl}+1}\}>\edgeK$, as required.
}{
First, we show inequality, next we show that it is strict. For the inequality we take any subset $\edgeK\in\kset{{\newl}+1}{{\setVertices}}$ such that
$\weight{edgeK}(\hypergraphH_{i+1})\neq \vec{0}$ and we find a subset $\edgeK'\in\kset{{\newl}+1}{{\setVertices}}$ such that
$\weight{edgeKe}(\hypergraphH_{i  })\neq \vec{0}$ and         $\edgeK\leq \edgeK'$. We have to consider two cases:
\begin{itemize}
 \item $\weight{edgeK}(\hypergraphH_{i+1})\neq \vec{0}$ because of $\weight{edgeK}(\hypergraphH_{i})\neq \vec{0}$ and
\item $\weight{edgeK}(\hypergraphH_{i+1})\neq \vec{0}$ because of $\weight{edgeK}(\simpleGraph{newl+1}{}_i)\neq \vec{0}.$
 \end{itemize} 
In the first case $\edgeK'=\edgeK$. In the second case $\edgeK'=\setL$. We observe that $\edgeK\leq \setL$ as $\edgeK=\{x_1,x_2\ldots x_{\newl+1}\}$ 
where $x_i\in \{\da_i, f^{-1}(\da_i)\}$ and $\setL=\{\da_1,\da_2\ldots \da_{\newl+1}\}$ which is the greatest set (with respect to the order $<$) among such 
sets.\\
Now\changed{}{,}{NULL} to prove strictness of the inequality we observe that $\weight{setL}(\hypergraphH_{i+1})=\vec{0}$ 
and that there is no set $\edgeK$ such that $\setL\leq 
\edgeK$ and $\weight{edgeK}{(\hypergraphH_{i+1})}\neq 0$, so $\hypergraphH_{i}\not\leq\hypergraphH_{i+1}.$
The formal proof requires the same case analysis as the proof of the \changed{non-strict}{nonstrict}{107} inequality. In the first case, we use maximality of 
$\setL$ in the family 
$\familyL$. In the second case, we use the fact that $\setL$ is is the greatest set (with respect to the order $<$) among possible sets.
}{911-913,915}

{\bf Reduced form.}
We call a hypergraph \emph{reduced} if one \changed{can not}{cannot}{357} improve it any further.
If we consequently improve a given hypergraph then eventually we reach ${\hypergraphH}_{last}$ which is reduced.
Indeed, every improvement goes down in the quasi-order on hypergraphs and the quasi-order is trivially well-founded.
What remains to prove is the following claim.

\begin{clm}
A reduced hypergraph is \changed{almost}{pre}{777} \changed{${\newl}+1$-}{$({\newl}+1)$-}{779}isolated. 
\end{clm}
{\it Proof of the claim.}
First observe that ${\hypergraphH}_{last}$ is ${\newl}$-isolated as ${\hypergraphH}$ and each $\simpleGraph{newl+1}{}_i$ are ${\newl}$-isolated (by 
Definition~\ref{def:simple_hypergraphs} 
every $({\newl}+1,{\placeHolder})$-\simple\ 
hypergraph 
is ${\newl}$-isolated).

Observe \changed{}{ that}{929} $\changedd{\support{{\hypergraphH}_{last}}}{\Vertices{{\hypergraphH}_{last}}}{1169}\subseteq {\setVertices}$. Suppose 
${\setVertices} =\{\da_1,\da_2\ldots 
\da_{\size{{\setVertices}}}\}$.  
Let $\familyL \subseteq \kset{{\newl}+1}{{\setVertices}}$ be \changed{a}{the}{NULL} family of sets of vertices such that for any $\setL\in \familyL$ we have 
$\weight{setL}({\hypergraphH}_{\changed{i}{last}{NULL}})\neq \vec{0}$ and $\setL$ 
contains a vertex
with index greater than $2{\newl}+1$. Our goal is to prove that $\familyL$ is empty; this implies the claim.

We prove it \changed{via}{by}{933} contradiction. Suppose $\familyL\neq \emptyset.$
For \changed{a}{any}{922} set $\setL\in \familyL$, $\setL=\{\da_{i_1},\da_{i_2}\ldots \da_{i_{{\newl}+1}}\}$ we introduce $\setY_{\setL}\subseteq \setL$ 
the maximal (in the sense of inclusion) set of vertices 
$\{\da_{i_j},\da_{i_{j+1}}\ldots\da_{i_{{\newl}+1}}\}$
such that $i_j\geq 2j,$ $i_{j+1}\geq 2(j+1),$ $\ldots i_{{\newl}+1}\geq 2({\newl}+1)$. By definition \changed{}{of the family $\familyF$,}{935} for every 
$\setL\in \familyF$ the set \changed{$\setY_{\setL}\neq 
\emptyset.$}{$\setY_{\setL}$ is nonempty.}{936}

Let $\setL\in \familyL$ be a set such that $\setY_{\setL}$ is maximal 
in the terms of size\changed{, i.e.}{\ i.e.}{22} for every $\edgeK\in \familyL$ \changed{holds}{it holds that}{451} $\size{\setY_{\setL}}\geq 
\size{\setY_{\edgeK}}$.
We show the contradiction by proving $\weight{setL}({\hypergraphH}_{last})=\vec{0}$. 

Observe that $\setL\neq \setY_{\setL}$ as in this case ${\hypergraphH}_{last}$ would not be reduced.

\changed{Take,}{Let}{940} ${\hypergraphG}'\eqdef 
\changed{\cut{{\hypergraphH}_{last}}{{\setY_{\setL}}}}{\cut{{(\hypergraphH}_{last})}{{\setY_{\setL}}}}{940}$ (recall Definition~\ref{def:cut vertex}). As 
$\setL\neq 
\setY_{\setL}$
we know that $\weight{setL}({\hypergraphH}_{last})=\weight{L-Y}({\hypergraphG}')$. So, if we prove that ${\hypergraphG}'$ is 
$(\size{\setL}-\size{\setY_{\setL}})$-isolated then 
$\vec{0}=\weight{L-Y}({\hypergraphG})=\weight{setL}({\hypergraphH}_{last})$ and we have the contradiction with the assumption 
that 
$\weight{setL}({\hypergraphH}_{last})\neq \vec{0}$.

So what remains, is to prove the following claim.
\begin{clm}
${\hypergraphG}'$ is  $(\size{\setL}-\size{\setY_{\setL}})$-isolated. 
\end{clm}
\changed{Track the proof on the Figure~\ref{fig:3}.}{}{946}

Let ${\newl}'=(\size{\setL}-\size{\setY_{\setL}})$.
It suffices to prove that ${\hypergraphG}'$ is \changed{almost}{pre}{777} ${\newl}'$-isolated. 
\changed{Indeed, then}{Then}{NULL} \changed{due to}{because of}{200} Lemma~\ref{lem:small k-1-isolated graphs vanish} 
\changed{}{the hypergraph}{NULL} ${\hypergraphG}'$ is ${\newl}'$-isolated.

Note that ${\hypergraphG}'$ is ${\newl}'-1$-isolated which is inherited from ${\hypergraphH}_{last}$, 
\changed{due to}{because of}{200} Lemma~\ref{lem:isolation_down}. So it remains to analyse weights 
of ${\newl}'$-subsets of ${\setVertices}$. 
\changed{We have to prove that every}{Our goal is to show that, for every}{950} $\edgeK\in\kset{{\newl}'}{{\setVertices}}$ if 
$\weight{edgeK}({\hypergraphG'})\neq \vec{0}$ then it is a subset of $\{\da_1\ldots 
\da_{2{\newl}'-1}\}$.

Let us take $\edgeK=\{\da_{i'_1},\da_{i'_2}\ldots \da_{i'_{{\newl}'}}\}$, such that $ {i'_{{\newl}'}}\geq 2{\newl}'$.
We prove that $\weight{edgeK}({\hypergraphG'})=\vec{0}$.
\changed{Observe\changed{,}{\ that}{954} $\setY_{\setL}\subseteq \setY_{\setL} \cup \{\da_{i'_{{\newl}'}}\}\subseteq
\setY_{(\edgeK\cup \setY_{\setL})}$.
Indeed, $\setY_{\setL}$ guarantees that in $\edgeK\cup \setY_{\setL}$ there is at least one element not smaller than $2({\newl}+1)$, at least 
two elements not smaller than $2{\newl}$\ldots at least $\size{\setY_{\setL}}$
 elements not smaller than $2({\newl}'+1)$. Moreover, because of ${i'_{{\newl}'}}\geq 2{\newl}'$, there are at least $\size{\setY_{\setL}}+1$ elements
 not smaller than $2{\newl}'$.
 }
 {
 Let us consider the set $\setLdwa=\edgeK\cup \setY_{\setL}$. We claim that $\size{\setY_{\setL}}< \size{\setY_{\setLdwa}}$.
 To see this we observe that in the set $\setLdwa$, because of the definition of the set $\setY_{\setL}$,
  there is at least one element not smaller than $\da_{2({\newl}+1)}$, at least 
 two elements not smaller than $\da_{2{\newl}}$\ldots at least $\size{\setY_{\setL}}$
  elements not smaller than $\da_{2({\newl}'+1)}$. Moreover, because\changed{ of}{}{956} ${i'_{{\newl}'}}\geq 2{\newl}'$, there are at least 
$\size{\setY_{\setL}}+1$ elements
  not smaller than $\da_{2{\newl}'}$. \changedd{This implies that $({\newl}+1)- ({\newl}'-1)\leq \size{\setY_{\setLdwa}}$ where $\size{\setY_{\setL}}=({\newl}+1)- 
{\newl}'$.}{So $\setY_{\setL}\cup\{\da_{2{\newl}'}\}\subseteq \setY_{\setLdwa}$.}{null}   
 }{NULL}

\changed{Due to}{Because of}{200} maximality of \changed{$\setL$}{$\setY_{\setL}$}{NULL}\changed{}{,}{958} we conclude that 
$\weight{setLdwa}({\hypergraphH}_{last})=\vec{0}$, thus 
$\weight{setLdwa-Y}(\cut{{\hypergraphH}_{last}}{{\setY_{\setL}}})
=\weight{edgeK}({\hypergraphG'})=\vec{0}$.
\changed{Therefore every $\edgeK$ an ${\newl}'$-set 
of vertices in ${\hypergraphG'}$, for which
}
{Therefore for every 
${\newl}'$-set of vertices $\edgeK$ 
in $\kset{{\newl}'}{\setVertices}$, such that}{959} $\weight{edgeK}\changed{}{(\hypergraphG')}{959}\neq \vec{0}$, \changed{}{it}{959} is a subset of $\{\da_1, 
\da_2\ldots \da_{2{\newl}'-1}\}$.
This ends the proof that ${\hypergraphG}'$ is \changed{almost}{pre}{777} ${\newl}'$-isolated.

Finally, as it was written earlier, from Lemma~\ref{lem:small k-1-isolated graphs vanish} 
we derive that ${\hypergraphG}'$ is ${\newl}'$-isolated.
\changed{
\noindent
\begin{figure}[t]
\begin{tikzpicture}
[place/.style={circle,draw=\mycolorOne!50,fill=\mycolorOne!20,thick,inner sep=0pt,minimum size=1mm}]
\pgfmathsetmacro{\B}{0.9}
\node[] (ZERO) at (0,1) {};
\node[place, color=black] (A0) at (0,0) {};
\node[place, color=black] (A1) at (1*\B,0) {};
\node[place, color=black] (A2) at (2*\B,0) {};
\node[place, color=black] (A3) at (3*\B,0) {};
\node[place, color=black] (A4) at (4*\B,0) {};
\node[place, color=black] (A5) at (5*\B,0) {};
\node[place, color=black] (A6) at (6*\B,0) {};
\node[place, color=black] (A7) at (7*\B,0) {};
\node[place, color=black] (A8) at (8*\B,0) {};
\node[place, color=black] (A9) at (9*\B,0) {};
\node[place, color=black] (A10) at (10*\B,0) {};
\node[place, color=black] (A11) at (11*\B,0) {};
\node[place, color=black] (A12) at (12*\B,0) {};
\node[place, color=black] (A13) at (13*\B,0) {};
\node[place, color=black] (A14) at (14*\B,0) {};
\node[place, color=black] (A15) at (15*\B,0) {};
\node[] (B0) at (0*\B,0.5) {$v_{1}$};
\node[] (B1) at (1*\B,0.5) {$v_{2}$};
\node[] (B2) at (2*\B,0.5) {$v_{3}$};
\node[] (B3) at (3*\B,0.5) {$v_{4}$};
\node[] (B4) at (4*\B,0.5) {$v_{5}$};
\node[] (B5) at (5*\B,0.5) {$v_{6}$};
\node[] (B6) at (6*\B,0.5) {$v_{7}$};
\node[] (B7) at (7*\B,0.5) {$v_{8}$};
\node[] (B8) at (8*\B,0.5) {$v_{9}$};
\node[] (B9) at (9*\B,0.5) {$v_{10}$};
\node[] (B10) at (10*\B,0.5) {$v_{11}$};
\node[] (B11) at (11*\B,0.5) {$v_{12}$};
\node[] (B12) at (12*\B,0.5) {$v_{13}$};
\node[] (B13) at (13*\B,0.5) {$v_{14}$};
\node[] (B14) at (14*\B,0.5) {$v_{15}$};
\node[] (B15) at (15*\B,0.5) {$v_{16}$};
\node[place, color=mycol2] (C0) at (0*\B,-1) {};
\node[place, color=mycol2] (C2) at (2*\B,-1) {};
\node[place, color=mycol2] (C3) at (3*\B,-1) {};
\node[place, color=mycol2] (C5) at (5*\B,-1) {};
\node[place, color=\mycolorOne] (C9) at (9*\B,-1) {};
\node[place, color=\mycolorOne] (C12) at (12*\B,-1) {};
\node[place, color=\mycolorOne] (C14) at (14*\B,-1) {};
\node[place, color=\mycolorOne] (D0) at (0*\B,-2) {};
\node[place, color=\mycolorOne] (D1) at (1*\B,-2) {};
\node[place, color=\mycolorDwa] (D8) at (8*\B,-2) {};
\node[place, color=\mycolorDwa] (D15) at (15*\B,-2) {};
\draw[\mycolorDwa, very thick, fill=mycol2, opacity=0.15] (-0.5*\B,-0.65) rectangle (17*\B,-1.35);
\draw[\mycolorDwa, very thick, fill=\mycolorOne, opacity=0.35] (7.5*\B,-0.75) rectangle (15*\B,-1.25);
\node[] (E1) at (16*\B,-1) {\red{$\setL$}};
\node[] (E2) at (13*\B,-1) {\textcolor{\mycolorOne}{$\setY_{\setL}$}};
\draw[\mycolorOne, very thick, fill=\mycolorOne, opacity=0.15] (-0.5*\B,-1.65) rectangle (17*\B,-2.35);
\draw[\mycolorOne, very thick, , fill=\mycolorDwa, opacity=0.35] (6.5*\B,-1.75) rectangle (16*\B,-2.25);
\node[] (F1) at (16.5*\B,-2) {\blue{$\setL'$}};
\node[] (F2) at (13*\B,-2) {\textcolor{\mycolorDwa}{$\setY_{\setL'}$}};
\end{tikzpicture} 
\caption{}\label{fig:3}
\end{figure}
}{}{946}
\end{proof}

\section{The construction of simple  hypergraphs.}\label{sec:The construction of simple hypergraphs}

\changed{}{The whole section is devoted to proving the following theorem:
\begin{thm}\label{th:simple_graphs}
Suppose ${\hypergraphG}$ is a ${\setSizeK}$-hypergraph.
Then there is $\setS'$ \changedd{a simple family for $\{{\hypergraphG}\}$}{an $\{{\hypergraphG}\}$-simplified family of hypergraphs}{1169}. Moreover 
$\setS'\subseteq \sums{\Z}{\eqs{{\hypergraphG}}}.$
\end{thm}
\begin{proof}
\changed{}{We want to construct the family $\setS'$. To achieve this,}{975}
\changed{Due to}{because of}{200} Remark~\ref{rem:simple family}\changed{}{,}{975} it is sufficient to propose an algorithm that for any 
${\setSizeK}$-hypergraph ${\hypergraphG}$ produces $\setS$ 
a simplification of $\{{\hypergraphG}\}$ such that $\setS\subseteq \sums{\Z}{\eqs{{\hypergraphG}}}$. Existence \changed{of a such algorithm}{of such 
an algorithm}{976} is 
a consequence of the following lemmas.
\changed{
\begin{lem}\label{lem:simple_graphs}
Let ${\hypergraphG}=({\setVertices}',\mu')$ be a ${\setSizeK}$-hypergraph and 
 $\setX \in \kset{m}{{\setVertices}'}$ where $m\leq \setSizeK$. Let $\vec{a}=\weight{setX}(\hypergraphG)$.
Then 
there is a \changed{$(m,\vec{a})-\simple$}{$(m,\vec{a})$-\simple}{979,980} ${\setSizeK}$-hypergraph 
${\simpleGraph{m}{\vec{a}}}$ such that ${\simpleGraph{m}{\vec{a}}}\in \sums{\Z}{\eqs{{\hypergraphG}}}$.
\end{lem}}
{}{983}
\changed{\begin{lem}\label{lem:simple_sum_is_simple}
\changed{}{Let ${\hypergraphG}$ be a ${\setSizeK}$-hypergraph.}{981} For any $0\leq m\leq \setSizeK$\changed{}{,}{980} if 
\changed{$(m,\vec{a})-$simple}{$(m,\vec{a})$-\simple}{979,980} and $(m,\vec{b})$-\simple\ 
hypergraphs are elements of $\sums{\Z}{\eqs{{\hypergraphG}}}$ then
\changed{$(m,\vec{a}+\vec{b})-$simple}{$(m,\vec{a}+\vec{b})$-\simple}{979,980} hypergraph is also a member of $\sums{\Z}{\eqs{{\hypergraphG}}}$.
\end{lem}}
{
\begin{lem}\label{lem:simple_sum_is_simple}
Let ${\hypergraphG}$ be a ${\setSizeK}$-hypergraph and $0\leq m\leq \setSizeK$. Suppose there are 
$(m,\vec{a})$-\simple\ and $(m,\vec{b})$-\simple\ 
hypergraphs that are elements of $\sums{\Z}{\eqs{{\hypergraphG}}}$. Then
$(m,\vec{a}+\vec{b})$-\simple\ hypergraph is also a member of $\sums{\Z}{\eqs{{\hypergraphG}}}$.
\end{lem}
}{980-981}
\changed{}{
\begin{lem}\label{lem:simple_graphs}
Let \changedd{${\hypergraphG}=({\setVertices}',\mu')$}
{
${\hypergraphG}$ 
}{1169}be a ${\setSizeK}$-hypergraph and 
\changedd{}{ $\setVertices'$ be its set of vertices. Suppose, }{1169}
 $\setX \in \kset{m}{{\setVertices}'}$ where $m\leq \setSizeK$. Let $\vec{a}=\weight{setX}(\hypergraphG)$.
Then 
there is a \changed{$(m,\vec{a})-\simple$}{$(m,\vec{a})$-\simple}{979,980}\ ${\setSizeK}$-hypergraph 
${\simpleGraph{m}{\vec{a}}}$ such that ${\simpleGraph{m}{\vec{a}}}\in \sums{\Z}{\eqs{{\hypergraphG}}}$.
\end{lem}}
{983}
Indeed, using this two lemmas we can produce all elements in the simplification.
\changed{}{Suppose, $m\leq \setSizeK$ and 
$\vec{a}\in \sums{\Z}{{\setof{\weight{setX}(\hypergraphG)}{\setX\in \kset{m}{\setVertices'}}}}$.
We have to show that $(m,\vec{a})$-\simple\ hypergraph is in $\sums{\Z}{\eqs{\hypergraphG}}$.
Suppose, $\vec{a}=c_1\vec{a_1}+c_2\vec{a_2}\ldots c_l\vec{a_l}$
where $\vec{a_i}\in$
$\setof{\weight{setX}(\hypergraphG)}{{\setX\in \kset{m}{\setVertices'}}}$ and $c_i\in \Z$.
Because of Lemma~\ref{lem:simple_sum_is_simple} it suffices to prove that 
$(m,\vec{a_i})$-\simple\ hypergraphs are in $\sums{\Z}{\eqs{\hypergraphG}}$. But this is exactly
Lemma~\ref{lem:simple_graphs}.}{NULL}
So it remains to provide proofs of lemmas~\ref{lem:simple_sum_is_simple}~and~\ref{lem:simple_graphs}.
The proof of Lemma~\ref{lem:simple_sum_is_simple} is easy so we start from it and then we concentrate on the \changedd{much }{}{null}more complicated proof of 
Lemma~\ref{lem:simple_graphs}. 
\end{proof}
}{973}

\changed{
We start \changed{from}{with}{964} an operator that is used in later proofs.
\begin{defi}\label{def:tou_trans}
Suppose ${\hypergraphG}=({\setVertices}', \mu')$ is a ${\setSizeK}$-hypergraph and $\da \in {\setVertices}', \da'\notin 
{\setVertices}'$ 
are two vertices. 
Let $\sigma_{\da} \colon {\setVertices}' \cup \{\da'\} \rightarrow {\setVertices}' \cup \{\da'\}$ such that:
\begin{displaymath}
\sigma_{\da}(x) \eqdef \left\{ \begin{array}{ll}
\da' & \textrm{when $x = {\da}$}\\
\da & \textrm{when $x = {\da}'$}\\
x & \textrm{otherwise}
\end{array} \right.
\end{displaymath}
A \emph{swap of $\da$ and $\da'$ in ${\hypergraphG}$} 
is defined as 
$\swap{\da}{\da'}{\hypergraphG} \eqdef {\hypergraphG} \circ \sigma_{\da}$.
\end{defi}
}{}{965}

\changed{
\begin{thm}\label{th:simple_graphs}
Suppose ${\hypergraphG}$ is a ${\setSizeK}$-hypergraph.
Then there is $\setS'$ a simple family for $\{{\hypergraphG}\}$. Moreover $\setS'\subseteq \sums{\Z}{\eqs{{\hypergraphG}}}.$
\end{thm}
\changed{Due to}{Because of}{200} Remark~\ref{rem:simple family} it is sufficient to propose an algorithm that for any ${\setSizeK}$-hypergraph ${\hypergraphG}$ 
produces $\setS$ 
a simplification of $\{{\hypergraphG}\}$ such that $\setS\subseteq \sums{\Z}{\eqs{{\hypergraphG}}}$. Existence of a such algorithm is 
a consequence of the following lemmas.
\begin{lem}\label{lem:simple_graphs}
Let ${\hypergraphG}=({\setVertices}',\mu')$ be a ${\setSizeK}$-hypergraph and 
 $\setX \in \kset{m}{{\setVertices}'}$ where $m\leq \setSizeK$. Let $\vec{a}=\weight{setX}(\hypergraphG)$.
Then 
there is a \changed{$(m,\vec{a})-\simple$}{$(m,\vec{a})$-\simple}{20} ${\setSizeK}$-hypergraph 
${\simpleGraph{m}{\vec{a}}}$ such that ${\simpleGraph{m}{\vec{a}}}\in \sums{\Z}{\eqs{{\hypergraphG}}}$.
\end{lem}
\begin{lem}\label{lem:simple_sum_is_simple}
For any $0\leq m\leq \setSizeK$ if $(m,\vec{a})-$simple and $(m,\vec{b})$-\simple\ 
hypergraphs are elements of $\sums{\Z}{\eqs{{\hypergraphG}}}$ then
$(m,\vec{a}+\vec{b})-$simple hypergraph is also a member of $\sums{\Z}{\eqs{{\hypergraphG}}}$.
\end{lem}
Indeed, using this two lemmas we can produce all elements in the simplification.
The proof of Lemma~\ref{lem:simple_sum_is_simple} is easy so we start from it and then we concentrate on the much more complicated proof of 
Lemma~\ref{lem:simple_graphs}.
}{}{NULL}

\begin{proof}[Proof of Lemma~\ref{lem:simple_sum_is_simple}]
Let ${\simpleGraph{m}{\vec{a}}}\in \sums{\Z}{\eqs{{\hypergraphG}}}$ \changed{is}{be}{985 \& 986} the $(m,\vec{a})-$simple hypergraph and 
$\simpleGraph{m}{\vec{b}}\in \sums{\Z}{\eqs{{\hypergraphG}}}$ \changed{is}{be}{985 \& 986} the $(m,\vec{b})-$simple hypergraph.
\changed{Thus we}{We}{986} can write 
\[
 {\simpleGraph{m}{\vec{a}}}=\sum_i a_i {\hypergraphG}_i \text{ where } a_i \in \Z \text{ and } {\hypergraphG}_i\in \eqs{{\hypergraphG}}
\]
\[
 \simpleGraph{m}{\vec{b}}=\sum_j b_{j} {\hypergraphG}_j \text{ where } b_{j} \in \Z \text{ and } {\hypergraphG}_j\in \eqs{{\hypergraphG}}
\]

We \changed{recall the definition}{recall the notation used in Definition~\ref{def:simple_hypergraphs}}{992} of \simple\  hypergraphs; the vertices of 
${\simpleGraph{m}{\vec{a}}}$ can be split into ${\setA}_{\simpleGraph{m}{\vec{a}}}, 
 {\setB}_{\simpleGraph{m}{\vec{a}}}, {\setC}_{\simpleGraph{m}{\vec{a}}}$ and similarly 
 vertices of $\simpleGraph{m}{\vec{b}}$ are in ${\setA}_{\simpleGraph{m}{\vec{b}}}, {\setB}_{\simpleGraph{m}{\vec{b}}}, {\setC}_{\simpleGraph{m}{\vec{b}}}$. 
\changed{We assume that sets ${\setA}_{{\placeHolder}}$ and ${\setB}_{{\placeHolder}}$ are ordered like in the definition of \simple\  hypergraphs.}{According 
to 
Definition~\ref{def:simple_hypergraphs}, the elements of sets ${\setA}_{\simpleGraph{m}{\vec{a}}}, 
 {\setB}_{\simpleGraph{m}{\vec{a}}}$ are paired, formally there is a bijection $f:{\setA}_{\simpleGraph{m}{\vec{a}}}\to{} 
 {\setB}_{\simpleGraph{m}{\vec{a}}}$, similarly elements of sets ${\setA}_{\simpleGraph{m}{\vec{b}}}$ and ${\setB}_{\simpleGraph{m}{\vec{b}}}$ are paired, 
formally there is a bijection $f':{\setA}_{\simpleGraph{m}{\vec{b}}}\to{} {\setB}_{\simpleGraph{m}{\vec{b}}}$.}{993} 

Then there is a bijection $\pi$ between vertices that transfers $\setA_{\simpleGraph{m}{\vec{a}}}$ to ${\setA}_{\simpleGraph{m}{\vec{b}}}$, 
$\setB_{\simpleGraph{m}{\vec{a}}}$ to ${\setB}_{\simpleGraph{m}{\vec{b}}}$, $C_{\simpleGraph{m}{\vec{a}}}$ to ${\setC}_{\simpleGraph{m}{\vec{b}}}$, and that 
\changed{preserves the orders 
on vertices in ${\setA}_{{\placeHolder}}$ and ${\setB}_{{\placeHolder}}$.}{it preserves the pairings i.e. $\pi(f(\da_i))=f'(\pi(\da_i))$ and 
$\pi(f^{-1}(\db_i))=f'^{-1}(\pi(\db_i))$ for 
every $\da_i\in \setA_{\simpleGraph{m}{\vec{a}}}$ and $\db_i\in \setB_{\simpleGraph{m}{\vec{a}}}$.}{993}
\changed{Thus }{Thus, }{177}
\[
 {{\simpleGraph{m}{\vec{a}}}} + {\simpleGraph{m}{\vec{b}}}\circ\pi = \sum_i a_i {\hypergraphG}_i + (\sum_j b_{j} {\hypergraphG}_j)\circ \pi = \sum_i a_i 
{\hypergraphG}_i 
+ \sum_j b_{j} ({\hypergraphG}_j)\circ \pi
\]
But this mean that \changed{${{\simpleGraph{m}{\vec{a}}}}+{\simpleGraph{m}{\vec{a}}} \circ \pi\in \sums{\Z}{\eqs{{\hypergraphG}}}$}
{
${{\simpleGraph{m}{\vec{a}}}}+{\simpleGraph{m}{\vec{b}}} \circ \pi\in \sums{\Z}{\eqs{{\hypergraphG}}}$}{1001}. It is not hard to see 
that 
\changed{${{\simpleGraph{m}{\vec{a}}}}+{\simpleGraph{m}{\vec{a}}} \circ \pi$}{
${{\simpleGraph{m}{\vec{a}}}}+{\simpleGraph{m}{\vec{b}}} \circ \pi$
}{1001}
is an $(m,\vec{a}+\vec{b})$-\simple\ hypergraph.
\end{proof}

\changed{}{
We may proceed to the proof of Lemma~\ref{lem:simple_graphs}. We start \changed{from}{with}{964} an 
operator that is used in the following proofs.
\begin{defi}\label{def:tou_trans}
Suppose ${\hypergraphG}\changedd{=({\setVertices}', \mu')}{}{405}$ is a ${\setSizeK}$-hypergraph and $\da \in 
\changedd{{\setVertices}'}{\Vertices{\hypergraphG}}{405}, \da'\notin 
\changedd{{\setVertices}'}{\Vertices{\hypergraphG}}{405}$ 
are two vertices. 
Let $\sigma_{\da} \colon \changedd{{\setVertices}'}{\Vertices{\hypergraphG}}{405} \cup \{\da'\} \rightarrow 
\changedd{{\setVertices}'}{\Vertices{\hypergraphG}}{405} \cup \{\da'\}$ such that:
\begin{displaymath}
\sigma_{\da}(x) \eqdef \left\{ \begin{array}{ll}
\da' & \textrm{when $x = {\da}$}\\
\da & \textrm{when $x = {\da}'$}\\
x & \textrm{otherwise}
\end{array} \right.
\end{displaymath}
A \emph{swap of $\da$ and $\da'$ in ${\hypergraphG}$} 
is defined as 
$\swap{\da}{\da'}{\hypergraphG} \eqdef {\hypergraphG} \circ \sigma_{\da}$.
\end{defi}
}{965}

The proof of Lemma~\ref{lem:simple_graphs} \changed{goes via}{is by}{1003} induction on ${\setSizeK}$. We encapsulate the most important steps of the proof in 
four lemmas. 
Lemma~\ref{lem:simple graph raising} is an auxiliary lemma. 
Lemma~\ref{lem:induction base} forms the induction base. Lemmas 
\ref{lem:induction step one}~and~\ref{lem:induction step two} cover the induction step.

\begin{lem}\label{lem:simple graph raising}
Let \changed{}{$\hypergraphG$ be a $\setSizeK$-hypergraph and}{1007} $\simpleGraph{m}{\vec{a}}\changed{=(\setVertices,\mu)}{}{NULL}$ be an 
\changed{$(m,\vec{a})-\simple$}{$(m,\vec{a})$-\simple}{20} 
${\setSizeK}$-hypergraph\changed{and let}{. Suppose,}{1007} $\da, \da'\notin 
\changed{{\setVertices}}{\changedd{\support{\simpleGraph{m}{\vec{a}}}}{\Vertices{\simpleGraph{m}{\vec{a}}}}{1169}}{NULL}$ 
\changed{be}{are}{1007} two 
vertices\changed{. 
 Suppose}{\ and\ }{1007}$\simpleGraph{m}{\vec{a}}\in \sums{\Z}{\eqs{\cut{{\hypergraphG}}{\da}}}$.
Then: \begin{enumerate}
 \item $\simpleGraph{m+1}{\vec{a}}=\revcut{\simpleGraph{m}{\vec{a}}}{\da} -
\revcut{\simpleGraph{m}{\vec{a}}}{\da'}$ is an \changed{$(m+1,\vec{a})-\simple$}{$(m+1,\vec{a})$-\simple}{20} 
 $(\setSizeK+1)$-hypergraph,
\item $\simpleGraph{m+1}{\vec{a}}\in \sums{\Z}{\eqs{{\hypergraphG}}}$.
 \end{enumerate}
\end{lem}
\begin{proof} 
\changed{{\bf Claim 1.}}{We start by showing Point 1.}{1011} 
Let \changed{${\setVertices}= \setA' \cup \setB' \cup \setC'$, where $\setA', \setB', 
\setC'$}{$\changed{{\setVertices}}{\changedd{\support{\simpleGraph{m}{\vec{a}}}}{\Vertices{\simpleGraph{m}{\vec{a}}}}{1169}}{NULL}= \setA \cup \setB \cup 
\setC$, where $\setA, 
\setB, \setC$}{1011} are \changed{such}{}{1011} as in 
the definition of \changed{$(m,\vec{a})-\simple$}{$(m,\vec{a})$-\simple}{20} hypergraph (Definition~\ref{def:simple_hypergraphs}). 
Let \changed{$\setA' = \{\da_1, {\da}_2 \ldots {\da}_m\}$ and $\setB' = 
\{{\db}_1, {\db}_2 \ldots {\db}_m\}.$}
{
$\setA = \{\da_1, {\da}_2 \ldots {\da}_m\}$ and $\setB = 
\{{\db}_1, {\db}_2 \ldots {\db}_m\}.$
}{1011}
We define \changed{$\setA\eqdef \{{\da}_1, {\da}_2 \ldots {\da}_m, {\da}_{m+1}=\da\}$ and 
$\setB \eqdef 
\{{\db}_1,$ ${\db}_2 \ldots {\db}_m, {\db}_{m+1}=\da'\}.$}
{
$\setA'\eqdef \{{\da}_1, {\da}_2 \ldots {\da}_m, {\da}_{m+1}=\da\}$ and 
$\setB' \eqdef 
\{{\db}_1,$ ${\db}_2 \ldots {\db}_m, {\db}_{m+1}=\da'\}.$
}{1011}
To show that $\simpleGraph{m+1}{\vec{a}}$ is \changed{$(m+1,\vec{a})-\simple$}{$(m+1,\vec{a})$-\simple}{20} we split 
$\changed{{\setVertices}}{
\changedd{
\support{\simpleGraph{m}{\vec{a}}}
}{
\Vertices{\simpleGraph{m}{\vec{a}}}
}{1169}
}{NULL}\cup\{\da,\da'\}$ to 
\changed{$\setA, \setB, $}{$\setA', \setB', 
$}{1011} and \changed{$\setC'$}{$\setC$}{NULL} and we verify 
\changed{the properties}{Properties}{1014-1033} $1$ to $7$. \changed{The properties}{Properties}{1014-1033} $1$ to $3$ are trivial.
Properties $4$ to $6$ speak about function $\weight{setX}({\simpleGraph{m+1}{\vec{a}}})$ where:
\begin{itemize}
 \item \changed{The property}{Property}{1014-1033} $4$. ${\setX}$ contains exactly one vertex
from every pair ${\da}_i, {\db}_i$ for $0<i\leq m+1$. \changed{We consider only one of the two 
cases}{Here, there are two cases: $\da\in \setX$ or $\da'\in \setX$. We consider only one of them}{1018}
as the second one is \changed{very\ }{}{NULL}similar. Suppose\changed{,}{\ that}{1018} $\da'\in {\setX}$. 
Then $\weight{setX}({\revcut{\simpleGraph{m}{\vec{a}}}{\da}})=\vec{0}$ so 
\[
\weight{setX}({\simpleGraph{m+1}{\vec{a}}})=-\weight{setX}({\revcut{\simpleGraph{m}{\vec{a}}}{\da'}})=
-\weight{setX-da}({\simpleGraph{m}{\vec{a}}}).
\]\changed{Thus }{Thus, }{177}$\weight{setX}({\simpleGraph{m+1}{\vec{a}}})=-1\cdot (-1)^{\size{{\setX}\cap 
\changed{\setB'}{\setB}{1011}}}\vec{a}=(-1)^{\size{{\setX}\cap \changed{B}{\setB}{1022}}}\vec{a}$, as required.
\item \changed{The property}{Property}{1014-1033} $5$. ${\setX}$ has $m+1$ elements but it is not\changed{ of the type considered in}{\ one of the sets 
considered in}{1023}
\changed{the property}{Property}{1014-1033} $4$. In this case 
$\weight{setX}({\revcut{\simpleGraph{m}{\vec{a}}}{\da}})=\weight{setX}({\revcut{\simpleGraph{m}{\vec{a}}}{\da'}})=\vec{0}$, thus
\changed{$\weight{setX}({\revcut{\simpleGraph{m}{\vec{a}}}{\da}})=\vec{0}$}
{$\weight{setX}({\simpleGraph{m}{\vec{a}}})=\vec{0}$}
{1024}, as required.
\item \changed{The property}{Property}{1014-1033} $6$. $\size{{\setX}}\leq m$. We consider four cases: 
\changed{
${\setX}\cap \{\da,\da'\}=\{\da,\da'\}$, 
${\setX}\cap \set{\da,\da'}=\set{\da}$,
${\setX}\cap \set{\da,\da'}=\set{\da'}$, and
${\setX}\cap \{\da,\da'\}= \emptyset$.
}{
\begin{enumerate}
 \item $\da$ and $\da'$ belong to the set $\setX$,
 \item $\da$ belongs to the set $\setX$, but $\da'$ does not, 
 \item $\da'$ belongs to the set $\setX$, but $\da$ does not,
 \item both $\da,\da'$ do not belong $\setX$.
\end{enumerate}
}{1025}

\vspace{0.2cm}

\begin{itemize}
 \item In the first case, $\da$ is 
not a vertex of 
${\revcut{\simpleGraph{m}{\vec{a}}}{\da'}}$ 
and $\da'$ is not a vertex of 
${\revcut{\simpleGraph{m}{\vec{a}}}{\da}}.$ 
So, $\vec{0}=\weight{setX}({\revcut{\simpleGraph{m}{\vec{a}}}{\da}})=\weight{setX}({\revcut{\simpleGraph{m}{\vec{a}}}{ \da' }}).$ Thus, 
$\weight{setX}(\simpleGraph{m+1}{\vec{a}})=\vec{0}$, as 
required.
\item In the second case, $\vec{0}=\weight{setX}({\revcut{\simpleGraph{m}{\vec{a}}}{\da'}})$ and 
$\weight{setX-da}({\simpleGraph{m}{\vec{a}}})=\vec{0}$ as $\size{\setX\setminus \set{\da}}<m$. Thus, $\weight{setX}(\simpleGraph{m+1}{\vec{a}})=\vec{0}$, as 
required.
\item The third case is almost the same as second.
\item In the fourth case, 
$\weight{setX}({\revcut{\simpleGraph{m}{\vec{a}}}{\da}})=\weight{setX}({\revcut{\simpleGraph{m}{\vec{a}}}{\da'}})$, so 
$\weight{setX}({\simpleGraph{m+1}{\vec{a}}})=$ \\
$\weight{setX}({\revcut{\simpleGraph{m}{\vec{a}}}{\da}})-\weight{setX}({\revcut{\simpleGraph{m}{\vec{a}}}{\da'}})=\vec{0},$ as required. 
\end{itemize}

\end{itemize}

\changed{{\bf Claim 2.}}{We proceed to the proof of Point 2.}{1011} Let $\simpleGraph{m}{\vec{a}}=\sum_i a_i {\hypergraphG}_i$ where $a_i\in \Z$, 
${\hypergraphG}_i\in \eqs{\cut{{\hypergraphG}}{\da}},$ and 
additionaly $\da'\not\in \bigcup_i \changedd{
\support{\hypergraphG_i} \cup  \support{\hypergraphG}
}{
\Vertices{\hypergraphG_i} \cup  \Vertices{\hypergraphG}
}{1169}
.$ Suppose 
$\pi_i$ are bijections \changed{that are identity on $\set{\da,\da'}$ and that
witness
equivalence of summands and $\cut{{\hypergraphG}}{\da}$ i.e. ${\hypergraphG}_i=\cut{{\hypergraphG}}{\da}\circ \pi_i.$}{
such that $\pi_i$ are identity on $\set{\da,\da'}$ and 
${\hypergraphG}_i=\cut{{\hypergraphG}}{\da}\circ \pi_i.$
}{1037}%
\changed{They exist as sets of vertices of $\cut{{\hypergraphG}}{\da}$ and $\hypergraphG_i$ do not contain neither $\da$ nor $\da'$.}{
The bijections $\pi_i$ are well defined as sets of vertices of $\cut{{\hypergraphG}}{\da}$ and $\hypergraphG_i$ do not contain neither $\da$ nor $\da'$.
}{1037}%
\changed{Due to}{Because of}{200} \changed{point 1}{Point $1$}{1038}, we may write the following equation \[
\simpleGraph{m+1}{\vec{a}}=\sum_i a_i 
\revcut{(\cut{{\hypergraphG}}{\da}\circ \pi_i)}{\da}-a_i\revcut{(\cut{{\hypergraphG}}{\da}\circ \pi_i)}{\da'}.
\]
\changed{As}{}{1044} $\pi_i$ \changed{is 
identity}{is the identity}{1043} on 
$\da,\da'$ \changed{we may transform the equation}{, so}{1044}. 
\[
\simpleGraph{m+1}{\vec{a}}=\sum_i a_i \revcut{\cut{{\hypergraphG}}{\da}}{\da}\circ \pi_i-a_i \revcut{\cut{{\hypergraphG}}{\da}}{\da'}\circ \pi_i.
\]
To prove \changed{the second claim}{Point 2}{1011} it suffices to prove \changed{}{that\ }{NULL}for each $i$\changed{}{, it holds that}{NULL} 
\begin{equation}\label{eq:simple graph raising}
 a_i \revcut{\cut{{\hypergraphG}}{\da}}{\da}\circ \pi_i-a_i \revcut{\cut{{\hypergraphG}}{\da}}{\da'} \circ \pi_i= a_i {\hypergraphG}\circ \pi_i -a_i 
{\hypergraphG}'\circ \pi_i 
\end{equation}
where ${\hypergraphG}' = \swap{\da}{\da'}{\hypergraphG}$.
We simplify further
\begin{equation}\label{eq:simple graph raising1}
 \revcut{\cut{{\hypergraphG}}{\da}}{\da}\circ \pi_i-\revcut{\cut{{\hypergraphG}}{\da}}{\da'} \circ \pi_i= {\hypergraphG} \circ \pi_i- {\hypergraphG}' \circ 
\pi_i 
\end{equation}

\changed{We prove~Equation~\ref{eq:simple graph raising1} by showing equality of functions $\weight{setX}$ on both sides.}
{
To prove~Equation~\ref{eq:simple graph raising1} we have to show that for every $\setSizeK$-subset 
\changedd{
$\support{\revcut{\cut{\hypergraphG}{\da}}{\da}\circ 
\pi_i}\cup
\support{\revcut{\cut{\hypergraphG}{\da}}{\da'}\circ \pi_i}  \cup \support{\hypergraphG \circ \pi_i} \cup \support{\hypergraphG'\circ \pi_i}$
}
{
$\Vertices{\revcut{\cut{\hypergraphG}{\da}}{\da}\circ 
\pi_i}\cup
\Vertices{\revcut{\cut{\hypergraphG}{\da}}{\da'}\circ \pi_i}  \cup \Vertices{\hypergraphG \circ \pi_i} \cup \Vertices{\hypergraphG'\circ \pi_i}$
}{1169}
(for every 
potential 
hyperedge)
the weight of that subset on both sides of~Equation~\ref{eq:simple graph raising1} is the same. 
\changedd{
$\support{\revcut{\cut{\hypergraphG}{\da}}{\da}\circ 
{\pi_i}}\cup  
\support{\revcut{\cut{\hypergraphG}{\da}}{\da'}\circ \pi_i}  \cup \support{\hypergraphG \circ \pi_i} \cup \support{\hypergraphG'\circ \pi_i} 
={\pi_i}^{-1}(\support{\hypergraphG}\cup\support{\hypergraphG'}),$}
{
$\Vertices{\revcut{\cut{\hypergraphG}{\da}}{\da}\circ 
{\pi_i}}\cup  
\Vertices{\revcut{\cut{\hypergraphG}{\da}}{\da'}\circ \pi_i}  \cup \Vertices{\hypergraphG \circ \pi_i} \cup \Vertices{\hypergraphG'\circ \pi_i} 
={\pi_i}^{-1}(\Vertices{\hypergraphG}\cup\Vertices{\hypergraphG'}),$
}{1169}so we formalise this as a following equation. 
}{1054,1057}
\changed{So for}{For}{1054,1057} any 
\changedd{
${\setX}\subseteq \pi_i\changed{^{-1}}{}{NULL}(\support{{\hypergraphG}}\cup \support{{\hypergraphG}'})$
}
{
${\setX}\subseteq \pi_i\changed{^{-1}}{}{NULL}(\Vertices{{\hypergraphG}}\cup \Vertices{{\hypergraphG}'})$
}{1169}
we prove
\begin{equation}\label{eq:simple graph raising2}
\weight{setX} \left(\revcut{\cut{{\hypergraphG}}{\da}}{\da} \circ \pi_i- \revcut{\cut{{\hypergraphG}}{\da}}{\da'}\circ \pi_i\right)= \weight{setX} 
\left( {\hypergraphG} \circ \pi_i- 
{\hypergraphG}' \circ \pi_i\right). 
\end{equation}
We consider $4$ cases depending on whether $\da,\da'\in \setX$.
\begin{itemize}
 \item If $\setX$ does not contain $\da$ and $\da'$ then 
$\weight{setX}({\hypergraphG}\circ \pi_i)=\weight{setX}({\hypergraphG}'\circ \pi_i)$, so the \changed{right hand}{right}{1059} side \changed{equals}{is equal 
to}{1060} $\vec{0}$
\changed{exactly like the \changed{left-hand}{left}{1057} side of the equality}
{
. On the left side $\weight{setX}(\cut{\hypergraphG}{\da}\circ \pi_i) = \weight{setX}(\revcut{\cut{\hypergraphG}{\da}}{\da}\circ \pi_i)= 
\weight{setX}(\revcut{\cut{\hypergraphG}{\da}}{\da'}\circ \pi_i)$ so the left side is equal to $\vec{0}$.  
}{1060}.
\item If $\setX$ contains only $\da$ then $\weight{setX}(\revcut{\cut{{\hypergraphG}}{\da}}{\da}\circ \pi_i)=\weight{setX}({\hypergraphG}\circ \pi_i)$
and $\weight{setX}(\revcut{\cut{{\hypergraphG}}{\da}}{\da'}\circ \pi_i)=\vec{0}=\weight{setX}({\hypergraphG}'\circ \pi_i)$, so the equality holds.
\item \changed{Similarly for only $\da'$.}{The case where $\setX$ contains only $\da'$ is similar.}{1063} 
\item If $\setX$ contains both $\da$ and $\da'$ then
$\vec{0}=\weight{setX}(\revcut{\cut{{\hypergraphG}}{\da}}{\da}\circ \pi_i)=\weight{setX}({\hypergraphG}\circ \pi_i)=
\weight{setX}(\revcut{\cut{{\hypergraphG}}{\da}}{\da'}\circ \pi_i)=\weight{setX}({\hypergraphG}'\circ \pi_i)$. So both sides of \changed{the 
equality}{Equation~\ref{eq:simple graph raising2}}{1067} are $\vec{0}$.
\end{itemize}

\changed{Thus }{Thus, }{177}Equality~\ref{eq:simple graph raising1} holds for 
every $\setX$.
\end{proof}

\begin{lemC}[{\cite[Th 15]{DBLP:conf/lics/HofmanLT17}}]
% \changed{\footnote{This lemma is reformulated Theorem~15 from \cite{DBLP:conf/lics/HofmanLT17}.}}
% {
% \textnormal{[From \cite{DBLP:conf/lics/HofmanLT17}, Th 15.]}
% }{1068}
\label{lem:induction base}
Let ${\hypergraphG}=({\setVertices\changed{}{'}{NULL}},\mu\changed{}{'}{NULL})$ be a $1$-hypergraph. 
Then for any $\setX\subseteq {\setVertices\changed{}{'}{NULL}}$ and $\size{\setX}\leq 1$ there is a 
\changed{$(\size{\setX},\weight{setX}(\hypergraphG))-\simple$}{$(\size{\setX},\weight{setX}(\hypergraphG))$-\simple}{20} 
$1$-hypergraph $\simpleGraph{setX}{\weight{setX}(\hypergraphG)}$, such that 
$\simpleGraph{setX}{\weight{setX}(\hypergraphG)}\in \sums{\Z}{\eqs{{\hypergraphG}}}$.
\end{lemC}
\begin{proof}
We consider two cases (i) $\size{\setX}=1$ and (ii) $\size{\setX}=0$.

The first case. Let ${\setX}=\{\da\}$, $\vec{a}=\weight{da}({\hypergraphG})$, and $\da'\notin {\setVertices\changed{}{'}{NULL}}$. 
We define $\simpleGraph{1}{\vec{a}} \eqdef {\hypergraphG}-\swap{\da}{\da'}{\hypergraphG}.$ 
For any vertex $\db \neq \da,\da'$ we have $\weight{db}(\simpleGraph{1}{\vec{a}}) = \vec{0}$. So $\simpleGraph{1}{\vec{a}}$ has two 
\changed{non-isolated}{nonisolated}{107} vertices 
$\{\da,\da'\}$ and satisfies all properties of a 
\changed{simple}{$(1,\vec{a})$-\simple\ }{NULL} 
$1$-hypergraph.

The second case. Let $\weight{emptyset}({\hypergraphG})= \vec{b}$, $\da'\notin {\setVertices\changed{}{'}{NULL}}$. Further, let 
$\simpleGraph{1}{\mu\changed{}{'}{NULL}(\db)} = 
{\hypergraphG}-\swap{\db}{\da'}{\hypergraphG}$ \changed{and}{}{1076} for any $\db\in {\setVertices\changed{}{'}{NULL}}$.
Then we put $\simpleGraph{0}{\vec{b}}\eqdef {\hypergraphG}-\left(\sum_{\db\in {\setVertices\changed{}{'}{NULL}}} \simpleGraph{1}{\mu\changed{}{'}{NULL}(\db)}    
 \right)$. 
Indeed, it is a hypergraph with only one \changed{non-isolated}{nonisolated}{107} 
vertex $\da'$ and 
$\weight{da'}(\simpleGraph{0}{\vec{b}})= -\sum_{\db\in 
{\setVertices\changed{}{'}{NULL}}}\weight{da'}({\simpleGraph{1}{\mu\changed{}{'}{NULL}(\db)}})=\weight{emptyset}({\hypergraphG}),$ as 
required. 
\end{proof}

\begin{lem}\label{lem:induction step one}
\changed{}{Let $\setSizeK\in \N$.}{1081}
Suppose that Theorem~\ref{th:simple_graphs} holds if restricted to ${\setSizeK}$-hypergraphs. Let 
${\hypergraphG}\changedd{=({\setVertices\changed{}{'}{NULL}},\mu\changed{}{'}{NULL})}{}{405}$ be a 
\changed{$\setSizeK+1$-}{$(\setSizeK+1)$-}{779}hypergraph.
Then for any nonempty ${\setX}\subseteq \changedd{{\setVertices\changed{}{'}{NULL}}}{\Vertices{\hypergraphG}}{405}$ \changed{and}{such that}{1082} $0<\size{\setX}\leq \setSizeK+1$ there is a 
\changed{$(\size{\setX},\weight{setX}({\hypergraphG}))-\simple$}{$(\size{\setX},\weight{setX}({\hypergraphG}))$-\simple}{20} 
\changed{$\setSizeK+1$-}{$(\setSizeK+1)$-}{779}hypergraph $\simpleGraph{setX}{\weight{setX}({\hypergraphG})}$, such that 
$\simpleGraph{setX}{\weight{setX}({\hypergraphG})}\in 
\sums{\Z}{\eqs{{\hypergraphG}}}$. 
\end{lem}

\begin{proof}[Proof of Lemma~\ref{lem:induction step one}]
We show how to construct the hypergraph $\simpleGraph{setX}{\weight{setX}({\hypergraphG})}$. Let $\da\in {\setX}$.
Consider $\cut{{\hypergraphG}}{\da}$ and a set ${\setX}'={\setX}\setminus\{\da\}$. Observe, 
$\weight{setX'}({\cut{{\hypergraphG}}{\da}})=\weight{setX}({\hypergraphG})$.
\changed{So, \changed{due to}{because of}{200}}{Because of}{1087} the assumption we know that there is $\simpleGraph{setX'}{\weight{setX}(\hypergraphG)}$ a 
\changed{$(\size{{\setX}'}, \weight{setX}({\hypergraphG}))-\simple$}
{
$(\size{{\setX}'}, \weight{setX}({\hypergraphG}))$-\simple
}{1089,1093}
${\setSizeK}$-hypergraph, such  that 
$\simpleGraph{setX'}{\weight{setX}({\hypergraphG})}\in \sums{\Z}{\eqs{\cut{{\hypergraphG}}{\da}}}$. 
Now, \changed{due to}{because of}{200} Lemma~\ref{lem:simple graph raising},  for some $\da'\notin \changedd{\setVertices\changed{}{'}{NULL}}{\Vertices{\hypergraphG}}{405}$, the 
\changed{$\setSizeK+1$-}{$(\setSizeK+1)$-}{779}hypergraph 
$\simpleGraph{setX}{\weight{setX}({\hypergraphG})}=
\revcut{\simpleGraph{setX'}{\weight{setX}({\hypergraphG})}}{\da}-
\revcut{\simpleGraph{setX'}{\weight{setX}({\hypergraphG})}}{\da'}$ is \changed{$(\size{\setX}, 
\weight{setX}({\hypergraphG}))-\simple$}
{
$(\size{\setX}, 
\weight{setX}({\hypergraphG}))$-\simple
}{1089,1093}
(point 1) and $\simpleGraph{setX}{\weight{setX}({\hypergraphG})}\in \sums{\Z}{\eqs{{\hypergraphG}}}$ (point 2).
\end{proof}

\begin{lem}\label{lem:induction step two}
\changed{}{Let $\setSizeK\in \N$.}{1081}
Suppose that Theorem~\ref{th:simple_graphs} 
holds if restricted to ${\setSizeK}$-hypergraphs. Let ${\hypergraphG}$ be a \changed{$\setSizeK+1$-}{$(\setSizeK+1)$-}{779}hypergraph and 
$\weight{emptyset}({\hypergraphG})=\vec{a}.$
Then there is a \changed{$(0,\vec{a})-\simple$}{$(0,\vec{a})$-\simple}{20} 
\changed{$\setSizeK+1$-}{$(\setSizeK+1)$-}{779}hypergraph $\simpleGraph{0}{\vec{a}}$ such that $\simpleGraph{0}{\vec{a}}\in \sums{\Z}{\eqs{{\hypergraphG}}}$. 
\end{lem}

\begin{proof}[Proof of Lemma~\ref{lem:induction step two}]
First, observe that \changed{$(0, \vec{a})$-\simple\ \changed{$\setSizeK+1$-}{$(\setSizeK+1)$-}{779}hypergraph}{$\simpleGraph{0}{\vec{a}}$}{1099} has to satisfy 
only two properties\changed{,}{:}{1100} it has $2k+1$ vertices 
and the sum of weights of 
all hyperedges equals $\vec{a}$. 

Thus, the lemma is a consequence of a procedure (defined below) that takes a 
\changed{$\setSizeK+1$-}{$(\setSizeK+1)$-}{779}hypergraph $\hypergraphF$ with $l$ vertices and
if $l>2\setSizeK+1$ then it returns a 
\changed{$\setSizeK+1$-}{$(\setSizeK+1)$-}{779}hypergraph $\hypergraphF'$ such 
that $\weight{emptyset} ({\hypergraphF})=\weight{emptyset}({\hypergraphF'})$, $\size{\changedd{\support{\hypergraphF'}}{\Vertices{\hypergraphF'}}{1169}}<l$, 
and 
${\hypergraphF}'\in \sums{\Z}{\eqs{\hypergraphF}}$.
\changed{
The procedure applied many times starting from ${\hypergraphG}$ produces the required \simple\  hypergraph.
}
{
We start from a hypergraph $\hypergraphG$.
We apply the procedure until we produce a hypergraph with no more than than $2\setSizeK+1$ vertices. It is the required \simple\ hypergraph.
We are guaranteed to finish, as with each application the number of vertices decreases. Moreover, each application of the procedure does not change 
$\weight{emptyset}$, so
 we know that $\weight{emptyset}$ of the produced graph is equal to $\vec{a}$, as required.
}{1106}

{\bf The procedure.} \changed{Let $\hypergraphF=(\setVertices,\mu)$.}{}{NULL} \changed{Suppose $\da\in {\setVertices}$}{We pick any $\da$ in 
$\changed{{\setVertices}}{\changedd{\support{\hypergraphF}}{\Vertices{\hypergraphF}}{1169}}{NULL}$}{1107}. 
Our goal is to construct a \changed{$\setSizeK+1$-}{$(\setSizeK+1)$-}{779}hypergraph $\hypergraphF'$ such that 
\changed{}{$\hypergraphF'\in \sums{\Z}{\eqs{\hypergraphF}}$}{NULL}, \changedd{$\support{\hypergraphF'}\subseteq \support{\hypergraphF}\setminus 
\{\da\}$}{$
\Vertices{\hypergraphF'}\subseteq \Vertices{\hypergraphF}\setminus 
\{\da\}$
}{1169}\changed{}{ 
\ and $\weight{emptyset}(\hypergraphF')=\weight{emptyset}(\hypergraphF)$.}{1110}
As the first step, we construct a 
\changed{$\setSizeK+1$-}{$(\setSizeK+1)$-}{779}hypergraph \changed{$\hypergraphK\in \sums{\Z}{\eqs{\hypergraphF}}$}{$\hypergraphK$}{1131} such that
\begin{enumerate}
 \item \changed{}{$\hypergraphK\in \sums{\Z}{\eqs{\hypergraphF}}$}{1131}
 \item $\weight{emptyset}({\hypergraphK})=\vec{0}$.
 \item For every $\setX$, a set of vertices containing $\da$, 
we have that $\weight{setX}(\hypergraphF)=\weight{setX}({\hypergraphK})$. 
\item \changedd{$\support{{\hypergraphK}} \subseteq \support{\hypergraphF}$}{
$\Vertices{{\hypergraphK}} \subseteq \Vertices{\hypergraphF}$
}{1169}.
\end{enumerate}
Then, $\hypergraphF'=\hypergraphF-{\hypergraphK}$.
\changed{}{The first property guaranties that $\hypergraphF'\in \sums{\Z}{\eqs{\{\hypergraphF\}}}$.}{1113}
\changed{}{Because of the second property $\weight{emptyset}(\hypergraphF')=\weight{emptyset}(\hypergraphF)-\vec{0}=\weight{emptyset}(\hypergraphF)$.}{1113}
 The \changed{second}{third}{1131} property is responsible for deleting all hyperedges containing $\da$ and the \changed{third}{fourth}{1131} property for 
not adding any new vertex.

We define ${\hypergraphK}$ as follows.
Let $\setS$ be \changedd{a simple family for $\cut{\hypergraphF}{\da}$}{an $\{\cut{\hypergraphF}{\da}\}$-simplified family of hypergraphs}{1169}. 
\changed{By,}{By}{1115} $\hat{\setS}$ we denote \changed{a set}{the set}{1115} of all ${\setSizeK}$-hypergraphs \changed{supported by 
$\changed{{\setVertices}}{\support{\hypergraphF}}{NULL}\setminus\{\da\}$ and}{$\hypergraphS$ such that 
\changedd{$\support{\hypergraphS}=\changed{{\setVertices}}{\support{\hypergraphF}}{NULL}\setminus\{\da\}$}
{
$\Vertices{\hypergraphS}=\changed{{\setVertices}}{\Vertices{\hypergraphF}}{NULL}\setminus\{\da\}$
}{1169}
and }{759}
\changed{equivalent to elements of $\setS$}{$\hypergraphS\in \eqs{\setS}$}{1116}. 
\changed{Due to}{Because of}{200} Theorem~\ref{th:expresibility with simple} we know 
that $\cut{\hypergraphF}{\da}\in \sums{\changed{Z}{\Z}{1117}}{\eqs{\setS}}$ \changed{supported}{is supported}{1117} 
by $\changed{{\setVertices}}{\changedd{\support{\hypergraphF}}{\Vertices{\hypergraphF}}{1169}}{NULL}\setminus\{\da\}$\changed{, i.e.}{\ i.e.}{22}
\begin{equation}
\cut{\hypergraphF}{\da}=\sum_i a_i \simpleGraph{any}{}_i \text{ where } a_i \in \Z \text{ and }\simpleGraph{any}{}_i\in \hat\setS.
\end{equation}
\changed{}{Thus,}{1121}
\begin{equation}\label{eq:F description}
 \changed{\text{Thus, }}{}{1121}\revcut{\cut{\hypergraphF}{\da}}{\da}=\sum_i a_i \revcut{\simpleGraph{any}{}_i}{\da}.
\end{equation}
We define ${\hypergraphK}$ as
\begin{equation}\label{eq:Y description}
\begin{split}
{\hypergraphK} \eqdef \revcut{\cut{\hypergraphF}{\da}}{\da}-\sum_i a_i \revcut{\simpleGraph{any}{}_i}{{\da}_i'}=\sum_i a_i 
(\revcut{\simpleGraph{any}{}_i}{\da}- 
\revcut{\simpleGraph{any}{}_i}{{\da}_i'})\\
\changed{\text{ where each }{\da}_i'\in \changed{{\setVertices}}{\support{\hypergraphF}}{NULL}\setminus (\support{\simpleGraph{any}{}_i}\cup\{\da\}).}{}{1128} 
\end{split}
\end{equation}
\changed{}{ where each ${\da}_i'$ is any vertex in the set \changedd{$\changed{{\setVertices}}{
\support{\hypergraphF}}{NULL}\setminus 
(\support{\simpleGraph{any}{}_i}\cup\{\da\})$}
{$
\changed{{\setVertices}}{\Vertices{\hypergraphF}}{NULL}\setminus 
(\Vertices{\simpleGraph{any}{}_i}\cup\{\da\})$}{1169}
.}{1128} 
Observe, that \changedd{
$\size{\support{\simpleGraph{any}{}_i} \cup \{\da\}} \leq 2\setSizeK + 1< 
\size{\changed{{\setVertices}}{\support{\hypergraphF}}{NULL}}$ as 
$\size{\changed{{\setVertices}}{\support{\hypergraphF}}{NULL}}>2\setSizeK+1$}
{
$\size{\Vertices{\simpleGraph{any}{}_i} \cup \{\da\}} \leq 2\setSizeK + 1< 
\size{\changed{{\setVertices}}{\Vertices{\hypergraphF}}{NULL}}$ as 
$\size{\changed{{\setVertices}}{\Vertices{\hypergraphF}}{NULL}}>2\setSizeK+1$
}{1169}, so $\da_i'$ can be always picked.

${\hypergraphK}$ satisfies the required properties:
\begin{enumerate}
 \item ${\hypergraphK}\in \sums{\Z}{\eqs{\hypergraphF}}$. It is sufficient to show that 
$\revcut{\simpleGraph{any}{}_i}{\da}-\revcut{\simpleGraph{any}{}_i}{{\da}_i'}\in 
\sums{\Z}{\eqs{\hypergraphF}}$.
 \changed{Due to}{Because of}{200} 
the assumption that Theorem~\ref{th:simple_graphs}
holds for ${\setSizeK}$-hypergraphs we know that $\simpleGraph{any}{}_i\in \sums{\Z}{\eqs{\cut{\hypergraphF}{\da}}}$. If we combine this 
with Lemma~\ref{lem:simple graph raising} (point 2) then we get that each hypergraph 
\changed{$(\revcut{\simpleGraph{any}{}_i}{\da}-\revcut{\simpleGraph{any}{}_i}{{\da}_i'})\in \sums{\Z}{\eqs{\hypergraphF}}$.}
{
$(\revcut{\simpleGraph{any}{}_i}{\da}-\revcut{\simpleGraph{any}{}_i}{{\da}_i'})$ is in $\sums{\Z}{\eqs{\hypergraphF}}$.
}{1135}

 \item $\weight{emptyset}({\hypergraphK})=\vec{0}$. \changed{Due to}{Because of}{200} Lemma~\ref{lem:simple graph raising} (point 1) we know that 
$\revcut{\simpleGraph{any}{}_i}{\da}-\revcut{\simpleGraph{any}{}_i}{{\da}_i'}$ is
an element of \changedd{a simple family for $\hypergraphF$}{an $\hypergraphF$-simplified family}{1169} and it is 
\changed{$(l,\cdot)-\simple$}{$(l,\cdot)$-\simple}{20} for $l>0$. Thus,
$\weight{emptyset}({\revcut{\simpleGraph{any}{}_i}{\da}-\revcut{\simpleGraph{any}{}_i}{{\da}_i'}})=\vec{0}$ and $\weight{emptyset}({\hypergraphK})=\vec{0}$ as 
${\hypergraphK}$ is a sum 
of 
hypergraphs 
$\revcut{\simpleGraph{any}{}_i}{\da}-\revcut{\simpleGraph{any}{}_i}{{\da}_i'}$.

 \item $\weight{setX}({\hypergraphK})=\weight{setX}(\hypergraphF)$ for any $\setX$ such that $\da\in {\setX}$. Indeed, 
 \[
  \weight{setX}({\hypergraphK})=\sum_i \weight{setX}({\revcut{\simpleGraph{any}{}_i}{\da}})-\weight{setX}({\revcut{\simpleGraph{any}{}_i}{{\da}_i'}}) \ \text{ 
(\changed{}{by\ }{1143,1148}Equation~\ref{eq:Y 
description})}
 \]
 but
 each $\revcut{\simpleGraph{any}{}_i}{{\da}_i'}$ does not contain $\da$ thus 
 \[
 \weight{setX}({\hypergraphK})=\sum_i 
\weight{setX}(\revcut{\simpleGraph{any}{}_i}{\da})\changed{)}{}{1148}=\weight{setX}({\revcut{\cut{\hypergraphF}{\da}}{\da}}) \ \text{ 
 (\changed{}{by\ }{1143,1148}Equation~\ref{eq:F description}).}
 \]
 But, according to Definition~\ref{def:cut vertex} (of the $\cut{\cdot}{\da}$ operation) and Lemma~\ref{lem:isolation_down} we know that
for any set of vertices \changedd{${\setY}\subseteq \changed{{\setVertices}}{\support{\hypergraphF}}{NULL}\setminus\{\da\}$}
{
${\setY}\subseteq \changed{{\setVertices}}{\Vertices{\hypergraphF}}{NULL}\setminus\{\da\}$
}{1169}
: 
\begin{equation}\label{eq:remove v add v}
\weight{Ycup{da}}(\hypergraphF)=\weight{setY}({\cut{\hypergraphF}{\da}})=\weight{Ycup{da}}(\revcut{\cut{\hypergraphF}{\da}}{\da}).
\end{equation}
\changed{Thus }{Thus, }{177}for ${\setX}={\setY}\cup\{\da\}$ we get $\weight{setX}({\hypergraphK})=\weight{setX}(\hypergraphF)$.
 \item \changedd{
 $\support{{\hypergraphK}}\subseteq \changed{{\setVertices}}{\support{\hypergraphF}}{NULL}$.}{
 $\Vertices{{\hypergraphK}}\subseteq \changed{{\setVertices}}{\Vertices{\hypergraphF}}{NULL}$. 
 }{1169}Indeed, for every $\simpleGraph{any}{}_i$ 

 \changed{hold}{it 
holds that}{451} 
\changedd{$\support{\revcut{\simpleGraph{any}{}_i}{\da}}\subseteq \changed{{\setVertices}}{\support{\hypergraphF}}{NULL}$ and
$\support{\revcut{\simpleGraph{any}{}_i}{{\da}_i}}\subseteq \changed{{\setVertices}}{\support{\hypergraphF}}{NULL}$.}
{
$\Vertices{\revcut{\simpleGraph{any}{}_i}{\da}}\subseteq \changed{{\setVertices}}{\Vertices{\hypergraphF}}{NULL}$ and
$\Vertices{\revcut{\simpleGraph{any}{}_i}{{\da}_i}}\subseteq \changed{{\setVertices}}{\Vertices{\hypergraphF}}{NULL}$.
}{1169} \qedhere % align qed box with the last line
\end{enumerate}
\end{proof}

\begin{proof}[Proof of Lemma~\ref{lem:simple_graphs}]
The proof \changed{goes via}{is by}{1158} induction on ${\setSizeK}$. The base for the induction is given by Lemma~\ref{lem:induction base}.
The induction step is given by Lemmas~\ref{lem:induction step one}~and~\ref{lem:induction step two}.
\end{proof}

\section{The proof of Theorem~\ref{thm:core} itself.}\label{sec:core_proof}
\changed{Having proven theorems}{Theorems}{NULL}~\ref{th:simple_graphs}~and~\ref{th:expresibility with simple}\changed{}{\ are proven, so now,}{NULL} the proof 
of Theorem~\ref{thm:core} is easy.
Let us recall its statement. 

{\bf Theorem~\ref{thm:core}}
\emph{The following conditions are equivalent,
 for a finite set ${\setHypergraphsH}$ of hypergraphs and a hypergraph $\hypergraphH$, all of the same arity and dimension:
 \begin{enumerate}
 \item $\hypergraphH$ is a $\Z$-sum of $\eqs{{\setHypergraphsH}}$;
 \item $\hypergraphH$ is locally a $\Z$-sum of ${\setHypergraphsH}$.
 \end{enumerate}
}

\begin{proof}
\changed{From the left to the right.}{From Point $1$ to Point $2$.}{1163} 
\changed{Observe that $\weight{setX}$ is a homomorphism, for any set $\setX$. 
\changed{Thus }{Thus, }{177}if}{}{1165} ${\hypergraphH}\in \sums{\Z}{\eqs{{\setHypergraphsH}}}$ i.e.
${\hypergraphH}=\sum_i a_i {\hypergraphG}_i$ where $a_i\in \Z$ and ${\hypergraphG}_i\in\eqs{{\setHypergraphsH}},$
\changed{then}{thus}{1165} $\weight{setX}({\hypergraphH})=\sum_i a_i \weight{setX}({\hypergraphG}_i)$, for any set ${\setX}\subseteq 
\changedd{{\setVertices}}{\Vertices{\hypergraphH}}{405}$.
The above holds for any set ${\setX}\subseteq \changedd{\setVertices}{\Vertices{\hypergraphH}}{405}$ thus ${\hypergraphH}$ is locally a \changed{$\Z$ 
sum}{$\Z$-sum}{1166} of ${\setHypergraphsH}$.

\changed{From the right to the left.}{From Point $2$ to Point $1$.}{1163} If ${\hypergraphH}$ is locally a $\Z$-sum of ${\setHypergraphsH}$ then for any set 
${\setX}\subseteq \changedd{{\setVertices}}{\Vertices{\hypergraphH}}{405}$ there is a 
${\setSizeK}$-hypergraph ${\hypergraphG}_{\setX}$ such 
that
$\weight{setX}({\hypergraphH})=\weight{setX}({{\hypergraphG}_{\setX}})$ and ${\hypergraphG}_{\setX}\in \sums{\Z}{\eqs{{\setHypergraphsH}}}.$ 
\changed{Due to}{Because of}{200}
Theorem~\ref{th:simple_graphs} we conclude that there is \changedd{a family simple for \changed{${\hypergraphH}$}
{${\{\hypergraphH\}}$}{NULL}}{an $\hypergraphH$-simplified family}{1169} such that its elements are in 
$\sums{\Z}{\eqs{{\setHypergraphsH}}}$. 
Now\changed{}{,}{NULL} we can apply Theorem~\ref{th:expresibility with simple} together with~Lemma~\ref{lem:sumofupto} and conclude that ${\hypergraphH}\in 
\sums{\Z}{\eqs{{\setHypergraphsH}}}$. 
\end{proof}

\section{Proof of Theorem~\ref{thm:problem1}.}\label{sec:red}

\changed{}{Before we prove the theorem, let us recall its statement.
\\{\bf Theorem~\ref{thm:problem1}}
\emph{For every fixed arity ${\setSizeK}\in \N$, the $\N$-solvability problem is in \NETIME.}\\
}{1172}

We prove Theorem~\ref{thm:problem1} by showing that in \NETIME\ it is possible to reduce $\N$-solvability to $\Z$-solvability.
The produced instance of $\Z$-solvability is of exponential size, and can be solved in
\ETIME\ \changed{due to}{because of}{200}  Theorem~\ref{thm:problem2}. 

Before we start, we need to recall some facts about solution of systems of linear equations.

\section*{Hybrid linear sets.}

\begin{defi}\label{def:semilinera_set}
A set of vectors is called \emph{hybrid linear} if it is the smallest set that 
includes a finite set $\setB$, called a base, and that is closed under the addition of elements
from a finite set $\setP$, called periods.
\end{defi}
\begin{thmC}[\cite{Pottier}] % omit the parentheses
Let $M$ be a \changed{$\dimension\times m$ -matrix}{$\dimension\times m$-matrix}{1182} with integer entries and $\vec{y}\in \ktuple{\dimension}{\Z}$.
The set of nonnegative integer solutions of linear equations 
\[
M \cdot \vec{x}=\vec{y}
\]
is a hybrid linear set. 
The \emph{base} $\setB$ and \emph{periods} \emph{$\setP$} are as follows:
\begin{itemize}
 \item {Base:} it is the set of minimal, in the pointwise sense, solutions \changed{}{of}{1186}
 $M\cdot \vec{x}=\vec{y}.$
 \item {Periods:} is the set of minimal nontrivial solutions of 
\begin{equation}\label{eq:Pottier1}
	M\cdot \vec{x}=\vec{0}.
\end{equation}
\end{itemize}
\end{thmC}

In the paper~\cite{Pottier}\changed{}{,}{1190} Pottier provides bounds on the norms of
$\setB$ and $\setP$. \changed{}{We present these bounds next.}{1190}

For a vector $\vec{v}\in \ktuple{m}{\Z} $ we introduce two norms: the infinity norm $\inorm{\vec{v}}$ and 
the norm one $\norm{\vec{v}}{1}$, defined as follows:
\begin{itemize}
 \item \changed{$\inorm{\vec{v}}\eqdef max (\setof{\abs{\vec{v}[i]}}{ \text{ for }1\leq i \leq m})$,}{
 $\inorm{\vec{v}}\eqdef max \setof{\abs{\vec{v}[i]}}{ \text{ for }1\leq i \leq m}$.}{1193}
 \item $\norm{\vec{v}}{1}\eqdef \sum_{i=1}^m \abs{\vec{v}[i]}$.
\end{itemize}
\changed{}{Let $\setU$ be a finite family of data vectors. We extend definitions of norms to families of data vectors.}{1196} 
Let \changed{$\norm{{\changed{\setHypergraphsH}{\setU}{1196}}}{\infty}\eqdef max(\setof{\inorm{\changed{}{\vec{v}}}}{ \changed{\vec{h}}{\vec{v}}{1196}\in 
{\changed{\setHypergraphsH}{\setU}{1196}}})$}
{$\norm{{\changed{\setHypergraphsH}{\setU}{1196}}}{\infty}\eqdef max\setof{\inorm{\changed{\vec{h}}{\vec{v}}{1196}}}{ \changed{\vec{h}}{\vec{v}}{1196}\in 
{\changed{\setHypergraphsH}{\setU}{1196}}}$}{1193}
and 
\changed{$\norm{{\changed{\setHypergraphsH}{\setU}{1196}}}{1,\infty}\eqdef max(\setof{\norm{\changed{\vec{h}}{\vec{v}}{1196}}{1}}{ 
\changed{\vec{h}}{\vec{v}}{1196}\in {\changed{\setHypergraphsH}{\setU}{1196}}})$}
{
$\norm{{\changed{\setHypergraphsH}{\setU}{1196}}}{1,\infty}\eqdef max\setof{\norm{\changed{\vec{h}}{\vec{v}}{1196}}{1}}{ \changed{\vec{h}}{\vec{v}}{1196}\in 
{\changed{\setHypergraphsH}{\setU}{1196}}}$
}{1193}.
Also, for a $d\times m$-matrix $M$ we introduce 
$\norm{M}{1,\infty}\eqdef \norm{\setM}{1,\infty}$ where $\setM$ is the set of columns of the matrix $M$.

\begin{lemC}[\cite{Pottier}]\label{lem:Pottier} % omit the parentheses
Let $ M \cdot \vec{x} = \vec{y}$ be a system of linear equations such that $M$ is a $d\times m$-matrix.
Then the set of solutions in $\ktuple{m}{\N}$ is the hybrid linear set \changed{and is }{}{1200}described by 
the base $\setB$ and the set of periods $\setP$ such that:
\begin{itemize}
    \item $\setB,\setP\subset \ktuple{m}{\N}$,
\item $\size{\setP}\leq d\cdot m$
\item $\inorm{\setB}, \inorm{\setP}\leq (\norm{M}{1,\infty}+\inorm{\vec{y}}+2)^{d+m}.$
\end{itemize}
\end{lemC}

\begin{defi}\label{def:reversibility_vector}
\changed{}{Let $\setU$ be a family of vectors.}{1205} A vector $\vec{x}\in \setU$ is \emph{reversible} in a family of vectors $\setU$ if $-\vec{x}\in 
\sums{\N}{\setU}$.
\changed{Vectors that are not reversible we call}{We call vectors that are not reversible}{237} \emph{\changed{non-reversible}{nonreversible}{107}}.
\end{defi}

\changed{Thus }{Thus, }{177}from Equation~\ref{eq:Pottier1} and Lemma~\ref{lem:Pottier} we conclude\changed{.}{:}{1207}
\begin{lem}\label{lem:bound-non-rev}
Let $\setU_1$ and $\setU_2$ be two finite 
sets of vectors in $\ktuple{\dimension}{\Z}$ 
such that every vector in $\setU_2$ is \changed{non-reversible}{nonreversible}{107} in $\setU= \setU_1\cup\setU_2$.
Suppose $\vec{y}\in \sums{\N}{\setU}$. \changed{}{For any solution:}{1218}
\begin{equation}\label{eq:reversability}
\vec{y}=\sum_{\vec{v}\in\setU_1} 
a_{\vec{v}} 
\cdot \vec{v}+ \sum_{\vec{w}\in\setU_2} b_{\vec{w}} \cdot \vec{w}
\end{equation}
where $a_{\vec{v}}, b_{\vec{w}}\in \N$\changed{. 
Then}{, it holds that}{NULL}
\[
\sum_{\vec{w}\in\setU_2} b_{\vec{w}}\leq 
\size{\setU_2}\cdot \
(\norm{\setU}{1,\infty}+\inorm{\vec{y}}+2)^{d+\size{
\setU}}
\]
i.e. \changed{}{$\sum_{\vec{w}\in\setU_2} b_{\vec{w}}$}{1217} is bounded exponentially.
\end{lem}
\begin{proof}[Proof of Lemma~\ref{lem:bound-non-rev}.]
The set of solutions of Equation~\ref{eq:reversability} is hybrid linear, and given by some $\setB,\setP \subset \ktuple{\size{\setU_1\cup 
{\setU}_2}}{\N}$.
Every period is a solution of the equation
\begin{equation}\label{eq:lemma5}
	\vec{0}=\sum_{\vec{v}\in\setU_1} 
a'_{\vec{v}} 
\cdot \vec{v}+ \sum_{\vec{w}\in\setU_2} b'_{\vec{w}} \cdot \vec{w}.   
\end{equation}
From the definition of reversibility we get that in any solution of Equation~\ref{eq:lemma5} all 
$b'_{\vec{w}}$ are equal to $0$. 
\changed{Thus }{Thus, }{177}in every solution of
\changed{the Equation~\ref{eq:reversability}}{Equation~\ref{eq:reversability},}{1124} 
the sum $\sum_{\vec{w}\in\setU_2} b_{\vec{w}}$ 
is bounded by $\size{\setU_2}\cdot \inorm{\setB}\leq \size{\setU_2}\cdot \
(\norm{\setU}{1,\infty}+\inorm{\vec{y}}+2)^{d+\size{
\setU}}$, where the last inequality is given by Lemma~\ref{lem:Pottier}.
\end{proof}

\section*{\changed{Proof of Theorem~\ref{thm:problem1}.}
{The main part of the proof of Theorem~\ref{thm:problem1}.}{NULL}}

\begin{figure*}[t]
\begin{tikzpicture}\label{fig:diagram_of_the_proof}
[place/.style={circle,draw=\mycolorOne!50,fill=\mycolorOne!20,thick,inner sep=0pt,minimum size=4mm},
stplace/.style={rectangle,draw=\mycolorDwa!20,fill=\mycolorDwa!20,thick,inner sep=0pt,minimum size=4mm}
]
\node[]		(P1) 	at (0,0)	{$\vec{v}\in \sums {\N}{\perm{\setInput}}$?};
\node[] 	(P2)	at (8.5,0)	{$h(\vec{v})\in \sums {\N}{h(\perm{\setInput}})$?};

\node[] 	(E1)	at (0,-3)	{\begin{minipage}{5cm}
        	    	         	  $\vec{v}=\sum_{\vec{v'}} a_{\vec{v'}} \vec{v'} + \sum_{\vec{w}} b_{\vec{w}} \vec{w}  $ where\\ 
        	    	         	  $a_{\vec{v'}}\in \N, b_{\vec{w}} \in \N$\\
        	    	         	  $\vec{v'}$ are data vectors reversible in $\perm{\setInput}$\\
        	    	         	  $\vec{w}$ are data vectors \changed{non-reversible}{nonreversible}{107} in $\perm{\setInput}$
        	    	         	 \end{minipage}
};
\node[] 	(E2)	at (8.5,-3) 	{\begin{minipage}{6cm}
        	    	         	  $h(\vec{v})=\sum_{\vec{v'}} a_{\vec{v'}} h(\vec{v'}) + \sum_{\vec{w}} b_{\vec{w}} h(\vec{w}) $\\ where 
        	    	         	  $a_{\vec{v'}}\in \N, b_{\vec{w}} \in \N$\\
        	    	         	  $h(\vec{v'})$ are data vectors reversible in $h(\perm{\setInput})$\\
        	    	         	  $h(\vec{w})$ are data vectors \changed{non-reversible}{nonreversible}{107} in $h(\perm{\setInput})$
        	    	         	 \end{minipage}
        	    	         	 };

\node[] 	(B1)at (0,-6)	{\begin{minipage}{3cm}
				    $\sum_{\vec{w}} b_{\vec{w}}$ is bounded \\
				    exponentially
				    \end{minipage}};
\node[]		(B2)	at (8.5,-6) 	{Lemma~\ref{lem:Pottier}: exponential bound on $\sum_{\vec{w}} b_{\vec{w}}$};

\draw[->] (P1) to node [above] {$h$} (P2) ;
\draw[-] (P1) to node [right] {a solution} (E1); 
\draw[->] (E1) to node [above, near end] {\begin{minipage}{3cm}$h$ preserves\\ reversibility \end{minipage} } (E2) ;

\draw [-] (P2) to node [ right]{a solution} (E2);
\draw [-] (E2) to node [right] {the bound} (B2);
\draw[->] (B2) to node [above] {transfer back} (B1) ;
\end{tikzpicture}  
  \caption{The diagram of the proof of Theorem~\ref{thm:problem1}.}\label{fig:5}
\end{figure*}

The idea of this reduction is as follows. We use $\setInput$ for a finite family of data vectors in
$\kset{\setSizeK}{\setD}\xrightarrow{}\ktuple{\dimension}{\Z}$ and $\vec{v}$ for a target vector.
Similarly to Definition~\ref{def:reversibility_vector} 
reversibility can be defined for data vectors. Namely, a data vector $\vec{y}$ is reversible in 
a set of data vectors $\setInput$ if $-\vec{y}\in \sums{\N}{\perm{\setInput}}$.

\changed{The schema of the proof may be tracked}{The schema of the proof may be followed}{1248} on the diagram in 
\changed{Figure\ref{fig:5}}{Figure~\ref{fig:5}}{1248}.
Next, we define a homomorphism from data vectors $\kset{\setSizeK}{\setD} \rightarrow \ktuple{\dimension}{\Z}$ to vectors in $\ktuple{\dimension}{\Z}$. 
The homomorphism is defined in such a way that it has \changed{ an additional property, namely,}{a property that}{1250} for every data vector $\vec{y}\in 
\setInput$ \changed{its image 
is reversible in the image of the set $\setInput$}{it holds that $h(\vec{y})$ is reversible in $h(\setInput)$}{1251}
if\changed{}{,}{325} and only if\changed{}{,}{325} $\vec{y}$ is reversible in $\setInput$. 
Now, the given instance \changed{}{of\ }{1252}$\N$-solvability problem for data vectors 
$\kset{\setSizeK}{\setD}\xrightarrow{} \ktuple{\dimension}{\Z}$ 
(called the first problem) is transformed, via the homomorphism, 
to the $\N$-solvability for vectors in $\ktuple{\dimension}{\Z}$\changed{, i.e.}{\ i.e.}{22} system of linear equations (called the second problem).
The homomorphic image of any solution to the first problem is also a solution \changed{of}{to}{1255} the second problem.
 \changed{Due to}{Because of}{200} Lemma~\ref{lem:Pottier}, 
 there is a bound on the \changed{usage}{number of appearances}{1256,1258} of \changed{non-reversible}{nonreversible}{107} vectors in any solution 
\changed{of}{to}{1255} the second problem. 
 This bound can be then transferred back through the homomorphism. This provides us with the bound 
 on the \changed{usage}{number of appearances}{1256,1258} of \changed{non-reversible}{nonreversible}{107} data vectors in any solution \changed{of}{to}{1255} 
the first problem.
With the bound, we can guess the \changed{non-reversible}{nonreversible}{107} part of the solution \changed{for}{to}{1255} the first problem. 
What remains is to verify that our guess is correct.
To do this we have to show that the target minus the guessed \changed{non-reversible}{nonreversible}{107} part 
can be expressed using the reversible data vectors. 
But this is exactly $\Z$-solvability.
Indeed, if data vectors can be reversed then we can subtract them freely.

\begin{defi}
Let $\vec{y}\colon \kset{\setSizeK}{\setD}\xrightarrow{} \ktuple{\dimension}{\Z}$ be a data vector. The \emph{\changed{ }{}{1264}data projection} of $\vec{y}$ 
is defined as 
$\Proj(\vec{y}) 
	\eqdef \sum_{x\in \kset{\setSizeK}{\setD}} \vec{y}(x)$. It is \changed{well defined}{well-defined}{1266} as $\vec{y}$ is almost everywhere equal to 
$\vec{0}$.
For a set of data vectors $\setInput$ we define 
\changed{$\Proj(\setInput)\eqdef\bigcup_{\vec{y}\in \setInput}\{\Proj(\vec{y})\}$}
{
    $\Proj(\setInput)\eqdef\setof{\Proj(\vec{y})}{\vec{y}\in \setInput}$
}
{1268}
to be the \emph{data projection} of $\setInput$.\changed{\ }{}{1268}\footnote{\changed{Data projection $\Proj$ is the same thing as 
$\weight{emptyset}({\placeHolder})$ but in the \changed{hypergraphs domain}{hypergraph domain}{1271}.}{
Data projection $\Proj$ is equal to a transformation that
takes a data vector and transforms it into a hypergraph (the transformation described in Section~\ref{sec:tohyp}) composed with
$\weight{emptyset}({\placeHolder})$.
}{1271}} 
\end{defi}
\begin{prop}
A data projection is a homomorphism from the group of data vectors with addition to $\ktuple{\dimension}{\Z}$.\qed
\end{prop}
\begin{defi}
	For a data vector \changed{$\vec{y}$}{$\vec{y}\in \kset{\setSizeK}{\setD}\to{}{\ktuple{\dimension}{\Z}}$}{NULL} its support denoted 
$\support{\vec{y}}\subset \setD$ is the set $\setof{\da}{\changed{\exists_{x\in \kset{\setSizeK}{\setD}}}{\exists {x\in \kset{\setSizeK}{\setD}}}{1274} \text{ 
such 
that }\da \in x \text{ and }\vec{y}(x)\neq 
\vec{0}}$. We say that the vector $\vec{y}$ is supported by a set $\setSupport\subset\setD$ if $\support{\vec{y}}\subseteq \setSupport$. 
This notion of the support comes from the concept that any data permutation which is identity on the support of a data vector does not modify the data vector 
itself.
The definition may 
be lifted to \changed{the set}{sets $\setInput$}{1275} of data vectors $\support{\setInput}\eqdef \bigcup_{\vec{y}\in\setInput} \support{\vec{y}}.$
\end{defi}
\changed{}{In the following definitions and lemmas we use notation introduced in the definition of the $\Numbers$-solvability problem (Section~\ref{sec:eq}), so $\vec{v}$ is a single data vector and 
$\setInput$ is a finite set of data vectors.}{1278}
\begin{defi}\label{def:smooth}
Let $\setSupport\eqdef \support{\setInput}\cup \support{\vec{v}}$ \changed{(as in the formulation of the $\N$-solvability problem)}{}{1278}. 
By $\Pi$ we denote \changed{a}{the}{1279} set of all data permutations \changed{that are identity outside of $\setSupport$ (all permutations of $\setSupport$)}{$\pi$ such that for any $\da\not\in\setSupport$ we have $\pi(\da)=\da$}{1279}. 
For a data vector $\vec{y}\in \setInput\cup\{\vec{v}\}$ we define the 
\emph{smoothing} operator, $\smooth{\vec{y}}\eqdef \sum_{\pi\in \Pi} \vec{y}\circ \pi$.
\end{defi}
\begin{exa}
\changed{}{Suppose $\vec{v}$ is a data vector $\setD\to{}\Z^{2}$ as follows: $\vec{\da}=[1,2], \vec{\db}=[1,3],$ and $\vec{\dc}=[0,0]$ for any $\dc\in \setD 
\setminus\{\da,\db\}$. Let $\setInput=\{\vec{v}\}$. Then $\setSupport=\{\da,\db\}$ and $\smooth{\vec{v}}(\da)=\smooth{\vec{v}}(\db)=[2,5]$
and $\smooth{\vec{v}}(\dc)=[0,0]$ for any $\dc\in \setD\setminus\{\da,\db\}$.
}{1282}
\end{exa}
\begin{lem}\label{lem:smoothing}
Let $\setSizeK, \setSupport$ be as defined above.
There is \changed{the}{a}{1283} constant $c$ depending on ${\setSizeK}$ and $\size{\setSupport}$, such that for any 
$x\in \kset{\setSizeK}{\setSupport}$ and any data vector $\vec{y}\in \setInput\cup\{\vec{v}\}$ the 
equality
$\smooth{\vec{y}}(x)=\Proj(\vec{y})\cdot c$ holds.
\end{lem}
\begin{proof}
\changed{}{It holds}{NULL} \changed{Due to}{because of}{200} symmetry\changed{}{\ of the formula defining the smooth operator}{1286}.
\changed{We do not need the explicit formula}{We do not need to calculate the value of $c$}{1286}.
\end{proof}

\begin{lem}\label{lem:reversibility-characterisation}
 $\vec{y}\in \setInput$ is reversible in 
$\perm{\setInput}$ if\changed{}{,}{325} and only if\changed{}{,}{325} $\Proj(\vec{y})$ is reversible in $\Proj(\setInput)$.
\end{lem}

\begin{proof}
$\implies$
Trivial as $\Proj$ is a homomorphism.

\noindent $\impliedby$
Let $\setSupport$ and $\Pi$ be as in Definition~\ref{def:smooth}.
We list elements of $\setInput=\{\vec{v_1}\ldots \vec{v_m}\}$.
By $M$ we denote the $1\times m$-matrix with entries in $\ktuple{\dimension}{\Z}$
such that \changed{$X[1][j]=\Proj(\vec{w})$}{$M[1][j]=\Proj(\vec{v_j})$}{1290}.

From the assumptions we know that there is a 
vector $\vec{x}\in \ktuple{m}{\N}$ such that
\begin{align}\label{eq:pattern1}
M\cdot \vec{x}= -\Proj(\vec{y}).
\end{align}
 Let $c$ be the constant from Lemma~\ref{lem:smoothing}.
 We multiply both sides of Equation~\ref{eq:pattern1} by $c$ getting 
\begin{align*}%\label{eq:pattern2}
c\cdot M\cdot \vec{x}= -c\cdot \Proj(\vec{y}).
\end{align*}
But from this and Lemma~\ref{lem:smoothing} we conclude that for any $z\in \kset{\setSizeK}{\setSupport}$
\begin{align*}%\label{eq:pattern3}
\left( \sum_{\vec{v_i}\in{\setInput}}(\smooth{\vec{v_i}}\cdot \vec{x}[i]) \right)(z) = -\smooth{\vec{y}}(z)
\end{align*}
\changed{and what follows}{and it follows that}{1301} 
\begin{align*}\label{eq:pattern3}
\changed{\left(\sum_{\vec{v_i}\in{\setInput}} (\smooth{\vec{v_i}}\cdot \vec{x}[i]) \right) = -\smooth{\vec{y}}.}
{\sum_{\vec{v_i}\in{\setInput}} (\smooth{\vec{v_i}}\cdot \vec{x}[i])  = -\smooth{\vec{y}}.
}{1303}
\end{align*}
Using the definition of smoothing we rewrite further
\begin{align*}%\label{eq:pattern4}
\sum_{\vec{v_i}\in\setInput}(\smooth{\vec{v_i}}\cdot \vec{x}[i]) +
\sum_{\pi\in \Pi\setminus\{identity\}} \vec{y}\circ \pi
= -\vec{y}.
\end{align*}
As $\vec{y}\in \setInput$ we see that we expressed $-\vec{y}$ as a sum of elements in 
$\perm{\setInput}.$
\end{proof}

\begin{cor}\label{cor:fast test}
\changed{Non-reversible in $\perm{\setInput}$ elements of $\setInput$ can be identified in $\PTIME$ (independently from $\dimension$).}
{
For a given vector $\vec{y}$ in $\setInput$, in \PTIME\ we can answer if $\vec{y}$ is nonreversible in $\perm{\setInput}$. The complexity depends
polynomially from $\setSizeK$ and $\dimension$. 
}{1310, 1311}
\end{cor}
\begin{proof}[Proof of Corollary~\ref{cor:fast test}.]
Let $\vec{y}\in \setInput$\changed{.}{. The data vector}{1312} $\vec{y}$ is reversible in $\perm{\setInput}$ if\changed{}{,}{325} and only if\changed{}{,}{325} 
there is $l\in \N$ such that $-l\cdot \Proj({\vec{y}})\in 
\sums{\N}{\Proj({\setInput})}$. Indeed, it is sufficient to add $(l-1) \Proj({\vec{y}})$. \changed{Thus }{Thus, }{177}the question about reversibility is 
equivalent to the question 
\changed{if}{of whether}{1314} $\Proj({\vec{y}})\in \sums{\Q_+}{\Proj({\setInput})}$, where $\Q_+$ stands for nonnegative rationals. The last question is known as the linear programming 
problem and is known to be solvable in \PTIME\ \cite{khachiyan1979polynomial,DBLP:conf/stoc/CohenLS19}.
\end{proof}

\begin{lem}\label{lem:boundeddatasolutions}
Let $\setInput_1$ and $\setInput_2$ form the partition of the set $\setInput$ such that $\setInput_1$ is the set of data vectors 
reversible in $\perm{\setInput}$ and $\setInput_2$ is the set of data vectors \changed{non-reversible}{nonreversible}{1319} in $\perm{\setInput}$. 
Suppose $\vec{v}$ can be expressed as a sum
\changed{
\begin{equation}\label{eq:lembound}
\vec{v}=\left(\sum_{\vec{v_i}\in \perm{\setInput_1}} a_i \vec{v_i}
\right)+\left(\sum_{\vec{w}\in \perm{\setInput_2}}b_{\vec{w}} \vec{w}
\right) \text{ where } a_i,b_{\vec{w}}\in \N.
\end{equation}
}
{
\begin{equation}\label{eq:lembound}
\vec{v}=\left(\sum_{\vec{y}\in \perm{\setInput_1}} a_{\vec{y}} \vec{y}
\right)+\left(\sum_{\vec{w}\in \perm{\setInput_2}}b_{\vec{w}}\vec{w}
\right) \text{ where } a_{\vec{y}},b_{\vec{w}}\in \N.
\end{equation}
}{1322}
Then the sum $\sum_{\vec{w}\in \perm{\setInput_2}}b_{\vec{w}}$ is bounded exponentially,
precisely
$\sum_{\vec{w}\in\perm{\setInput_2}} b_{\vec{w}}\leq 
\size{\Proj(\setInput_2)}\cdot \
(\norm{\Proj(\setInput)}{1,\infty}+\inorm{\Proj(\vec{v})}+2)^{\dimension+\size{
\setInput}}.$

\end{lem}
\begin{proof}
The proof is based on the bound for a solution in $\N$ of a system of linear equations (Lemma~\ref{lem:bound-non-rev}). By $\Auth$ we denote the set of data 
permutations (bijections $\setD \xrightarrow{} \setD$).
\changed{Indeed we}{We}{1328} can apply the homomorphism $\Proj$ to both sides of Equation~\ref{eq:lembound} and get
\begin{align*}
\Proj(\vec{v})=\Proj\left(\sum_{\vec{y}\in \perm{\setInput_1}} a_{\vec{y}}\vec{y}
\right)+\Proj\left(\sum_{\vec{w}\in \perm{\setInput_2}}b_{\vec{w}}\vec{w}
\right) 
\end{align*}
and further
\begin{align}\label{eq:projection}
\Proj(\vec{v})=\left(\sum_{\vec{y}\in \perm{\setInput_1}}a_{\vec{y}} \Proj(\vec{y})
\right)+\left(\sum_{\vec{w}\in \perm{\setInput_2}}b_{\vec{w}} \Proj(\vec{w})
\right)
\end{align}
\begin{align*}
\Proj(\vec{v})=\left(\sum_{\vec{y}\in \setInput_1} \Proj(\vec{y}) \sum_{\pi\in \Auth} a_{\vec{y},\pi} 
\right)+\left(\sum_{\vec{w}\in \setInput_2} \Proj(\vec{w}) \sum_{\pi\in \Auth} b_{\vec{w},\pi} 
\right)
\end{align*}
\changed{}{Where $\sum_{\pi\in \Auth} a_{\vec{y},\pi}$ and $\sum_{\pi\in \Auth} b_{\vec{w},\pi}$ are obtained by grouping
elements of the sums in Equations~\ref{eq:projection} by $\Proj(\vec{y})$ and $\Proj(\vec{w})$, respectively.}{1337}

But\changed{ due to}{, because of}{200} Lemma~\ref{lem:reversibility-characterisation} we know that $\Proj$ preserves reversibility so the bound on the sum 
$\sum_{\vec{w}\in \perm{\setInput_2}}b_{\vec{w}}= \sum_{\vec{w}\in \setInput_2}\sum_{\pi\in \Auth} b_{\vec{w},\pi}$ can be taken from Lemma~\ref{lem:bound-non-rev}, 
namely
$\sum_{\vec{w}\in \perm{\setInput_2}} b_{\vec{w}}\leq 
\size{\Proj(\setInput_2)}\cdot \
(\norm{\Proj(\setInput)}{1,\infty}+\inorm{\Proj(\vec{v})}+2)^{\dimension+\size{\setInput}}.$
\end{proof}

\newcommand{\setSbar}{\widetilde{\setSupport}}
\newcommand{\setShat}{\widehat{\setSupport}}

\begin{thm}\label{thm:prob2aux}
Let $\setInput_1, \setInput_2$ be a partition of $\setInput$ into elements reversible and \changed{non-reversible}{nonreversible}{1319} in $\perm{\setInput}$, 
respectively.   
Suppose, $s_{max}\eqdef max(\setof{\size{\support{\vec{w}}}}{\vec{w}\in \setInput_2}).$ 
Let $\setSbar\subset \setD$ \changed{is}{be}{NULL} a set such that
\changed{$\support{\vec{v}}\subseteq \setSbar$}{$\vec{v}$ is supported by $\setSbar$}{1345} and 
\begin{equation}
\size{\setSbar}=\size{\support{\vec{v}}} + \\
s_{max}\cdot \bigg(\size{\Proj(\setInput_2)}\cdot \
(\norm{\Proj(\setInput)}{1,\infty}+\inorm{\Proj(\vec{v})}+2)^{\dimension+\size{
\setInput}}\bigg).
\end{equation}
Finally, let $\widetilde{\setInput_2} \subset \perm{\setInput_2}$ \changed{is}{be}{1348} the set of all data vectors in 
\changed{$\perm{\setInput}$}{$\perm{\setInput_2}$}{1348} that are supported by $\setSbar$.

Then $\vec{v}\in \sums{\N}{\perm{\setInput}}$ if\changed{}{,}{325} and only if\changed{}{,}{325}
\[
\vec{v} =\sum_{\vec{y}\in \perm{\setInput_1}}  a_{\vec{y}} \vec{y} +\sum_{\vec{w}\in \widetilde{\setInput_2}} b_{\vec{w}} \vec{w} \text{ where } a_{\vec{y}}, b_{\vec{w}}\in \N. 
\]
Note that the second sum is only over vectors supported by $\setSbar.$
\end{thm}
\begin{proof}
Implication from right to left is trivial. Implication from left to right is \changed{the}{a}{1355} consequence of 
Lemma~\ref{lem:boundeddatasolutions}. 
Suppose 
$\vec{v} =\sum_{\vec{y}\in \perm{\setInput_1}}  a_{\vec{y}}' \vec{y} +\sum_{\vec{w}\in \perm{\setInput_2}} b_{\vec{w}}' \vec{w}.$
Let $\setShat$ be the union of supports of $\vec{w}\in \perm{\setInput_2}$ that appear with nonzero coefficients in the sum above.
Observe $\size{\setShat} +\size{\support{\vec{v}}}\leq \size{\setSbar}$ due to Lemma~\ref{lem:boundeddatasolutions} and the definition of $\size{\setSbar}$. 
Let $\pi$ be a data permutation that
is the identity on $\support{\vec{v}}$ and injects $\setShat$ into $\setSbar$.
Then, $\vec{v}=\vec{v}\circ \pi^{-1}=\sum_{\vec{y}\in \perm{\setInput_1}}  a_{\vec{y}}' \vec{y}\circ\pi^{-1} +\sum_{\vec{w}\in \perm{\setInput_2}} b_{\vec{w}}' 
\vec{w}\circ 
\pi^{-1},$
as required\changed{}{\ (since $\vec{w}\circ 
\pi^{-1}\in \widetilde{\setInput_2}$)}{1362}.
\end{proof}

\begin{proof}[Proof of Theorem~\ref{thm:problem1}]
\changed{Due to}{Because of}{200} Theorem~\ref{thm:prob2aux}, in \NETIME\ it is possible to construct an \changed{instance of the
$\Z$-solvability with data problem of exponential size}{exponential size instance of $\Z$-solvability with data problem}{1365}, that has a solution 
if\changed{}{,}{325} and only if\changed{}{,}{325} 
the original problem has a solution. Precisely, \changed{using the notation from Theorem~\ref{thm:prob2aux} we ask
if 
$\vec{v}-\sum_{\vec{w}\in \setInput_2\circ \Pi} b_{\vec{w}} \vec{w} \in \sums{\Z}{\perm{\setInput_1}}$}
{we ask
if
$\vec{v}-\sum_{\vec{w}\in \widetilde{\setInput_2}} b_{\vec{w}} \vec{w} \in \sums{\Z}{\perm{\setInput_1}}$, where $\widetilde{\setInput_2}$
is defined in Theorem~\ref{thm:prob2aux}
}{1367}. According to Theorem~\ref{thm:problem2}
it can be solved in \ETIME. \changed{Thus }{Thus, }{177}the algorithm works in \NETIME.
\end{proof}

\begin{rem}\label{rem:generalization}
The presented reduction works in the same way for a more general class of data vectors of the form $\ktuple{\setSizeK}{\setD}\xrightarrow{} 
\ktuple{\dimension}{\Z}$. 
\end{rem}

\section{Conclusions and future work.}\label{sec:conclusions}

We have shown \ETIME\ and \PTIME\ upper\changed{\ complexity}{}{1372} bounds for \changed{}{the\ }{1372}$\N$- and $\Z$-solvability problems 
over sets (unordered tuples) of data, respectively,
by (1) reducing the former problem to the latter one (with nondeterministic exponential blowup),
(2) reformulating the latter problem in terms of (weighted) hypergraphs, 
and (3) solving the corresponding hypergraph problem by providing a local characterisation
testable in polynomial time.

The characterisation of $\Z$-solvability provided by Theorem~\ref{thm:core} identifies several simple to test properties that all together are equivalent to 
$\Z$-solvability. \changed{Each of properties}{Each of the properties}{1378} is independent of others and may be used as a partial test for 
\changed{non-reachability}{nonreachability}{107} in data nets. 
It is not clear if in every \changed{}{industrial}{1380} application we should use all of 
them.

The proposed characterisation does not work for directed structures, as all \changed{weight}{$\weight{non}$ morphisms}{1381} of the following nontrivial \changed{graphs}{graph}{1381} are $\vec{0}$.

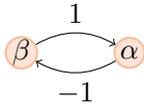
\begin{tikzpicture}
[place/.style={circle,draw=\mycolorOne!50,fill=\mycolorOne!20,thick,inner sep=0pt,minimum size=4mm}]
\node[place]      (A1)         {$\da$};
\node[place]      (B1)       [left=of A1]  {$\db$};
\draw [->] (B1) [bend left] to node [above] {$1$} (A1); 
\draw [->] (A1) [bend left]to node [below] {$-1$} (B1) ;
\end{tikzpicture}

\changed{Thus }{Thus, }{177}we should somehow refine the characterisation for data vectors going from ${\setSizeK}$-tuples to 
$\dimension$-dimensional vectors of \changed{integer 
numbers}{integers}{1387}, $\ktuple{\setSizeK}{\setD} \xrightarrow{} \ktuple{\dimension}{\Z}$. 
For example, for directed graphs the condition 
stated in Equation~\ref{eq:ass1} can be formulated as follows:
For every vertex $\da$ we define $\overrightarrow{\weight{da}}(\hypergraphH)=(out, in)$ where $out$ is the sum of all edges in $\hypergraphH$ \changed{outgoing from}{leaving}{1389} 
$\da$ 
and $in$ is the sum all 
edges 
in $\hypergraphH$ \changed{incoming to}{entering}{1390} $\da$. Now\changed{}{, instead of 
\begin{equation*}
{\weight{da}}(\hypergraphH) \in \sums \Z {\setof{{\weight{da'}}(\hypergraphH')}{\\ \hypergraphH' = ({\setVertices}', \mu') 
\in 
{\setInput}, \da' \in {\setVertices}'}} 
\end{equation*}
we have
}{NULL}
\begin{equation}
\label{eq:ass1_direc}
\overrightarrow{\weight{da}}(\hypergraphH) \in \sums \Z {\setof{\overrightarrow{\weight{da'}}(\hypergraphH')}{\\ \hypergraphH' = ({\setVertices}', \mu') 
\in 
{\setInput}, \da' \in {\setVertices}'}} 
\end{equation}

This is still not sufficient \changed{due to}{because of}{200} the following example with all weights equal \changed{}{to}{1393} $\vec{0}$:

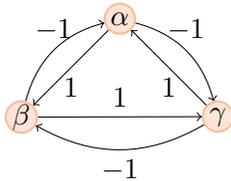
\begin{tikzpicture}
[place/.style={circle,draw=\mycolorOne!50,fill=\mycolorOne!20,thick,inner sep=0pt,minimum size=4mm}]
\node[place]      (A)                       {$\da$};
\node[place]      (B)       [below left=of A]  {$\db$};
\node[place]      (C)      [below right=of A] {$\dc$};
\draw [->] (B) [bend left] to node [above] {$-1$} (A); 
\draw [->] (A) [bend left]to node [above] {$-1$} (C);
\draw [->] (C) [bend left]to node [below] {$-1$} (B);

\draw [->] (A) to node [below] {$1$} (B);
\draw [->] (C)  to node [below] {$1$} (A); 
\draw [->] (B)  to node [above] {$1$} (C); 
\end{tikzpicture} 

To deal with this example we need to define another invariant, namely for every pair of \changed{nodes}{vertices}{NULL}\changed{ its weight}{\ $\da,\db$ the 
value of
$\weight{{da,db}}$}{1399} is a pair of \changed{numbers,}{numbers:}{1399}
sum of \changed{edges}{two edges between vertices $\da$ and $\db$,}{1399} and the absolute value of \changed{their difference}{the difference of the 
weights of the two edges}{1399}. Unfortunately, it is not clear how to generalise such invariants for hypergraphs.

Another approach to \changed{define weights}{a definition of $\weight{non}$}{1401-1404} for directed hypergraphs is to linearly order vertices and \changed{add 
only hyperedges that are oriented in the same way with 
respect 
to the order.}{ group edges depending on their orientation with respect to the order. Calculating $\weight{non}$ we add edges within the 
groups.}{1401-1404}\changed{}{\ For example, if the order is $\db<\da<\dc$ then $\weight{da}$ of the triangle above is equal to $[-2,2]$.}{1401-1404} 
\changed{However, in}{In}{1401-1404} this approach the weight depends on the orientation, which does not look as a desired solution.

Although we do not have the theorem which states that the $\Z$-solvability is equivalent to some 
lifted version of local characterisation\changed{}{,\ the local characterizations may be useful.}{1406} \changed{}{Observe that}{1408} any of 
\changed{}{the}{1408} 
described invariants \changed{ may be used as a partial test for the nonexistence of the solution. Indeed, they are}{is}{1408} homomorphisms from 
data vectors in $\ktuple{\setSizeK}{\setD} \xrightarrow{} \ktuple{\dimension}{\Z}$ to \changed{$\ktuple{\dimension}{\Z} \cdot n$}{$\ktuple{\dimension\cdot 
n}{\Z 
}$}{1409} for some $n\in \N$ depending on the 
homomorphism (for example in case of the $(out, in)$ weight the value $n=2$). \changed{Thus }{Thus, }{177}we identified a family of 
new heuristics that 
can be used in the 
analysis of data nets.

We also should mention that although in the proof of Theorem~\ref{thm:problem2} we use data vectors with \changed{the}{}{1413} image in $\ktuple{\dimension}{\Z}$ all proofs 
\changed{works}{work}{1414} for any Abelian \changed{group}{groups}{1414} in which 
the system of equations~\ref{eq:unique_solutions} has a unique solution. In the general case, it corresponds to the Abelian group in which
Lemma~\ref{lem:main_matrix}\changed{~in~Section~\ref{sec:matrix}}{}{1415} holds. \changed{It}{Lemma~\ref{lem:main_matrix}}{1415} speaks about maximality of the rank of some specific family of matrices.

We leave two interesting open problems. 
First, concerning the complexity of \changed{the}{}{1417} $\N$-solvability, we leave \changed{the gap}{a gap}{1418} between
the \NPTIME\ lower bound and our \NETIME\ upper bound.
Currently, available methods seem not suitable for closing this gap.
Second, the line of research we establish in this paper calls for continuation, in particular for
investigation of the solvability problems over the \emph{ordered} tuples of data.
\changed{Due to}{Because of}{200} Remark~\ref{rem:generalization} we know that the problematic bit is a generalisation of the Theorem~\ref{thm:problem2}
In our opinion, the solution for this case should \changed{be possible by further development of
the techniques}{be within reach by further developing techniques}{1423} proposed in this paper. 

\section*{Acknowledgment}
Authors thank Prof. Sławomir Lasota as well as many anonymous conference reviewers for suggesting improvements.

\bibliographystyle{alphaurl}
\bibliography{refer.bib}

\end{document}